\documentclass[11pt,letterpaper]{article}
\usepackage[margin=1.1in]{geometry}
\usepackage{amsmath,amssymb,amsfonts,setspace,url,enumitem}
\usepackage{amsthm}
\usepackage{tikz}
\usepackage{color}
\usepackage{bigstrut}
\usepackage{multirow}
\usepackage{ascmac}
\usepackage{mathpazo}
\usepackage{authblk}

\renewcommand{\arraystretch}{1.2}

\makeatletter
\newcommand{\newreptheorem}[2]{\newtheorem*{rep@#1}{\rep@title}\newenvironment{rep#1}[1]{\def\rep@title{#2 \ref*{##1}}\begin{rep@#1}}{\end{rep@#1}}}
\makeatother

\makeatletter
\newcommand{\newconjecture}[2]{\newconjecture*{rep@#1}{\rep@title}\newenvironment{rep#1}[1]{\def\rep@title{#2 \ref*{##1}}\begin{rep@#1}}{\end{rep@#1}}}
\makeatother
\usepackage{comment}
\usepackage{amsthm}
\usepackage{quantikz}
\usepackage{tikz}
\usepackage{mathabx}

\newcommand{\propref}[1]{\hyperref[#1]{Proposition~\ref{#1}}}
\newcommand{\corref}[1]{\hyperref[#1]{Corollary~\ref{#1}}}
\newcommand{\defref}[1]{\hyperref[#1]{Definition~\ref{#1}}}
\newcommand{\secref}[1]{\hyperref[#1]{Section~\ref{#1}}}
\newcommand{\appref}[1]{\hyperref[#1]{Appendix~\ref{#1}}}
\newcommand{\chapref}[1]{\hyperref[#1]{Chapter~\ref{#1}}}
\newcommand{\thmref}[1]{\hyperref[#1]{Theorem~\ref{#1}}}
\newcommand{\lemref}[1]{\hyperref[#1]{Lemma~\ref{#1}}}
\newcommand{\figref}[1]{\hyperref[#1]{Fig.~\ref{#1}}}
\renewcommand{\eqref}[1]{\hyperref[#1]{Eq. (\ref{#1})}}
\newcommand{\ineqref}[1]{\hyperref[#1]{Inequality (\ref{#1})}}
\newcommand{\tableref}[1]{\hyperref[#1]{Table (\ref{#1})}}
\newcommand{\algref}[1]{\hyperref[#1]{Algorithm (\ref{#1})}}
\renewcommand{\epsilon}{\varepsilon}

\newtheorem{innercustomgeneric}{\customgenericname}
\providecommand{\customgenericname}{}
\newcommand{\newcustomtheorem}[2]{%
  \newenvironment{#1}[1]
  {%
   \renewcommand\customgenericname{#2}%
   \renewcommand\theinnercustomgeneric{##1}%
   \innercustomgeneric
  }
  {\endinnercustomgeneric}
}

\newcustomtheorem{customthm}{Theorem}
\newcustomtheorem{customlemma}{Lemma}

\usepackage{booktabs}
\usepackage{algorithm}
\usepackage{amssymb}
\usepackage{amsmath}
\usepackage[colorlinks=true, urlcolor=blue, citecolor=blue]{hyperref}

\usepackage{xcolor}
\definecolor{mina}{rgb}{.2,.5,.1}

\definecolor{elham}{rgb}{.5,.1,.5}

\definecolor{armando}{rgb}{.2,.1,.5}
\newcommand{\anote}[1]{{\color{armando}$\big[\![$ \raisebox{.7pt}{Armando}\!\!\:\raisebox{-.7pt}{}: \textit{#1}$\ ]\!\big]$}}

\usepackage[numbers]{natbib}

\theoremstyle{plain}
\newtheorem{theorem}{Theorem}[section]
\newreptheorem{theorem}{Theorem}
\newreptheorem{conjecture}{Conjecture}
\newtheorem{lemma}{Lemma}[section]

\newtheorem{cor}{Corollary}[section]

\newtheorem{definition}{Definition}[section]

\newtheorem{prop}{Proposition}[section]

\theoremstyle{definition}

\newcommand{\stat}[1]{\vert#1\vert_{\mathrm{tv}}} 

\newcommand{\Tr}{\mathrm{Tr}} 

\newcommand{\tr}{\mathrm{tr}}

\begin{document}

\title{A unifying framework for differentially private quantum algorithms}
\author[,1]{Armando Angrisani \thanks{corresponding author: armando.angrisani@lip6.fr}}
\author[2]{Mina Doosti}
\author[1,2]{Elham Kashefi}

\affil[1]{LIP6, CNRS, Sorbonne Université, 75005 Paris, France}
\affil[2]{School of Informatics, University of Edinburgh, EH8 9AB Edinburgh, United Kingdom}

\date{\today}
\maketitle
\thispagestyle{empty}
\setcounter{page}{1}
\begin{abstract}

Differential privacy is a widely used notion of security that enables the processing of sensitive information. In short, differentially private algorithms map ``neighbouring'' inputs to close output distributions. Prior work proposed several quantum extensions of differential privacy, each of them built on substantially different notions of neighbouring quantum states. In this paper, we propose a novel and general definition of neighbouring quantum states. We demonstrate that this definition captures the underlying structure of quantum encodings and can be used to provide exponentially tighter privacy guarantees for quantum measurements.
Our approach combines the addition of classical and quantum noise and is motivated by the noisy nature of near-term quantum devices.
Moreover, we also investigate an alternative setting where we are provided with multiple copies of the input state. In this case, differential privacy can be ensured with little loss in accuracy combining concentration of measure and noise-adding mechanisms. 
\emph{En route}, we prove the advanced joint convexity of the quantum hockey-stick divergence and we demonstrate how this result can be applied to quantum differential privacy.
Finally, we complement our theoretical findings with an empirical estimation of the certified adversarial robustness ensured by differentially private measurements.

\end{abstract}
\newpage
\section{Introduction}\label{sec:introduction}

In recent years, the availability of large datasets and advanced computational tools has sparked progress across various fields, including natural sciences, medicine, finance, and social sciences.  This advance came also with privacy concerns since even the release of aggregated data can compromise the sensitive information contained in the original dataset. This poses a significant challenge for the researcher, who must adopt privacy-preserving techniques to avoid the exposure of private data.
Privacy-preserving data processing is a non-trivial task, and ill-defined notions of privacy led to impressive privacy breaches \cite{narayanan2007break}. This motivated the quest for a robust framework to assess privacy.

Over the last decade, differential privacy (DP) has become the de facto standard for ensuring privacy both in statistical data analysis and machine learning applications \cite{dwork1, dwork2, cummings2023challenges, JMLR:v12:chaudhuri11a, abadi2016, papernot2017semisupervised, bassily2018}.
Intuitively, a differentially private algorithm $\mathcal{A}(\cdot)$ can learn a statistical property of a dataset consisting of $n$ elements, yet it leaks \emph{almost} nothing about each individual element. In other words, given two inputs $x$ and $x'$ which are very close according to some chosen metric, the output distributions $\mathcal{A}(x)$ and $\mathcal{A}(x')$ should be almost indistinguishable. We call $x$ and $x'$ neighbouring inputs. If $x$ and $x'$ represent datatsets about $n$ individuals, then it's customary to consider $x$ and $x'$ neighbouring if one of such individuals is present in $x$ and absent in $x'$. Then, if $\mathcal{A}(\cdot)$ is differentially private, the output alone doesn't allow for inferring whether the input contained a given individual.
This goal is pursued by combining various techniques, that usually involve randomising the input or perturbing the output by adding noise.
The challenge is then to achieve the desired level of privacy by adding less noise as possible, hence preserving accuracy.

Apart from privacy-preserving data analysis and machine learning, differential privacy has also found several applications in other fields of computer science such as statistical learning theory \cite{equiv, JMLR:v17:15-313, online, arunachalam2021private}, adaptive data analysis \cite{dwork2015preserving, bassily2021algorithmic,feldman2017generalization} and mechanism design \cite{design}.

More recently, the major influence of quantum computing and quantum information has led to the exploration of differentially private quantum algorithms.
Since many near-term quantum algorithms involve a classical optimiser as a subroutine, one possible approach consists in privatising such optimiser and leaving the rest of the algorithm unchanged. This strategy is adopted in \cite{senekane2017privacy,Li_2021,Du_2022,watkins2023quantum}. 

Alternatively, we can rely on several notions of \emph{quantum} differential privacy. Quantum differential privacy allows the design of private measurements and channels combining classical and quantum noise.
This is extremely relevant with the emergence of Noisy Intermediate Scale Quantum devices (NISQ) today \cite{preskill_quantum_2018}. The noisy nature of these devices on the one hand, and the potential capabilities of quantum algorithms, on the other hand, make such quantum or hybrid quantum-classical mechanisms, an interesting subject of study from the point of view of privacy.  
Several efforts has been made in this area of research, including \cite{quantumDP, aaronson2019gentle, franca, farokhi2023privacy, nuradha2023quantum}.
Furthermore, the connection between machine learning and differential privacy \cite{feldman2017generalization,lecuyer2019certified} suggests that exploring quantum differential privacy can lead to intriguing insights into the capabilities of quantum machine learning.

One of the main challenges in translating the definition of DP in the quantum setting is to characterise the notion of neighbouring quantum states, i.e. choose the right metric to measure the similarity between the input states.
The first notion of quantum differential privacy was proposed in \cite{quantumDP} and it's based on bounded trace distance, whereas the definition introduced in \cite{aaronson2019gentle} is based on reachability by a single-qudit operation. 
Another possible definition is based on the quantum Wasserstein distance of order 1. This metric was introduced in \cite{wasserstein2021} and the authors mention quantum differential privacy as one potential application of their work. 
Furthermore, quantum private PAC learning has been defined in \cite{arunachalam2020quantum} and a quantum analogue of the equivalence between private classification and online prediction has been shown in \cite{arunachalam2021private}. Moreover, an equivalence between learning with quantum local differential privacy and quantum statistical query (QSQ) learning was provided in \cite{Angrisani2022QuantumLD}. 
Other authors compared classical and quantum mechanisms in the context of {local} differential privacy \cite{yoshida2020classical, yoshida2021mathematical}. 
Building upon these prior contributions, the present paper aims at establishing a general framework for differentially private quantum algorithms, providing a more general definition of neighbouring quantum states and attaining better privacy guarantees combining classical and quantum noisy channels.

\subsection{Motivation: connecting neighbouring relationships with quantum encodings}
Our work is motivated by a practical goal: we want to design quantum algorithms that satisfy differential privacy with respect to a classical input $x$. We assume that this input belongs to a set equipped with a neighbouring relationship.
Moreover, we consider quantum algorithms that include a quantum encoding as a subroutine, where $x$ is mapped to a quantum state $\rho(x)$.
Thus, we want to define a quantum neighbouring relationship that mimics the underlying classical neighbouring relationship. In particular, we require the following property:
\[
\text{$x$ and $x'$ are neighbouring }\implies \text{ $\rho(x)$ and $\rho(x')$ are neighbouring}.
\]
It's easy to see why the above property is extremely useful. If an algorithm $\mathcal{A}$ is $(\epsilon,\delta)$-differentially private with respect to $\rho(x)$, then $\mathcal{A}\circ \rho$ is $(\epsilon,\delta)$-differentially private with respect to $x$ (this is stated more formally in \propref{prop:transferring}). 
In the meantime, we want to avoid neighbouring relationships that are excessively loose, as this would make the output almost independent of the input. A paradigmatic example of a ``pathological'' relationship is the one based on \emph{constant} trace distance:
\[
\text{ $\rho$ and $\rho'$ are neighbouring} \iff \frac{1}{2}\|\rho-\rho'\|_1 \leq \tau = \Theta(1).
\]
To fix the ideas, let $\tau = 0.1$.
It's easy to see that for any pair of states $\rho,\sigma$ we can build a sequence $\rho_0,\rho_2,\dots,\rho_{10}$, such that
\begin{align*}
    \begin{cases}
      \rho_0=\rho\\
      \rho_{10}=\sigma
    \end{cases}  
    \;\; \text{ and for all $i$, }\; \rho_i\sim\rho_{i+1}.
\end{align*} 
 By triangle inequality, the outputs of $\rho$ and $\sigma$ will be $(10 \epsilon)$-close. This, significantly limits the capabilities of private algorithms, since, independently of the input states, the output distribution would be highly concentrated around the same value. We discuss this more formally in \secref{sec:cost}.

Surprisingly, the current quantum neighbouring relationships fulfil these two natural requirements only for a limited number of specific quantum encodings. Given that this property is crucial for the framework of differential privacy and the need to handle various types of quantum and classical data in the quantum setting, it becomes necessary to develop an approach that can account for different encodings. Therefore, we present a generalised neighbouring relationship that enables the handling of a wide range of near-term and long-term algorithms.

\subsection{Overview of main results}
In this work, we tackle several technical problems arising in the field of quantum differential privacy and we try to address them within a broader framework using different tools and techniques from quantum information. The following is a summary of our main contributions.

\begin{itemize}
    \item \textbf{Improved privacy bounds for noisy channels.} Our first contribution consists of tighter privacy guarantees for a general family of noisy channels, which includes local Pauli noise and particularly as a special case, the depolarising channel. To this end, we prove the advanced joint convexity of the quantum hockey-stick divergence.
    Moreover, we provide a tighter analysis of the privacy of quantum measurements post-processed with classical stochastic channels, such as the Laplace or Gaussian noise. This approach allows us to be able to study both classical and quantum noisy mechanisms for differential privacy, within a unified framework.
    \item \textbf{Generalised neighbouring relationship.} Our second contribution is a generalised neighbouring relationship, that allows us to recover the previous definition as special cases. We demonstrate how to design differentially private measurements according to this definition by introducing both classical and quantum noise into the computation. Notably, we show that local measurements can be made differentially private by adding a modest amount of noise. Our work is the first to incorporate the locality in the analysis of quantum differential privacy.
    \item \textbf{Privacy-utility tradeoff for quantum differential privacy.} There exists an unavoidable tradeoff between the desired level of privacy and the resulting loss in accuracy. Here, we make a crucial observation: different neighbouring relationships have different tradeoffs. In particular, this limits the applicability of neighbouring relationships based solely on the bounded trace distance. We also show no-go results for pure quantum differential privacy under the Wasserstein distance of order 1.
    \item \textbf{Private estimation with multiple copies.} We provide differentially private mechanisms for estimating the expected values of observables given $m$ copies of a quantum state. These mechanisms can find applications in privatising the results of experiments on physical devices where estimating the expectation value is the main figure of merit.
    \item \textbf{Applications.} Our results can be applied to variational quantum algorithms and other quantum machine learning models to enhance or certify privacy. We specifically focus on certified adversarial robustness through differential privacy and we perform numerical simulations to assess the robustness to adversarial attacks of private quantum classifiers.
\end{itemize}

  {\renewcommand{\arraystretch}{1.4} 
\begin{figure*}[!ht]
\centering
\begin{tikzpicture}
[->,>=stealth,shorten >=1pt,auto,  thick,yscale=0.8,
main node/.style={circle,draw}, node distance = 1cm and 1.8cm,
block/.style   ={rectangle, draw, text width= 2.5cm, text centered, rounded corners, minimum height=1cm, fill=white, align=center, font={\footnotesize}, inner sep=5pt}]

    \node[main node,block] (obs) at (1,3) {\textsf{Private estimation of observables (\secref{sec:multiple})}};
    \node[main node,block] (npe) at (5,3){\textsf{neighbouring-preserving encodings}};
    \node[main node,block] (div) at (-2.5,-2) {\textsf{New results on quantum divergences (\appref{sec:div})}};
   \node[main node,block] (qdp) at (3,0.5) {\textsf{Quantum DP}};
   \node[main node,block] (qnr) at  (5,-2) {\textsf{Quantum neighbouring relationship}};
    \node[main node,block] (w1) at (8,-5) {\textsf{Bounded $W_1$ distance \\ (\secref{sec:cost})}};
    \node[main node,block] (new) at (5,-5.5) {\textsf{$(\Xi,\tau)$-neighbouring states \\(\secref{sec:neighbour})}};
  \node[main node,block] (trace) at (1,-2) {\textsf{Tighter privacy under bounded trace distance (\secref{sec:trace})}};
    \node[main node,block] (inspired) at (-2.5,-5) {\textsf{Private quantum-inspired sampling (\secref{sec:inspired},\\\appref{app:inspired})}};
    \node[main node,block] (gen_privacy) at (1,-5.5) {\textsf{Private single-copy \\ measurements for $(\Xi,\tau)$-neighbouring states \\(\secref{sec:local})}};
    
    \node[main node,block] (dp) at  (8,1) {\textsf{Classical DP}};
    \node[main node,block] (adv) at (8,-1) {\textsf{Certified adversarial robustness (\secref{sec:rob})}};

    \path [->](div) edge node {} (trace);
    \path [->](npe) edge node {} (qnr);
    \path [->](npe) edge node {} (obs);
    \path [->](qnr) edge node {} (npe);
    \path [->](qdp) edge node {} (trace);
    \path [->](qdp) edge node {} (obs);
    \path [->](qdp) edge node {} (qnr);
    \path [->](trace) edge node {} (gen_privacy);
    \path [->](new) edge node {} (gen_privacy);
    \path [->](trace) edge node {} (inspired);
    \path [->](qdp) edge node {} (npe);
    \path [->](npe) edge node {} (dp);
    \path [->](dp) edge node {} (adv);
    \path [dotted,->](qnr) edge node {\small{no-go}} (w1);
    \path [->](qnr) edge node {} (new);

\end{tikzpicture}
    \caption{High-level summary of the main concepts covered in the paper and their mutual dependencies. Our contributions spans includes tighter information-theoretic bounds, novel approaches to quantum differential privacy and applications to private quantum machine learning. We remark that the no-go results for $W_1$ distance holds for single-copy measurements under pure differential privacy with respect to mixed states with bounded $W_1$ distance. On the other hand, the $W_1$ distance can be conveniently used for the private estimation of expected values of observables with multiple copies.
    }
    \label{fig:Dependencies}
\end{figure*}
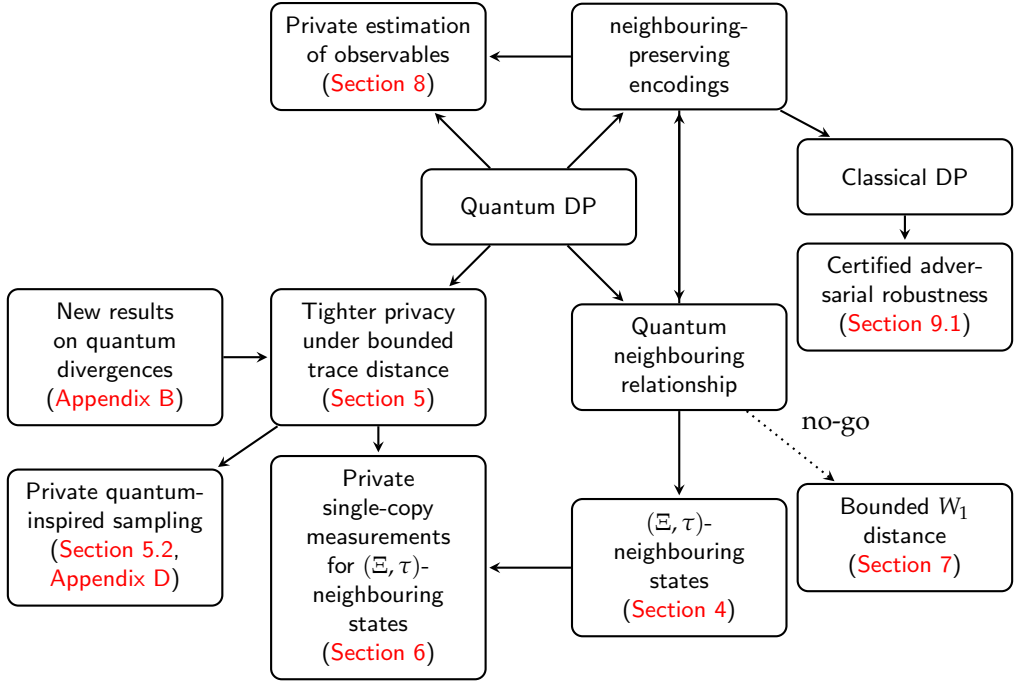}

\subsection{Organisation of the paper}
The paper is organised as follows. First, in \secref{sec:back} we give a brief background on the notion and definition of differential privacy in the classical world, as well as introduce the notations we are using throughout the paper. In \secref{sec:qdp}, we discuss quantum differential privacy and the differences among different approaches to defining differential privacy in the quantum world, specifically the neighbouring relationship between quantum states, and we prove that classical differential privacy can be ensured through quantum differential privacy. In \secref{sec:neighbour} we introduce our generalised framework for quantum differential privacy and we discuss its properties. Within this formal framework, we prove several results. Starting with \secref{sec:trace}, we provide several improved privacy bounds for the case where the neighbouring relationship is specified with a bounded trace distance between quantum states. Our results include classical post-processing and quantum-inspired sampling mechanisms. Our techniques hinge on a novel result information-theoretic result, namely the advanced joint-convexity of the quantum hockey-stick divergences, discussed in detail in \appref{sec:div}, along with a proof of the quantum Bretagnolle-Huber inequality. We then turn to the unique properties of our framework in \secref{sec:local} which allows us to study local measurements as quantum differentially private mechanisms, as well as addressing the question of how quantum and classical noise can be studied together in the context of differential privacy. In \secref{sec:cost} we define the cost of differential privacy and benchmark different approaches and notions of neighbouring, under this lens, providing negative and positive results which clarify and justify the applicability of our framework. In \secref{sec:multiple} we introduce mechanisms for privately estimating expectation values. Finally, in \secref{sec:applications}, we discuss applications of some of our results in quantum machine learning, particularly for certified adversarial robustness, and we support our theoretical findings with numerical simulations.

\section{Background}
\label{sec:back}
We start by introducing the notations we use in the paper as well as essential definitions of classical differential privacy. Other technical tools such as different norms and divergences are introduced in \appref{app:pre}.

\subsection{Notation}

We denote by $\log(\cdot)$ the natural logarithm.
We denote by $P(X)$ the power set of a set $X$, i.e. the set of all subsets of $X$. For a vector $\boldsymbol{x}=(x_1,\dots,x_n)$, we denote as $\|\boldsymbol{x}\|_p$ its $p$-norm, where $\|\boldsymbol{x}\|_p=(\sum_{i=1}^n|x_i|^p)^{1/p}$ for $1\leq p < \infty$ and $\|\boldsymbol{x}\|_\infty= \max_i \|x_i\|$.
It's convenient to introduce also the $0$-norm (which is technically not a norm): $\|\boldsymbol{x}\|_0 = |\{i : x_i \neq 0\}|$, which is the number of the non-zero entries of $\boldsymbol{x}$.

We consider a set $V$ corresponding to a system of $|V| =
n$ qudits, and denote by $\mathcal{H}_n = \bigotimes_{v\in V} \mathbb{C}^d$ the Hilbert space
of $n$ qudits. We denote by $\mathcal{L}(\mathcal{H}_n)$ the set of linear operators on $\mathcal{H}_n$. We denote by $\mathcal{O}_n$ the set of self-adjoint linear operators on $\mathcal{H}_n$. By $\mathcal{O}^+_n$
 we denote the subset of positive semidefinite linear operators on $\mathcal{H}_n$, and
$\mathcal{S}_n \subset \mathcal{O}^+_n$ denotes the set of quantum states. Similarly,
we denote by $\mathcal{P}_n$ the set of probability measures on $[d]^V$.
For any subset $A \subseteq V$, we use the standard notation $\mathcal{O}_A, \mathcal{S}_A, \dots$ for
the corresponding objects defined on subsystem $A$. Given
a state $\rho \in \mathcal{S}_n$, we denote by $\rho_A$ its marginal on subsystem $A$. For any $X \in \mathcal{O}_n$, we denote by $\|X\|_p$ its Schatten $p$ norm. For any subset $A\subseteq V$, the identity on $\mathcal{O}_A$ is denoted by $\mathbb{1}_A$, or more simply $\mathbb{1}$. Given an observable $O$, we define $\langle O\rangle_\sigma = \Tr[\sigma O]$. 
Moreover, given a number $a \in \mathbb{R}$, we define $\{O \geq a\}$ to be the projector onto the subspace spanned by the eigenvectors of $O$ corresponding to eigenvalues greater than or equal to $a$. We denote the probability of measuring an eigenvalue of $O$ greater than $a \in \mathbb{R}$ in state $\sigma$ as $\Pr_\sigma (O \geq a) := \Tr[\sigma\{O \geq a\}]$.
For a subset $F$ of the spectrum of $O$, we denote the probability of that the measurement outcome lies in $F$ as $\Pr_\sigma[O \in F]$.
For an observable $O$, we write its Pauli expansion as $O=\sum_{P\in\{X,Y,Z,\mathbb{1}\}^n}c_P P$. We say that a Pauli string $P=P_1P_2\dots P_n\in\{X,Y,Z,\mathbb{1}\}^n$ acts non trivially on $\mathcal{I}\subseteq [n]$ if $c_P \neq 0$ and $\exists i \in \mathcal{I}: P_i \neq \mathbb{1}$.
A quantum channel $\mathcal{N} : V \rightarrow{} W$ is a linear completely
positive and trace-preserving map from the operators on
$\mathcal{H}_{|V|}$ to the operators on $\mathcal{H}_{|W|}$. Similarly, a classical channel can be defined as a randomised mapping from $V$ to $W$.
We'll refer to a channel as an algorithm when we want to emphasise the input-output relationship.
For a channel ${\Phi}$, either classical or quantum, we denote as $\mathsf{range}({\Phi})$ the set of all the possible outputs of ${\Phi}$.
Given a quantum channel $\Phi$ acting on $n$ qubits, we define its light-cone as follows: first, for any qubit $i$,
we denote by $\mathcal{I}_i$ the minimal subset of qubits such that $\Tr_{\mathcal{I}_i} \Phi(\rho) = \Tr_{\mathcal{I}_i}\Phi(\sigma)$ for any two $n$-qubit states $\rho$ and $\sigma$ such that $\Tr_i(\rho) = \Tr_i(\sigma)$. Then, the light-cone of ${\Phi}$ is defined as $|\mathcal{I}|:=\max_{i\in [k]}|\mathcal{I}_i|$.

\subsection{Classical differential privacy}
\label{sec:dp}
We concisely introduce the definition of differential privacy. For a comprehensive introduction to the topic, we refer to \cite{dwork2}, \cite{complexity} and \cite{cummings2023challenges}.
Throughout this paper, we'll denote by $\sim$ the neighbouring condition, i.e. a relationship between two inputs, consisting of either classical vectors or quantum states. We'll write $\overset{Q}{\sim}$ when we want to emphasise that the neighbouring relationship refers to quantum states.  The choice of the relationship is problem-dependent. In many practical cases, it's convenient to say that two binary vectors $x,x'\in\{0,1\}^n$ are neighbouring if their Hamming distance is at most one, i.e. 
\[x\sim x' \iff d_H(x,x')\leq 1.\]
In alternative, we can select a $p$-norm and a threshold $\gamma \geq 0$ and opt for the following neighbouring relationship:
\[x\sim x' \iff \|x-x'\|_p\leq \gamma.\]
We say that a randomised algorithm $\mathcal{A}(\cdot)$ is $(\epsilon,\delta)$-differentially private (DP) if for all $x \sim x'$ and for all $S \subseteq \mathsf{range}(\mathcal{A})$, it satisfies
\[
\Pr[\mathcal{A}(x) \in S] \leq e^\epsilon \Pr[\mathcal{A}(x') \in S] +\delta.
\]
We say that $\mathcal{A}(\cdot)$ is $\epsilon$-DP when it is $(\epsilon,0)$-DP.
Equivalently, differential privacy can be defined in terms of hockey-stick divergence $E_\gamma$ and the smooth max-relative entropy (or smooth max-divergence) $D_\infty^\delta$:
\[
\mathcal{A}\text{ is $(\epsilon,\delta)$-DP }\iff \forall x\sim x': E_{e^\epsilon}(\mathcal{A}(x)\|\mathcal{A}(x'))\leq \delta \iff \forall x\sim x': D_\infty^\delta(\mathcal{A}(x)\|\mathcal{A}(x'))\leq \epsilon,
\]
where the (classical) hockey-stick divergence $E_\gamma$ between two distributions $P$ and $Q$ is defined as follows \cite{hs-div}:
\[
E_\gamma(P\|Q) := \frac{1}{2} \int |\text{d} P - \gamma \text{d}Q| -\frac{1}{2}(\gamma - 1), 
\]
for $\gamma\geq 1$.
These information-theoretic divergences can be thought of as a measure of closeness between distributions, thus these reformulations are consistent with the intuition that private algorithms map neighbouring inputs to ``close'' output distributions.
Differential privacy with $\delta=0$ is also referred to as \emph{pure} differential privacy, whereas the case with $\delta\neq0$ is referred to as \emph{approximate} differential privacy. Roughly speaking, an $(\epsilon,\delta)$-DP algorithm can be thought of as an algorithm that is $\epsilon$-DP with probability $1-\delta$. We remark that this intuition is slightly imprecise, and thus we refer to the following references for a more detailed explanation \cite{bun2016concentrated, meiser2018approximate, complexity}. 

It's also worth noticing that the max-divergence corresponds to the Rényi divergence of order $\infty$.
Thus, it's possible to relax pure differential privacy by replacing the max-divergence with the Rényi divergence of order $\alpha$, for $\alpha\geq 1$ \cite{Mironov_2017}. We say that $\mathcal{A}$ is $(\alpha,\epsilon)$-RDP (Rényi differentially private) if for all $x \sim x'$,
\[
D_\alpha(\mathcal{A}(x)\|\mathcal{A}(x'))\leq \epsilon.
\]
As a consequence, for all $S \subseteq \mathsf{range}(\mathcal{A})$, we have
\[
\Pr[\mathcal{A}(x) \in S] \leq e^{\epsilon} \Pr[\mathcal{A}(x') \in S]^{(\alpha -1)/\alpha}.
\]
If $\mathcal{A}$ is $(\alpha,\epsilon)$-RDP then it is also $(\epsilon + \frac{\log (1/\delta)}{\alpha - 1},\delta)$-DP for any $0<\delta<1$.
Similarly, if $\mathcal{A}$ is $(\epsilon,0)$-DP then it is also $(\alpha,2\alpha\epsilon^2)$-RDP for any $\alpha \geq 1$.

\subsubsection*{Privacy via classical noisy channels}
Now we present two widely used mechanisms that ensure differential privacy by injecting noise into the output.
To this end, we introduce two classical channels $\Lambda_{\mathcal{L},b}:\mathbb{R}\rightarrow \mathbb{R}$ and $\Lambda_{\mathcal{G},\sigma}:\mathbb{R}\rightarrow \mathbb{R}$, that corresponds to an additive noise coming from either the Laplace distribution of scale $b$ or the Gaussian distribution of variance $\sigma^2$, both centred in zero.
The channels are defined as follows:
\[
\Lambda_{\mathcal{L},b}(x) = x + \eta \text{\;\; and \;\;} \Lambda_{\mathcal{G},\sigma}(x) = x + \zeta,
\]
where $\eta \sim \frac{1}{2b} \exp \left(- \frac{|\eta|}{b}\right)$ and
 $\zeta \sim \frac{1}{\sigma \sqrt{2\pi}} \exp \left(- \frac{\zeta^2}{2\sigma^2}\right)$. 
 
Let $f:\mathcal{X}\rightarrow \mathbb{R}$ be a scalar function. We define the \emph{sensitivity} of $f$ as
\begin{equation}
\label{eq:classical-sens}
    \Delta_f := \max_{\substack{x,x'\in \mathcal{X}\\x\sim x'}}|f(x) - f(x')|.
\end{equation}
Then $\Lambda_{\mathcal{L},b}(f(\cdot))$ is $\epsilon$-DP if $b\geq \Delta/\epsilon$. Similarly, $\Lambda_{\mathcal{G},\sigma}(f(\cdot))$ is $(\epsilon,\delta)$-DP if $\sigma^2 \geq 2 \ln(1.25/\delta) \Delta^2/\epsilon^2$. The addition of Laplace noise is referred to as \emph{Laplace mechanism} \cite{dwork1}, whereas the addition of Gaussian noise is referred to as \emph{Gaussian mechanism} \cite{dwork2}.
Both mechanisms can also be analysed within the relaxed framework of Rényi differential privacy \cite{Mironov_2017}. 

\section{Quantum differential privacy}\label{sec:qdp}
Let $\rho, \sigma$ two neighbouring quantum states, i.e. $\rho \overset{Q}{\sim} \sigma$.  We'll discuss appropriate neighbouring conditions for quantum states in the next sections and for the moment we use the letter $Q$ as a placeholder. We also say that $\rho$ and $\sigma$ are $Q$-neighbouring in order to emphasise that we selected a suitable relationship $Q$ over quantum states. Following \cite{quantumDP,franca}, we say that a quantum channel $\mathcal{C}(\cdot)$ is $(\epsilon,\delta)$-DP if for all $\rho\overset{Q}{\sim}\sigma$, for all POVM $M=\{M_m\}$ and for all $m$, we have that
\[
\Tr[M_m\mathcal{C}(\rho)]\leq e^\epsilon \Tr[M_m\mathcal{C}(\sigma)] + \delta.
\]
As in the classical case, this can be equivalently expressed in terms of the quantum hockey-stick divergence or the quantum smooth max-relative entropy:
\begin{align*}
    \mathcal{C}\text{ is $(\epsilon,\delta)$-DP }\iff \forall \rho,\sigma: \rho\overset{Q}{\sim}\sigma: E_{e^\epsilon}(\mathcal{C}(\rho)\|\mathcal{C}(\sigma))\leq \delta 
    \\ \iff \forall \rho,\sigma: \rho\overset{Q}{\sim}\sigma : D_\infty^\delta(\mathcal{C}(\rho)\|\mathcal{C}(\sigma))\leq \epsilon,
\end{align*}
where the quantum hockey-stick divergence $E_\gamma$ is defined as follows:
\[
E_\gamma(\rho\|\sigma) := \Tr(\rho-\gamma\sigma)^+, 
\]
for $\gamma\geq 1$. Here $X^+$ denotes the positive part of the eigendecomposition of a Hermitian matrix $X = X^+ - X^-$. 
We refer to Lemma III.2 in \cite{franca} for more details. 
A special case of particular interest is one of quantum-to-classical channels (i.e. POVM measurements), mapping states to probability distributions. For a measurement $\mathcal{M}$, denote as $\mathcal{M}(\rho)$ the probability distribution induced by measuring $\mathcal{M}$ on input $\rho$.
Quantum differential privacy shares many useful properties with classical differential privacy. Notably, it is robust to parallel composition and post-processing (also referred to as sequential composition).

\begin{prop}[Adapted from Corollary III.3, \cite{franca}]
The following properties hold.
\begin{itemize}
    \item (Post-processing) Let $\mathcal{A}$ be $(\epsilon, \delta)$-differentially private and $\mathcal{N}$ be an arbitrary quantum channel, then
$\mathcal{N} \circ \mathcal{A}$ is also $(\epsilon, \delta)$-differentially private.
    \item (Parallel composition) Let $\mathcal{A}_1$ be $(\epsilon_1, \delta_1)$-differentially private and $\mathcal{A}_2$ be $(\epsilon_2, \delta)$-differentially
private. Define that $\rho_1 \otimes \rho_2 \overset{Q}{\sim} \sigma_1 \otimes \sigma_2$ if $\rho_1 \overset{Q}{\sim} \sigma_1$
and $\rho_2\overset{Q}{\sim}\sigma_2$. Then $\mathcal{A}_1 \otimes \mathcal{A}_2$ is $(\epsilon_1 + \epsilon_2, \overline{\delta})$-differentially private on such product states, with
$\overline{\delta} = \min\{\delta_1 + e^{\epsilon_1} \delta_1, e^{\epsilon_2}\delta_1 + \delta_2\}$.

Moreover, if $\mathcal{A}_1$ and $\mathcal{A}_2$ are quantum-classical channels (measurements), we have that $\mathcal{A}_1 \otimes \mathcal{A}_2$ is $(\epsilon_1 + \epsilon_2, \delta_1+\delta_2)$-differentially private.
\end{itemize}
\end{prop}
\begin{proof}
The proposition coincides with Corollary III.3 in \cite{franca}, except for the final statement about the parallel composition of differentially private measurements. Since the output of a measurement is a classical distribution, the proof of this part is identical to the one of Theorem 3.16 in \cite{dwork2}.
\end{proof}

In short, the composition theorem ensures that performing $k$ times an $\epsilon$-DP algorithm is $(\epsilon k)$-differentially private, and then the privacy budget scales as the number of repetitions $k$. However, under mild assumptions, this scaling can be improved to $O(\sqrt{k})$. This result is called \emph{advanced composition} (we refer to Theorem 3.20 in \cite{dwork2} for the classical case).
Moreover, advanced composition holds also for quantum measurements under suitable assumptions (Theorem 6, \cite{quantumDP}).

Rényi quantum differential privacy has also been defined in \cite{franca}. Due to the non-commutative nature of quantum mechanics, the quantum generalisation of the Rényi divergence is not unique. However, we don't need to fix a particular definition of the quantum Rényi divergence, since we can define Rényi quantum differential privacy in terms of an arbitrary family of Rényi divergences $\mathbb{D}_\alpha$, as defined in \cite{tomamichel2015quantum}.
Thus, a quantum channel $\mathcal{C}$ is $(\epsilon,\alpha)$-Rényi differentially private if 
\[
\sup_{\rho\sim\sigma} \mathbb{D}_\alpha(\mathcal{C}(\rho)\|\mathcal{C}(\sigma)) \leq \epsilon.
\]

\subsection{From quantum to classical differential privacy}
Now we show how quantum differential privacy can be used as a proxy to ensure the privacy of a classical input encoded in a quantum state.
First, we introduce a preliminary definition.

\begin{definition}[Privacy-preserving quantum encodings]
Let $\mathcal{X}$ a set equipped with a neighbouring relationship $\sim$. A quantum encoding $\rho(\cdot)$ is $Q$-neighbouring-preserving if 
\[
x\sim x' \implies \rho(x)\overset{Q}{\sim} \rho(x').
\]
\end{definition}
The following proposition formalizes the intuitive fact that $Q$-neighbouring-preserving encodings can be used to transfer privacy guarantees and ensure the privacy of the underlying classical input.
\begin{prop}[Transferring privacy guarantees]
\label{prop:transferring}
Let $\rho(\cdot)$ a quantum encoding, i.e. a function mapping a classical vector $x\in\mathcal{X}$ to a quantum state $\rho(x)$. Assume $\mathcal{X}$ is equipped with a neighbouring relationship $\mathcal{\sim}$ and $\mathcal{S}_n$ is equipped with a neighbouring relationship $\overset{Q}{\sim}$.
Assume that $\rho(\cdot)$ is $Q$-neighbouring-preserving.
Let $\mathcal{M}$ be a measurement. We have,
\[
\text{$\mathcal{M}$ is $(\epsilon,\delta)$-DP with respect to $\overset{Q}{\sim}$ }\implies \text{$\mathcal{M}(\rho(\cdot))$ is $(\epsilon,\delta)$-DP with respect to ${\sim}$.}
\]
\end{prop}
\begin{proof}
The proposition follows from the definition of differential privacy. Assuming $\mathcal{M}(\cdot)$ is $(\epsilon,\delta)$-DP, we have
\[
\forall \sigma,\sigma': \sigma\overset{Q}{\sim}\sigma',\;\forall S \subseteq \mathsf{range}(\mathcal{M}):\Pr[\mathcal{M}(\sigma)\in S]\leq e^\epsilon \Pr[\mathcal{M}(\sigma')\in S] +\delta.
\]
Since $\rho(\cdot)$ is $Q$-neighbouring-preserving, the above inequality still holds if we set $\sigma:=\rho(x)$ and $\sigma':=\rho(x')$ for $x\sim x'$. Moreover, we replace $\mathsf{range}(\mathcal{M})$ with $\mathsf{range}(\mathcal{M}\circ \rho(\cdot))$ (we can do it since $\mathsf{range}(\mathcal{M}\circ \rho(\cdot))$ is a subset of $\mathsf{range}(\mathcal{M})$). The result readily follows.
\[
\forall x,x': x{\sim}x',\;\forall S \subseteq \mathsf{range}(\mathcal{M}\circ \rho(\cdot)):\Pr[\mathcal{M}(\rho(x))\in S]\leq e^\epsilon \Pr[\mathcal{M}(\rho(x'))\in S] +\delta.
\]
\end{proof}
\section{Generalised neighbouring relationship}
\label{sec:neighbour}
In this section, we present the cornerstone of our work, which is a general definition of neighbouring quantum states.
\begin{definition}
Let $\rho,\sigma \in \mathcal{S}_n$ and let $\Xi \subset P([n])$, i.e. let $\Xi$ be a collection of subsets of $[n]$. Let $\tau > 0$ be a parameter. We say that $\rho$ and $\sigma$ are $(\Xi,\tau)$-neighbouring and we write $\rho \overset{(\Xi,\tau)}{\sim} \sigma$ if 
\[
\exists \mathcal{I} \in \Xi: \Tr_{\mathcal{I}}\rho = \Tr_{\mathcal{I}}\sigma  \; \wedge \;\frac{1}{2}\|\rho - \sigma\|_1 \leq \tau. 
\]
If $\Xi = \{\mathcal{I} : \mathcal{I}=\{i,i+1,\dots,i+\ell\} \text{ for some }i\}$, i.e. each subset $\mathcal{I}$ is a collection of $\ell$ consecutive integers (modulo $n$), we say that $\rho$ and $\sigma$ are $(\ell,\tau)$-neighbouring and we write $\rho\overset{(\ell,\tau)}{\sim}\sigma$ .
When $\Xi=\{[n]\}$, we simply write $\rho \overset{\tau}{\sim} \sigma$ and we say that $\rho$ and $\sigma$ are $\tau$-neighbouring.
\end{definition}

This definition extends one of the neighbouring states used in previous works. In \cite{quantumDP,liu2021,franca}, two states are neighbouring if they have bounded trace distance $\tau$, i.e. if they are $\tau$-neighbouring. Moreover, setting $\ell=1$ and $\tau =1$  we recover the definition of quantum differential privacy based on convertibility by local measurements, used in \cite{aaronson2019gentle}.

This notion is particularly suitable to handle local measurements, i.e. measurements expressible as sums of local terms, as we show in \secref{sec:local}.
We remark that local measurements are of particular interest since they can be considered practically feasible measurements for extracting classical information from quantum data (or quantum systems). They also play a major role in variational learning algorithms as they are provably resilient to barren plateaus \cite{cerezo2021cost}.

\begin{table}[t]
\caption{As we discuss in details in \secref{app:encodings}, the encodings above are $(\Xi,\tau)$-neighbouring-preserving for appropriate $\Xi$ and $\tau$ depending on the encodings. We assume that the initial vectors $\boldsymbol{x}$ and $\boldsymbol{x'}$ are neighbouring if $\|\boldsymbol{x}-\boldsymbol{x'}\|_0\leq\gamma_0$, $\|\boldsymbol{x}-\boldsymbol{x'}\|_1\leq\gamma_1$ and $\|\boldsymbol{x}-\boldsymbol{x'}\|_2\leq\gamma_2$. We also assumed that the Hamiltonian encoding is implemented by a $1D$ circuit of depth at most $L$.
We refer to \secref{app:encodings} for more details on the noise models.}
\label{tab:enc}
\vskip 0.15in
\begin{center}
\begin{small}
\begin{sc}
\begin{tabular}{lcccr}
\toprule
Encoding $\rho(\cdot)$ && $\max_{\mathcal{I} \in \Xi} |\mathcal{I}|$ && $\tau$ \\
\midrule

Amplitude encoding && $n$ &&  $\gamma_2$\\
Rotation encoding && $\gamma_0$&& $1$ \\
Coherent state encoding &&$\gamma_0$ &&$\sqrt{1-e^{-\gamma_2^2}}$ \\
1D-Hamiltonian encoding && $2L\gamma_0$ &&  $O(1)\gamma_1$\\
1D-Hamiltonian encoding (low noise) && $2L\gamma_0$ &&  $O(1)\sqrt{n}\exp(-L)$\\
1D-Hamiltonian encoding (high noise) && $2L\gamma_0$ &&  $O(1)\exp(-L)\gamma_1$\\

\bottomrule
\end{tabular}
\end{sc}
\end{small}
\end{center}
\vskip -0.1in
\end{table}

On the other hand, several encodings widely used in quantum machine learning are $(\Xi,\tau)$-neighbouring-preserving, for appropriate choices of $(\Xi,\tau)$. We include upper bounds for $\max_{\mathcal{I} \in \Xi} |\mathcal{I}|$ and $\tau$ in \tableref{tab:enc}.
We delay to \appref{app:encodings} the definition of the various encodings and the proof of upper bounds.

We also show that the notion of $(\Xi,\tau)$-neighbouring states degrades gently under quantum postprocessing, assuming that the post-processing channel has a bounded light-cone.

\begin{prop}[Robustness to quantum post-processing]
Let $\rho$ and $\sigma$ be two $(\Xi,\tau)$-neighbouring states and consider a channel $\Phi$ with light-cone bounded by $K$. Then $\Phi(\rho)$ and $\Phi(\sigma)$ are $(\Xi',\tau)$-neighbouring, where
\[
\max_{\mathcal{I} \in \Xi'} |\mathcal{I}|\leq K \max_{\mathcal{I} \in \Xi} |\mathcal{I}|.
\]
\end{prop}
\begin{proof}
The proposition follows from the fact that the trace distance is non-increasing and from the definition of light-cone provided in \secref{sec:back}. We have
\[
\frac{1}{2}\|\Phi(\rho) - \Phi(\sigma) \|_1 \leq \frac{1}{2}\|\rho - \sigma \|_1 \leq \tau
\]
Moreover,
\[
\Tr_\mathcal{J} \rho = \Tr_\mathcal{J} \sigma
\]
for $\mathcal{J} \in \Xi$.
Since the channel $\Phi$ has bounded light-cone $K$, there exists $\mathcal{J'}\subseteq [n]$ 
\[
\Tr_{\mathcal{J}'} \rho = \Tr_{\mathcal{J}'} \sigma
\]
where $|\mathcal{J}'| \leq K |\mathcal{J}|$. This implies the desired result.
\end{proof}

We conclude this section by observing that our definition can be easily related to the quantum Wasserstein distance of order 1.  Combining \lemref{lem:w1} and \eqref{eq:w1-tr}, we obtain
\begin{equation}
\label{eq:ell-tau-w1}
     \rho\overset{(\Xi,\tau)}{\sim}\sigma \implies W_1(\rho,\sigma) \leq \min \left\{\max_{\mathcal{I}\in\Xi}|\mathcal{I}| \frac{3}{2}\tau, n\tau\right\}.
\end{equation}
It's natural to ask whether it would be convenient to define neighbouring quantum states in terms of the $W_1$ distance. The answer to this question is twofold. On the one hand, when we dispose of a single copy of the input state, the $W_1$ distance leads to a suboptimal tradeoff between privacy and accuracy, as we show in 
\thmref{thm:conc_w1}.
On the other hand, when we dispose of multiple copies of the input state, neighbouring quantum states can be suitably defined in terms of the $W_1$ distance. We will discuss this alternative setting in \secref{sec:multiple}. 
\section{Improved privacy for states with bounded trace distance}
\label{sec:trace}
Before dealing with the general case of $(\Xi,\tau)$-neighbouring states, we provide several new results for $\tau$-neighbouring states, i.e. states with trace distance bounded by $\tau$. This corresponds to the definition previously explored in \cite{quantumDP,franca}.
In particular, we provide tighter guarantees for two private mechanisms, namely a generalised noisy channel and the addition of classical noise on the output of a quantum measurement. Following the convention used in \cite{franca}, we state the results of this section using the quantum hockey-stick divergence.

Let $\mathcal{M}(\cdot)$ an arbitrary channel and let $\mathcal{N}_p(\cdot) = p \frac{\mathbb{1}}{2^n} + (1-p) \mathcal{M}(\cdot)$. We briefly discuss how several noisy channels can be recovered as special cases of $\mathcal{N}_p(\cdot)$.
For $\mathcal{M}(\cdot)$ equal to the identity channel $\text{Id}(\cdot)$, $\mathcal{N}_p$ is the depolarising channel.
For $n=1$ we can also recover the single qubit Pauli channel $\mathcal{P}$ as a special case. Following \cite{nibp}, the action of $\mathcal{P}$ on a
local Pauli operator $\sigma \in \{X, Y, Z\}$ can be expressed as
\[
\mathcal{P}(\sigma) = q_\sigma \sigma,
\]
where $-1<q_X,q_Y,q_Z<1$. It's customary to characterize the noise strength with a single parameter $q= \sqrt{\max\{|q_X|, |q_Y |, |q_Z|\}}$.
Then for a single qubit state $\rho = \frac{1}{2}\left(\mathbb{1} + r_X X + r_Y Y + R_Z Z\right)$ we have
\begin{align*}
    \mathcal{P}(\rho) = \frac{1}{2}\left(\mathbb{1} + q_X r_X X + q_Y r_Y Y + q_Z r_Z Z\right) \\= (1-q^2)\frac{\mathbb{1}}{2} + q^2 \times \frac{1}{2}\left(\mathbb{1} + \frac{q_X}{q^2} r_X X + \frac{q_Y}{q^2} r_Y Y + \frac{q_Z}{q^2} r_Z Z\right) 
    \\:= (1-q^2) \frac{\mathbb{1}}{2} + q^2\mathcal{M'}(\rho),
\end{align*}
where we defined $\mathcal{M'}(\rho):=\frac{1}{2}\left(\mathbb{1} + \frac{q_X}{q^2} r_X X + \frac{q_Y}{q^2} r_Y Y + \frac{q_Z}{q^2} r_Z Z\right)$.
We proceed by analysing the privacy guarantees of the channel $\mathcal{N}_p$.

\begin{lemma}
\label{thm:dep1}
Let $\mathcal{N}_p(\cdot) = p \frac{\mathbb{1}}{2^n} + (1-p) \mathcal{M}(\cdot)$ a channel.
For $0\leq p\leq 1$ and $\gamma \geq 1$ we have 
\[
E_{\gamma'} (\mathcal{N}_p(\rho)\|\mathcal{N}_p(\sigma)) \leq (1-p)(1-\beta)E_\gamma(\rho\|\mathbb{1}/2^n) + (1-p)\beta E_\gamma(\rho\|\sigma),
\]
where $\gamma' = 1+ (1-p)(\gamma - 1)$ and $\beta = \gamma'/\gamma$.
\end{lemma}
\begin{proof}
    The result follows from \lemref{lem:ajc} by plugging $\rho_0 = \mathbb{1}/2^n$, $\rho_1 = \rho$ and $\rho_2=\sigma$.
\end{proof}

Recall that from Lemma IV.1 in \cite{franca} we have that for the depolarising noise (hence for $\mathcal{M} = \text{Id})$ and for any $\gamma\geq 1$,
\[
E_\gamma (\mathcal{N}_p(\rho)\|\mathcal{N}_p(\sigma)) \leq \max \left\{0,(1-\gamma) \frac{p}{2^n} + (1-p)E_\gamma(\rho\|\sigma)\right\}.
\]
In the following theorem, we extend this previous bound to an arbitrary channel $\mathcal{M}$ and we combine it with $\lemref{thm:dep1}$.

\begin{theorem}
\label{thm:combined}
Let $\mathcal{N}_p(\cdot) = p \frac{\mathbb{1}}{2^n} + (1-p) \mathcal{M}(\cdot)$ a channel.
For $0\leq p\leq 1$ and $\gamma' \geq 1$ we have 
   \begin{align*}
    E_{\gamma'} (\mathcal{N}_p(\rho)\|\mathcal{N}_p(\sigma)) \leq \\\min\left\{(1-p)(1-\beta)E_{\gamma}(\rho\|\mathbb{1}/2^n) + (1-p)\beta E_{\gamma}(\rho\|\sigma), \max \left\{0,(1-\gamma') \frac{p}{2^n} + (1-p)E_{\gamma'}(\rho\|\sigma)\right\}\right\}.
\end{align*} 
where $\gamma = 1 + (\gamma' - 1)/(1-p)$ and $\beta = \gamma'/\gamma$.
\end{theorem}
\begin{proof}
\lemref{thm:dep1} implies that
\[
E_{\gamma'} (\mathcal{N}_p(\rho)\|\mathcal{N}_p(\sigma)) \leq (1-p)(1-\beta)E_\gamma(\rho\|\mathbb{1}/2^n) + (1-p)\beta E_\gamma(\rho\|\sigma),
\]
Then it remains to show that
\[
E_{\gamma'} (\mathcal{N}_p(\rho)\|\mathcal{N}_p(\sigma)) \leq \max \left\{0,(1-\gamma') \frac{p}{2^n} + (1-p)E_{\gamma'}(\rho\|\sigma)\right\}.
\]
The proof closely follows the one of Lemma IV.1 and Lemma IV.4 in \cite{franca}. We have
\begin{align*}
    E_{\gamma'}(\mathcal{N}_p(\rho)\|\mathcal{N}_p(\sigma)) \\
    = \Tr((1-\gamma')p\frac{\mathbb{1}}{2^n} + (1-p)\mathcal{M}((\rho-\gamma'\sigma)))^+ \\
    = \Tr P^+((1-\gamma')p\frac{\mathbb{1}}{2^n} + (1-p)\mathcal{M}((\rho-\gamma'\sigma))), 
\end{align*}
where $P^+$ is the projector onto the positive subspace of $((1-\gamma')p\frac{\mathbb{1}}{2^n} + (1-p)\mathcal{M}((\rho-\gamma'\sigma)))$. Observe that 
\begin{align*}
    E_{\gamma'}(\mathcal{N}_p(\rho)\|\mathcal{N}_p(\sigma))>0 \quad\Rightarrow\quad \Tr P^+\geq 1.
\end{align*}
Considering this case we get
\begin{align*}
    &E_{\gamma'}(\mathcal{N}_p(\rho)\|\mathcal{N}_p(\sigma)) \\
    &= (1-{\gamma'})\frac{p}{2^n}\Tr P^+ + (1-p)(\Tr P^+(\mathcal{M}(\rho-{\gamma'}\sigma))) \\
    &\leq (1-{\gamma'})\frac{p}{2^n} + (1-p)E_{\gamma'}(\mathcal{M}(\rho)\|\mathcal{M}(\sigma)) \\
    &\leq (1-{\gamma'})\frac{p}{2^n} + (1-p)E_{\gamma'}(\rho\|\sigma) \\
    &\leq (1-{\gamma'})\frac{p}{2^n} + (1-p). 
\end{align*}
Note that for sufficiently large ${\gamma'}$ the upper bound could become negative, but one can easily check that in this case $E_{\gamma'}(\mathcal{N}_p(\rho)\|\mathcal{N}_p(\sigma))=0$ implying that we are in the other case.

\end{proof}

For single-qubit product channels, we give the following bound:
\begin{theorem}
\label{thm:local-noise}
   Let $\mathcal{N}_p(\cdot) = p \frac{\mathbb{1}}{2} + (1-p) \mathcal{M}(\cdot)$ a single-qubit channel.
For $0\leq p\leq 1$ and $\gamma' \geq 1$ we have 
   \begin{align*}
    E_{\gamma'} (\mathcal{N}^{\otimes k}_p(\rho)\|\mathcal{N}^{\otimes k}_p(\sigma)) \leq \\\min\left\{(1-p^k)(1-\beta)E_{\gamma}(\rho\|\mathbb{1}/2^k) + (1-p^k)\beta E_{\gamma}(\rho\|\sigma), \max \left\{0,(1-\gamma') \frac{p^k}{2^k} + (1-p^k)E_{\gamma'}(\rho\|\sigma)\right\}\right\}.
\end{align*} 
where $\gamma = 1 + (\gamma' - 1)/(1-p)$ and $\beta = \gamma'/\gamma$. 
\end{theorem}
\begin{proof}
It suffices to note that $\mathcal{N}^{\otimes k}_p$ can be rearranged as:
\[
\mathcal{N}^{\otimes k}_p(\cdot) = p^k \frac{\mathbb{1}}{2^k} + (1-p^k)\mathcal{M}'(\cdot),
\]
where $\mathcal{M}'$ is a quantum channel. Then the result follows from \thmref{thm:combined}.
\end{proof}

These first two technical results show that several quantum noisy channels contract the quantum hockey-stick divergence. This can be used to prove that those channels ensure quantum differential privacy for $\tau$-neighbouring states. In particular, we derive the following corollaries, that improve Lemma IV.2 and Lemma IV.5 in \cite{franca}.
\begin{cor}
\label{cor:global}
Let $\mathcal{N}_p(\cdot) = p \frac{\mathbb{1}}{2^n} + (1-p) \mathcal{M}(\cdot)$ a channel.
$\mathcal{N}_p$ is $(\epsilon,\delta)$-DP with respect to $\tau$-neighbouring states with 
\begin{equation}
\label{eq:old_bound}
   \delta \leq \max \left\{0,(1-e^\epsilon)\frac{p}{2^n} + (1-p)\tau\right\}. 
\end{equation}
Let $\gamma = 1 + (e^\epsilon -1)/(1-p)$ and $\beta = e^\epsilon/\gamma$. Under the additional assumption that the input state $\rho$ satisfies $E_\gamma(\rho\|\frac{{\rho}}{\mathbb{1}/2^n})\leq \eta$, we also have
\begin{equation}
\label{eq:new_bound}
    \delta \leq (1-p)(1-\beta)\eta + (1-p)\beta\tau.
\end{equation}

\end{cor}
It's not straightforward whether \eqref{eq:new_bound} provides any advantage over \eqref{eq:old_bound}. Thus, in \figref{fig:privacy-profiles} we plot both bounds of $\delta$ as a function of $\epsilon$, for a specific set of parameters, and we observe that no bound is always tighter, and thus the choice of the bound will depend on the value of $\epsilon$. An upper bound of $\delta$ as a function of $\epsilon$ is also referred to as \emph{privacy profile}, a concept introduced in \cite{subsampling}.
\begin{cor}
\label{cor:local}
Let $\mathcal{N}_p(\cdot) = p \frac{\mathbb{1}}{2} + (1-p) \mathcal{M}(\cdot)$ single-qubit a channel.
$\mathcal{N}^{\otimes k}_p$ is $(\epsilon,\delta)$-DP with respect to $\tau$-neighbouring states with 
\[
\delta \leq \max \left\{0,(1-e^\epsilon)\frac{p^k}{2^k} + (1-p^k)\tau\right\}.
\]
Let $\gamma = 1 + (e^\epsilon -1)/(1-p^k)$ and $\beta = e^\epsilon/\gamma$. Under the additional assumption that the input state $\rho$ satisfies $E_\gamma(\rho\|\frac{{\rho}}{\mathbb{1}/2^k})\leq \eta$, we also have
\[
\delta \leq (1-p^k)(1-\beta)\eta + (1-p^k)\beta\tau.
\]

\end{cor}

\subsubsection*{Bounding privacy with the purity}
Our results improve the prior bounds under the additional assumption that the divergence $E_{\gamma}(\rho\|\mathbb{1}/2^n)$ is relatively small.
The value of $E_{\gamma}(\rho\|\mathbb{1}/2^n) $ can be thought as a ``distance'' between the state $\rho$ and the maximally mixed state, thus small values of $E_{\gamma}(\rho\|\mathbb{1}/2^n)$ are associated to high levels of noise.  Hence, we can connect it to the purity $\Tr[\rho^2]$ of the state $\rho$, or the related $D_2$ divergence. By definition, we have
\[
\Tr[\rho^2] = 2^{-n+D_2(\rho\|\mathbb{1}/2^n)}.
\]
The hockey stick divergence and the Rényi divergence satisfy the following relationship (\cite{tomamichel2015quantum}, Proposition 6.22)
\begin{equation}
\label{eq:hs_d2}
    E_{e^\epsilon}(\rho\|\mathbb{1}/2^n) \leq \delta,
\end{equation}

where $\epsilon = D_2(\rho\|\mathbb{1}/2^n) - \log(1-\sqrt{1-\delta^2})\leq D_2(\rho\|\mathbb{1}/2^n) + \log(2/\delta^2)$.
We also note that two states with low purity are also close in hockey-stick divergence:
\[
E_{\gamma}(\rho\|\sigma) \leq E_1(\rho\|\sigma) \leq E_1(\rho\|\mathbb{1}/2^n) + E_1(\sigma\|\mathbb{1}/2^n).  
\]
And then $E_1(\rho\|\mathbb{1}/2^n) = \frac{1}{2}\left\|\rho-{\mathbb{1}}/{2^n}\right\|_1$ can be bounded either with the quantum Bretagnolle Huber inequality (\lemref{lem:qbh}) or the Pinsker's inequality. 
Now, we show how \corref{cor:global} and \corref{cor:local} can be rephrased in terms of the purity of the input state.
\begin{cor}
\label{cor:global_purity}
Let $\mathcal{N}_p(\cdot) = p \frac{\mathbb{1}}{2^n} + (1-p) \mathcal{M}(\cdot)$ a channel that acts on state $\rho$ with bounded purity $\Tr[\rho^2]\leq\zeta<1$.
Let $\gamma = 1 + (e^\epsilon -1)/(1-p)$, $\beta = e^\epsilon/\gamma$ and $\eta = \sqrt{2n\zeta^{\frac{1}{\log 2}}\gamma^{-1}}$. Then $\mathcal{N}_p$ is $(\epsilon,\delta)$-DP with respect to $\tau$-neighbouring states with 
\[
\delta \leq (1-p)(1-\beta)\eta + (1-p)\beta\tau.
\]
\end{cor}
\begin{proof}
The proof follows by plugging the relation between purity and hockey-stick divergence into \corref{cor:global}.
We have
\[
D_2(\rho\|\mathbb{1}/2^n) \leq \log_2(\zeta) + n,
\]
and hence, by \eqref{eq:hs_d2},
\[
E_\gamma(\rho\|\mathbb{1}/2^n) \leq \sqrt{2n\zeta^{\frac{1}{\log 2}}\gamma^{-1}}:=\eta,
\]
which satisfies the hypothesis of \corref{cor:global}.
\end{proof}
Proceeding in a similar way can also prove a purity-based bound for local channels.
\begin{cor}
\label{cor:local_purity}
Let $\mathcal{N}_p(\cdot) = p \frac{\mathbb{1}}{2^n} + (1-p) \mathcal{M}(\cdot)$ a single-qubit channel and assume that $\mathcal{N}_p^{\otimes k}$ acts on state $\rho$ with bounded purity $\Tr[\rho^2]\leq\zeta<1$.
Let $\gamma = 1 + (e^\epsilon -1)/(1-p^k)$, $\beta = e^\epsilon/\gamma$ and $\eta = \sqrt{2n\zeta^{\frac{1}{\log 2}}\gamma^{-1}}$. Then $\mathcal{N}_p$ is $(\epsilon,\delta)$-DP with respect to $\tau$-neighbouring states with 
\[
\delta \leq (1-p^k)(1-\beta)\eta + (1-p^k)\beta\tau.
\]
\end{cor}

\begin{figure}
    \centering
    \includegraphics[width=10cm]{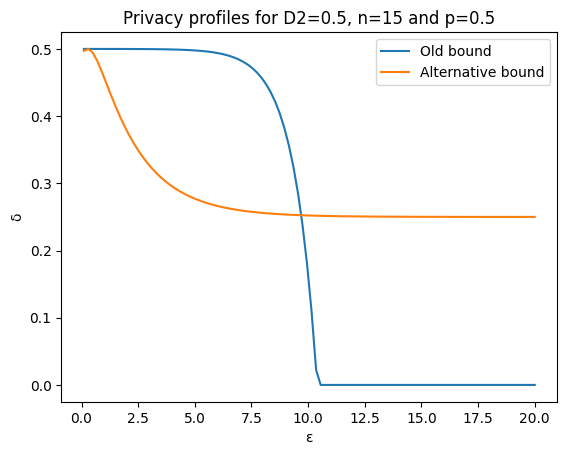}
    \caption{In this figure we compare the former upper bound from \cite{franca} (\eqref{eq:old_bound}) with the novel upper bound provided in this section (\eqref{eq:new_bound}). We emphasise that each bound outperforms the other for some values of $\epsilon$. We assumed that the input state satisfy $D_2(\rho\|\mathbb{1}/2^n)\leq 0.5, n= 15$ and $p=0.5$. The upper bound on $\tau$ is derived from $\|\rho-\sigma\|_1 \leq \|\rho- \mathbb{1}/2^n\|_1 + \|\rho- \mathbb{1}/2^n\|_1 \leq 2\sqrt{2 D_2(\rho\|\mathbb{1}/2^n)}$, i.e. combining the triangle inequality and the Pinsker's inequality. 
    }
    \label{fig:privacy-profiles}
\end{figure}

\subsection{Privacy  via classical post-processing}
Now, we show that the output of a quantum measurement can be privatised by adding classical noise. This result is particularly interesting since firstly, it provides a practical approach for using the existing tools and techniques from classical differential privacy to privatize the outputs of our quantum systems and algorithms, and secondly allows us to be able to combine classical noise with the output distributions resulting from a quantum measurement.
In particular, we can account for quantum and classical noise in the analysis by noting that quantum noisy channels contract the trace distance between any two quantum states, and, moreover, the privacy guarantees obtained by adding classical noise are inversely proportional to the trace distance between two neighbouring states.

\begin{lemma}
\label{lem:classical}
    Let $\rho,\sigma$ such that $\frac{1}{2}\|\rho-\sigma\|_1\leq\tau$. Let $M$ be a POVM measurement and $\Lambda$ a classical channel such that $\forall x,x' \in \mathsf{range}(M): E_{e^\epsilon}(\Lambda(x)\|\Lambda(x'))\leq \delta$. Then we have that
        \[
        E_{e^{\epsilon '}} (\Lambda(M(\rho))\|\Lambda(M(\sigma)))\leq \tau \delta,
        \]
    where ${{\epsilon '}} = \log(1+\tau({e^{\epsilon }}-1))$, which for small $\epsilon$ gives $\epsilon'\simeq\tau\epsilon$.
\end{lemma}
\begin{proof}
Let  $\nu := M(\rho)$ and  $\nu' := M(\sigma)$.   We have that
\[
d_{TV}(\nu,\nu'):=\eta\leq \tau,
\]
which follows from the data processing inequality.
Moreover,  there always exists some distributions $\nu_0,\nu_1,\nu_1'$ such that
\[
\nu = (1-\eta)\nu_0 + \eta\nu_1,\;\;\;\;\;\;\; \nu' = (1-\eta)\nu_0 + \eta\nu_1'.
\]
The above identities are discussed in detail in (\cite{subsampling}, Section 3).
We also have, 
\[\max\{E_{e^{\epsilon }}(\Lambda(\nu_1)\|\Lambda(\nu_0)), E_{e^{\epsilon }}(\Lambda(\nu_1)\|\Lambda(\nu_1'))\} \leq \delta\]
This follows by noting that $\nu_0,\nu_1,\nu_1'$ are supported in $\mathsf{range}(M)$ and applying the (standard) joint-convexity of the hockey-stick divergence.
By advanced joint convexity (\lemref{lem:ajc}), we have that for all states $\rho_0,\rho_1,\rho_2$  and $\gamma' = 1+ (1-p)(\gamma - 1)$,
\[
 E_{\gamma'}(p\rho_0 + (1-p)\rho_1\|p\rho_0 + (1-p)\rho_2)
\leq (1-p)(1-\beta) E_\gamma (\rho_1\|\rho_0) + (1-p)\beta E_\gamma(\rho_1\|\rho_2),
\]
Then, 
\[
E_{e^{\epsilon '}} (\Lambda(M(\rho))\|\Lambda(M(\sigma)))\leq \tau \delta.
\]
\end{proof}
\lemref{lem:classical} is stated in terms of a general classical noisy channel. In the following theorem we consider the special cases of the Laplace and Gaussian mechanisms, two noisy channels widely used in many differentially private classical algorithms and defined in \secref{sec:dp}.
\begin{theorem}
\label{thm:cl_noise}
Let $M$ a  measurement with range $[a,a + \Delta]$ for $a\in\mathbb{R}$.
\begin{itemize}
    \item (Laplace mechanism) Let $\Lambda_{\mathcal{L},b}$ the Laplace noise of scale $b$. Then $\Lambda_{\mathcal{L},b}(M(\cdot))$ is $\epsilon'$-DP with respect to $\tau$-neighbouring states, where
    \[
    \epsilon' = \log(1+ \tau(e^{\Delta/b} -1)).
    \]
    \item (Gaussian mechanism) Let $\Lambda_{\mathcal{G},\sigma}$ the Gaussian noise of variance $\sigma^2 \geq 2\ln (1.25/\delta)\Delta^2/\epsilon^2$. Then $\Lambda_{\mathcal{G},\sigma}(M(\cdot))$ is $(\epsilon',\delta')$-DP with respect to $\tau$-neighbouring states, where
    \[
    \epsilon' = \log(1+\tau(e^{\epsilon}-1))\;\;\;\text{ and } \;\;\; \delta' = \tau \delta.
    \]
\end{itemize}
\end{theorem}

\begin{proof}
The theorem follows by replacing the channel $\Lambda$ in \lemref{lem:classical}  with the Laplace and Gaussian noise, respectively.
\end{proof}

\subsection{Implications for quantum-inspired sampling}
\label{sec:inspired}
As the trace distance generalizes the total variation distance, the range of applicability of \thmref{thm:cl_noise} includes also classical algorithms. In particular, we show here an application for private quantum-inspired sampling.
In quantum-inspired algorithms  \cite{tang1, tang2, tang3, tang4, dequantizing}, a classical vector $u\in\mathbb{C}^{N}$ is accessed through quantum-inspired sampling: i.e. an entry $u_i$ is sampled with probability proportional to $|u_i|^2$. This is equivalent to encoding $u$ into the state
\[
\ket{u} = \frac{1}{\|u\|_{2}}\sum_{i=1}^{N} u_i \ket{i},
\]
and performing a computational-basis measurement. Let $p_u$ be the distribution induced by such measurements. 
%
Say that $u\sim u'$ if $u$ and $u'$ differ in only one entry. In particular, let $u_i=u_i'$ for all $i\neq j$. 
\begin{align*}
   |\|u\|_2^2 - \|u'\|_2^2| = \left|\sum_i |u_i|^2 - \sum_i |u'_i|^2 \right| \\\leq \left|\sum_{i\neq j} |u_i|^2 - \sum_{i \neq j} |u'_i|^2 + |u_j|^2 - |u'_j|^2 \right|\leq \max\{|u_j|^2 ,|u'_j|^2 \} 
\end{align*}
It's easy to see that $p_u$ and $p_{u'}$ are close in total variation distance. 
\begin{align*}
\stat{p_u - p_{u'}} = \frac{1}{2}\sum_i \left|\frac{|u_i|^2}{\|u\|_2^2}-\frac{|u'_i|^2}{\|u'\|_2^2}\right| 
\\\leq  \frac{1}{2} \left(\sum_{i\neq j}|u_i|^2 \left|\frac{1}{\|u\|_2^2}-\frac{1}{\|u'\|_2^2}\right| +  \left|\frac{|u_j|^2}{\|u\|_2^2}-\frac{|u'_j|^2}{\|u'\|_2^2}\right| \right)
\\ \leq \frac{1}{2}\left(\min\{\|u\|_2^2,\|u'\|_2^2\}\frac{\max\{|u_j|^2 ,|u'_j|^2 \} }{\|u\|_2^2\|u'\|_2^2}+ \frac{|u_j|^2}{\|u\|_2^2} + \frac{|u'_j|^2}{\|u'\|_2^2}\right) 
\\\leq\frac{3}{2} \max\left\{\frac{|u_j|^2}{\|u\|_2^2},\frac{|u'_j|^2}{\|u'\|_2^2}\right\} :=\alpha.
\end{align*}
Then, by subadditivity of the total variation distance,
\[
\stat{p_u^{\otimes m} - p_{u'}^{\otimes m}} \leq m\alpha.
\]

We will show the intuitive fact that quantum-inspired subsampling amplifies DP. 
First, we can consider the encoding $u \mapsto p_u^{\otimes m}$ and derive the following special case of \thmref{thm:cl_noise}.

\begin{cor}
Let $u,u'$ be neighbouring if they differ in at most one entry. Consider the oracle $\mathsf{O}_u$ that returns a $u_i$ with probability $\frac{|u_i|^2}{\|u\|_2^2}$. For $a\in\mathbb{R}$ and $\Delta\geq0$, let $\mathcal{S}$ a randomised algorithm with range $[a,a+\Delta]$  that makes $m$ queries to $\mathsf{O}_u$ and assume that $\frac{3}{2} \frac{|u_j|^2}{\|u\|_2^2} \leq \alpha$. 
\begin{itemize}
    \item (Laplace mechanism) Let $\Lambda_{\mathcal{L},b}$ the Laplace noise of scale $b$. Then $\Lambda_{\mathcal{L},b}(\mathcal{S}(\cdot))$ is $\epsilon'$-DP, where
    \[
    \epsilon' = \log(1+ \alpha m(e^{\Delta/b} -1)).
    \]
    \item (Gaussian mechanism) Let $\Lambda_{\mathcal{G},\sigma}$ the Gaussian noise of variance $\sigma^2 \geq 2\ln (1.25/\delta)\Delta^2/\epsilon^2$. Then $\Lambda_{\mathcal{G},\sigma}(\mathcal{S}(\cdot))$ is $(\epsilon',\delta')$-DP, where
    \[
    \epsilon' = \log(1+\alpha m(e^{\epsilon}-1))\;\;\;\text{ and } \;\;\; \delta' = \alpha m \delta.
    \]
\end{itemize}
\end{cor}

The approach described above is tailored to noise-adding mechanisms. In \appref{app:inspired} we provide a more general result that applies to any private mechanism and it builds upon prior work on privacy amplification by subsampling  \cite{subsampling,ullman}.

\section{Differential privacy for $(\Xi,\tau)$-neighbouring states}
\label{sec:local}

While in \secref{sec:trace} we provided tighter bounds for quantum differential privacy with respect to states with bounded trace distance, here we add two additional ingredients: the locality of the measurements and the generalised neighbouring relationship defined in \secref{sec:neighbour}.
Under these stronger assumptions, we can improve the privacy guarantees of local noisy channels and classical post-processing.
First, we need to introduce the following quantity.

\begin{definition}[Worst-case quantum sensitivity] Let $O$ be an observable expressed as a weighted sum of Pauli operators,  $O = \sum_{P \in \{X,Y,Z,\mathbb{1}\}^n} c_P  P$. Let $\mathcal{I} \subseteq [n]$ and consider the subset $\mathcal{S}_\mathcal{I}$ of all the Pauli strings that act non trivially on $\mathcal{I}$. The worst-case quantum sensitivity of $O$ with respect to $\mathcal{I}$ is defined as
\[
\Delta(O;\mathcal{I}):= 2 \sum_{P\in\mathcal{S}_\mathcal{I}}|c_P|.
\]
Let $\Xi \subseteq P([n])$, i.e. $\Xi$ is a collection of subsets of $[n]$. The worst-case quantum sensitivity of $O$ with respect to $\Xi$ is defined as 
\[
\Delta_\Xi(O):= \max_{\mathcal{I}\in\Xi} \Delta(O;\mathcal{I}).
\]
We will omit the index $\Xi$ and simply write $\Delta(O)$ when there is no ambiguity.
\end{definition}
So, if $O = \sum_{i=1}^n Z_i$ and $\Xi = \{\{1\},\{2\},\dots,\{n\}\}$, the worst-case quantum sensitivity equals $\Delta(O)=2$. This is consistent with the fact that, if $\rho$ and $\sigma$ satisfy $\Tr_j \rho = \Tr_j \sigma$, then all the terms but $Z_j$ induce the same distributions when measured on either $\rho$ or $\sigma$. Moreover, the outcome of term $Z_j$ will be either $1$ or $-1$, then it belongs to an interval of length 2.  We can also consider the more general case where $O_\ell = \sum_{i=1}^n \bigotimes_{j=i}^{i+\ell-1} Z_j$ and $\Xi = \{ \{i,i+1,\dots,i+k\}|\text{ for $i =1,2,\dots,n-k$}\}$. It's easy to see that $\Delta(O_\ell) = 2k + 4\ell -4$.

We can now state the first result of this section, concerning a class of local noisy channels, which includes the local Pauli noise. 
\begin{theorem}[Generalised private measurement via local noisy channels]
\label{thm:local_noisy}
Let $O =\sum_P c_P P$ be an observable consisting of a weighted sum of commuting Pauli operators. Let $\mathcal{M}$ an arbitrary single qubit channel and let $\mathcal{N}(\cdot)= p{\mathbb{1}}/{2} + (1-p)\mathcal{M}(\cdot)$. Let $k = \max_{\mathcal{I}\in\Xi}|\mathcal{I}|$.
Then $\mathcal{O}\circ \mathcal{N}^{\otimes n}$ satisfies $(\epsilon,\delta_k)$-DP with respect to $(\Xi,\tau)$-neighbouring states, where
\[
\delta_k \leq \max \left\{0,(1-e^\epsilon) \frac{p^k}{2^k} + (1-p^k)\tau\right\}.
\]
Let $\gamma = 1 + (e^\epsilon - 1)/(1-p)$ and $\beta = e^\epsilon/\gamma$.
Under the additional assumption the the input state $\rho$ satisfies $E_\gamma(\rho\|\mathbb{1}/2^n)\leq \eta$, the following inequality also holds
\[
\delta_k \leq (1-p^k)(1-\beta)\eta + (1-p^k)\beta \tau.
\]

\end{theorem}

\begin{proof}
Since $\rho\overset{(\Xi,\tau)}{\sim }\sigma$, there exists $\mathcal{I}\in\Xi$ such that
\begin{equation}
\label{eq:local_id}
    \Tr_\mathcal{I} \rho = \Tr_\mathcal{I} \sigma\;\; \text{ and }\;\;|\mathcal{I}|\leq k.
\end{equation}
We also have
\[
\mathcal{N}^{\otimes n}(\rho) = p^{|\mathcal{I}|} \left(\Tr_\mathcal{I} \mathcal{M}(\rho)\otimes \frac{\mathbb{1}}{2^{|\mathcal{I}|}}\right)
+ (1-p^{|\mathcal{I}|})\mathcal{M}(\rho).
\]
The measurement $O$ can be implemented by measuring each qubit in a different Pauli basis and then performing classical postprocessing. As quantum differential privacy is robust to postprocessing, we only need to prove that Pauli measurements preserve $(\epsilon,\delta_k)$-DP.
We can assume without loss of generality that the qubits in the subsystem $\mathcal{I}^c$ are measured first, since we assumed that $O$ is a weighted sum of commuting Pauli operators, and hence the measurement order doesn't alter the overall statistics.
Assume that measuring the subsystem $\mathcal{I}^c$ produces the outcome $\boldsymbol{y} \in\{\pm 1\}^{n-|\mathcal{I}|}$. \eqref{eq:local_id} implies that
\[
p(\boldsymbol{y}) := \Pr[\text{$\boldsymbol{y}$ is obtained on input $\rho$}] = \Pr[\text{$\boldsymbol{y}$ is obtained on input $\sigma$}].
\]
Denote by $\rho_{\boldsymbol{y}}$ the post-measurement state produced by measuring the system $\mathcal{I}^c$ and obtaining outcome $\boldsymbol{y}$.
Let $\mathcal{T}_{\boldsymbol{y}}$ be the quantum channel mapping $\rho$ to $\rho_{\boldsymbol{y}}$.
\begin{align*}
  \Tr_{\mathcal{I}^c} \mathcal{T}_{\boldsymbol{y}}(\mathcal{N}^{\otimes n}(\rho)) \\
  = p^{|\mathcal{I}|}\Tr_{\mathcal{I}^c} \mathcal{T}_{\boldsymbol{y}} \left(\Tr_\mathcal{I} \mathcal{M}(\rho)\otimes \frac{\mathbb{1}}{2^{|\mathcal{I}|}}\right)+
  (1-p^{|\mathcal{I}|})\Tr_{\mathcal{I}^c} \mathcal{T}_{\boldsymbol{y}}( \mathcal{M}(\rho))
  \\ = p^{|\mathcal{I}|} \frac{\mathbb{1}}{2^{|\mathcal{I}|}} +
  (1-p^{|\mathcal{I}|})\Tr_{\mathcal{I}^c} \mathcal{T}_{\boldsymbol{y}}( \mathcal{M}(\rho))\\
  := p^{|\mathcal{I}|} \frac{\mathbb{1}}{2^{|\mathcal{I}|}} +
  (1-p^{|\mathcal{I}|})\mathcal{M}'(\rho).
\end{align*}
where the second equality follows from $\mathcal{T}_{\boldsymbol{y}} \left(\Tr_\mathcal{I} \mathcal{M}(\rho)\otimes \frac{\mathbb{1}}{2^{|\mathcal{I}|}}\right) = \Tr_\mathcal{I} \mathcal{T}(\mathcal{M}(\rho))\otimes \frac{\mathbb{1}}{2^{|\mathcal{I}|}}$ and we defined $\mathcal{M}' := \Tr_{\mathcal{I}^c}\circ \mathcal{T}_{\boldsymbol{y}} \circ \mathcal{M} $.
By \corref{cor:local},
\begin{equation}
\label{eq:privacy-conditioned}
  E_{e^{\epsilon}}(\Tr_{\mathcal{I}^c} \mathcal{T}_{\boldsymbol{y}}(\mathcal{N}^{\otimes n}(\rho)) \|\Tr_{\mathcal{I}^c} \mathcal{T}_{\boldsymbol{y}}(\mathcal{N}^{\otimes n}(\sigma))  ) \leq \delta_k.
\end{equation}

In order to prove that measuring $O$ on $\mathcal{N}^{\otimes n}(\rho)$ preserves $(\epsilon,\delta_k)$-DP, it's sufficient consider the outcome $\boldsymbol{y}$ and the partial post-measurement state $\Tr_{\mathcal{I}^c} \mathcal{T}_{\boldsymbol{y}}( \mathcal{N}^{\otimes n}(\rho))$. Thus we need to ensure that
\[
E_{e^\epsilon}\left(\sum_{\boldsymbol{y}} p_{\boldsymbol{y}} (\Tr_{\mathcal{I}^c} \mathcal{T}_{\boldsymbol{y}}( \mathcal{N}^{\otimes n}(\rho)) \otimes \ket{\boldsymbol{y}}\bra{\boldsymbol{y}})\bigg\|
    \sum_{\boldsymbol{y}} p_{\boldsymbol{y}} (\Tr_{\mathcal{I}^c} \mathcal{T}_{\boldsymbol{y}}( \mathcal{N}^{\otimes n}(\sigma))\otimes \ket{\boldsymbol{y}}\bra{\boldsymbol{y}})\right) \leq \delta(\epsilon,k)
\]

We also have, for all $\gamma\geq 1$,
\begin{align}
\begin{split}
\label{eq:conv-stab}
    E_\gamma\left(\sum_{\boldsymbol{y}}  p_{\boldsymbol{y}} (\Tr_{\mathcal{I}^c} \mathcal{T}_{\boldsymbol{y}}(\mathcal{N}^{\otimes n}(\rho)) \otimes \ket{\boldsymbol{y}}\bra{\boldsymbol{y}})\bigg\|
    \sum_{\boldsymbol{y}}  p_{\boldsymbol{y}} (\Tr_{\mathcal{I}^c} \mathcal{T}_{\boldsymbol{y}} (\mathcal{N}^{\otimes n}(\sigma)) \otimes \ket{\boldsymbol{y}}\bra{\boldsymbol{y}})\right)
    \\\leq 
    \sum_{\boldsymbol{y}} p_{\boldsymbol{y}} E_\gamma\left(\Tr_{\mathcal{I}^c} \mathcal{T}_{\boldsymbol{y}}( \mathcal{N}^{\otimes n}(\rho) ) \otimes \ket{\boldsymbol{y}}\bra{\boldsymbol{y}}\|
    \Tr_{\mathcal{I}^c} \mathcal{T}_{\boldsymbol{y}}(\mathcal{N}^{\otimes n}(\sigma) )\otimes \ket{\boldsymbol{y}}\bra{\boldsymbol{y}}\right)
    \\\leq 
     \sum_{\boldsymbol{y}} p_{\boldsymbol{y}} E_\gamma\left(\Tr_{\mathcal{I}^c} \mathcal{T}_{\boldsymbol{y}}(\mathcal{N}^{\otimes n}(\rho)) \|
    \Tr_{\mathcal{I}^c} \mathcal{T}_{\boldsymbol{y}}(\mathcal{N}^{\otimes n}(\sigma))\right)
    \\ \leq \max_{\boldsymbol{y}} E_\gamma\left(\Tr_{\mathcal{I}^c} \mathcal{T}_{\boldsymbol{y}} (\mathcal{N}^{\otimes n}(\rho))\|
    \Tr_{\mathcal{I}^c} \mathcal{T}_{\boldsymbol{y}} (\mathcal{N}^{\otimes n}(\sigma)) \right),
\end{split}
\end{align}
where the second line follows from the convexity of the hockey-stick divergence \eqref{eq:conv-hs} and the third line follows from the stability of the hockey-stick divergence \eqref{eq:stab-hs}.
Combining \eqref{eq:privacy-conditioned} with \eqref{eq:conv-stab} gives the desired result:
\[
E_{e^\epsilon}(\mathcal{O}(\mathcal{N}^{\otimes n}(\rho))\|\mathcal{O}(\mathcal{N}^{\otimes n}(\sigma)))\leq \delta_k.
\]
\end{proof}

We emphasise that the number of qubits $n$ appearing in the guarantees of \corref{cor:local} is now replaced by $k = \max_{\mathcal{I}\in \Xi}|\mathcal{I}|$. Thus if $k = \text{polylog}(n)$, this new bound is exponentially tighter than the previous one.
In a similar fashion, we can adapt \thmref{thm:cl_noise} to the generalised neighbouring relationship.

\begin{theorem}[Generalised private measurement via classical post-processing]
\label{thm:cl_noise_local}
Let $\rho$ and $\sigma$ two $(\Xi,\tau)$-neighbouring quantum states, i.e. $\rho \overset{(\Xi,\tau)}{\sim}{\sigma}$. Let $O$ be an observable, and denote $\mathcal{O}$ as a quantum-to-classical channel implementing a measurement of $O$.
\begin{itemize}
    \item (Laplace mechanism) Let $\Lambda_{\mathcal{L},b}$ the Laplace noise of scale $b$. Then $\Lambda_{\mathcal{L},b}(\mathcal{O}(\cdot))$ is $\epsilon'$-DP with respect to ${(\Xi,\tau)}$-neighbouring states, where
    \[
    \epsilon' = \log(1+ \tau(e^{\Delta(O)/b} -1)).
    \]
    \item (Gaussian mechanism) Let $\Lambda_{\mathcal{G},\sigma}$ the Gaussian noise of variance $\sigma^2 \geq 2\log (1.25/\delta)\Delta(O)^2/\epsilon^2$. Then $\Lambda_{\mathcal{G},\sigma}(\mathcal{O}(\cdot))$ is $(\epsilon',\delta')$-DP with respect to ${(\Xi,\tau)}$-neighbouring states, where
    \[
    \epsilon' = \log(1+\tau(e^{\epsilon}-1))\;\;\;\text{ and } \;\;\; \delta' = \tau \delta.
    \]
\end{itemize}
\end{theorem}
\begin{proof}
Proceeding as in the proof of \thmref{thm:local_noisy}, consider $\mathcal{I}\in\Xi$ such that $\Tr_\mathcal{I} \rho = \Tr_\mathcal{I} \sigma$ and let $\mathcal{S}_\mathcal{I}$ be the subset of all the Pauli strings that act non trivially on $\mathcal{I}$. Thus, we can decompose $O$ as $O= O_1 + O_2$, where
$O_1 =  \sum_{P \not\in \mathcal{S}_\mathcal{I}} c_P  P$
and
$O_2 = O - O_1 = \sum_{P \in \mathcal{S}_\mathcal{I}} c_P  P$.
Assume without loss of generality that $O_1$ is measured first. Since $\Tr_\mathcal{I} \rho = \Tr_\mathcal{I} \sigma$ and $O_1$ acts non trivially only on $\mathcal{I}^c = [n]\setminus \mathcal{I}$, then this measurement produces no loss of privacy, i.e.
\[
\forall y : p(y):=\Pr_\rho[O_1 = y] =\Pr_\sigma[O_1 = y].
\]
Observe that $O_2$ is a measurement whose output is comprised into $[-\Delta(O)/2,\Delta(O)/2]$. 
Moreover,  let $\rho_y$ be the post-measurement state obtained when $O_1$ returns outcome $y$. As the trace distance is non-increasing, we have,
\[
\frac{1}{2}\|\rho_y - \sigma_y\|\leq \frac{1}{2}\|\rho-\sigma\| \leq \tau,
\]
Conditioning on input $y$, the output of $O=O_1+O_2$ lies in $[y-\Delta/2,y+\Delta/2]$.
Then \thmref{thm:cl_noise} yields
\begin{align*}
   E_{e^{\epsilon'}}\left(\sum_y p(y)\Lambda_{\mathcal{L},b}(\mathcal{O}(\rho_y))\bigg\|\sum_y p(y)\Lambda_{\mathcal{L},b}(\mathcal{O}(\rho_y))\right)
   \\ \leq \max_y E_{e^{\epsilon'}}\left(\Lambda_{\mathcal{L},b}(\mathcal{O}(\rho_y))\|\Lambda_{\mathcal{L},b}(\mathcal{O}(\rho_y))\right)
   \leq 0.
\end{align*}
for $\epsilon' = \log(1+\tau(e^{\Delta(O)/b}-1))$.  Similarly, replacing the Laplace noise with the Gaussian noise and applying again \thmref{thm:cl_noise},
\begin{align*}
    E_{e^\epsilon}\left(\sum_y p(y)\Lambda_{\mathcal{G},\sigma}(\mathcal{O}(\rho_y))\bigg\|\sum_y p(y)\Lambda_{\mathcal{G},\sigma}(\mathcal{O}(\rho_y))\right)
\\ \leq \max_y E_{e^\epsilon}\left(\Lambda_{\mathcal{G},\sigma}(\mathcal{O}(\rho_y))\|\Lambda_{\mathcal{G},\sigma}(\mathcal{O}(\rho_y))\right)
   \leq \delta',
\end{align*}
where $\sigma^2\geq 2\log (1.25/\delta)\Delta(O)^2/\epsilon^2$, $\epsilon' = \log(1+\tau(e^\epsilon - 1))$ and $\delta' = \tau \delta$.
\end{proof}
We observe that similar results can be derived for multiple sources of noise, beyond the Laplace or the Gaussian channels, along the lines of \lemref{lem:classical}. We leave it to the reader to extend \thmref{thm:cl_noise_local} to alternative stochastic channels.

\section{The cost of quantum differential privacy}
\label{sec:cost}

Differential privacy, both in the classical and in the quantum setting, can be achieved by introducing noise into the computation, thus reducing the final accuracy. Intuitively, large values of $\epsilon$ can be attained with little loss in accuracy, while for $\epsilon = 0$ the output is totally independent of the input.
In particular, if an algorithm is $\epsilon$-DP with respect to Hamming distance, we have that
\begin{equation}
\label{eq:hamming}
  \forall x,x' : D_\infty(\mathcal{A}(x)\|\mathcal{A}(x'))\leq \epsilon n,  
\end{equation}
thus if $\epsilon = O(1/n)$, any pair of inputs (not necessarily neighbouring) are mapped to outputs $O(1)$-close in max-divergence. 
This result follows from the fact that the max-relative entropy satisfies the triangle inequality (both in the classical and in the quantum cases), i.e. $\forall \rho_1,\rho_2,\sigma: D_\infty(\rho_1\|\rho_2) \leq D_\infty(\rho_1\|\sigma) + D_\infty(\sigma\|\rho_2)$.
We can pick a sequence of $n+1$ inputs $x_0,x_1,\dots,x_n$ such that $x=x_0,$ $x'=x_n$ and $x_i\sim x_{i+1}$. Then iterating the triangle inequality yields \eqref{eq:hamming}.
However, for most applications $\epsilon$ can be chosen as a constant independent of $n$, avoiding this undesired concentration of the output around a unique value.

A vast portion of the literature about differential privacy is devoted to optimising the tradeoff between the value of $\epsilon$ and the loss in utility.
In this section we make a crucial observation: the privacy-utility tradeoff doesn't depend solely on the value of $\epsilon$, but also on the notion of neighbouring inputs. Thus, the privacy-utility tradeoff is an important figure of merit for the comparison of different approaches to quantum differential privacy. 

In particular, we argue that some prior definitions of neighbouring quantum states suffer from a poor tradeoff between privacy and accuracy, leading to a suboptimal scaling with respect to the number of qubits $n$. This is the case, for instance, if we require two neighbouring states to have bounded trace distance $\tau = \Theta(1)$. We also provide a similar result for the Wasserstein distance of order 1.

\subsection{Concentration inequalities for private measurements}
It's well known that noisy quantum algorithms suffer from severe limitations, that often hinder quantum advantage. Prior works \cite{Stilck_Fran_a_2021,de2023limitations} showed that, if the noise exceeds a given threshold, the output of noisy devices is concentrated around the maximally mixed state, and then it can be efficiently approximated with a classical computer.
Since quantum differential privacy involves the injection of noise, it's not surprising that similar concentration inequalities hold for quantum private algorithms.
In the remainder of this section, we will show how this concentration affects the accuracy of private measurements. For the sake of simplicity, we will state our results in terms of simple, local observables such as $O = \sum_{i=1}^n Z_i$. Similar results can be obtained for any observable with bounded Lipschitz constant, as also discussed in \cite{de2023limitations}, but our choice is sufficient to display the shortcomings of a poor choice of the neighbouring relationship.
If we measure $O$ on the maximally mixed state $\mathbb{1}/2^n$, the outcome satisfies a Gaussian concentration inequality \cite{de2023limitations}:
\begin{equation}
\label{eq:concentration}
    \Pr_{\mathbb{1}/2^n} (|O|\geq a n) \leq K e^{-a^2n},
\end{equation}
for $K=1$. So, if a state $\rho$ satisfies $D_\infty(\rho\|\mathbb{1}/2^n)\leq \epsilon$, the definition of the quantum max-relative entropy yields,
\begin{equation}
\label{eq:transferred}
    \Pr_{\rho} (|O|\geq a n) \leq e^\epsilon \Pr_{\mathbb{1}/2^n} (|O|\geq a n)  \leq K' e^{-a^2n},
\end{equation}
where $K'=e^{\epsilon}$.
For the sake of simplicity, throughout this section, we consider the special case of \emph{pure} differential privacy, i.e. $(\epsilon,0)$-DP, but our results can be suitably extended to the more general \emph{approximate} differential privacy, i.e. $(\epsilon,\delta)$-DP, under the assumption that $\delta \ll 1$.

Consider a quantum channel $\mathcal{A}(\cdot)$ and assume for the sake of simplicity that $\mathcal{A}$ is unital, i.e. $\mathcal{A}(\mathbb{1}) = \mathbb{1}$.  
We show that different neighbouring relationships $\overset{Q}{\sim}$ have a disparate impact on the accuracy. 
The first result is devoted to states with bounded trace distances.

\begin{theorem}[Concentration inequality for bounded trace distance]
\label{thm:trace}
Consider the observable $O =\sum_{i=1}^n Z_i$ and let $\mathcal{A}$ be a unital quantum channel satisfying $\epsilon$-DP with respect to $\tau$-neighbouring states, i.e. $ D_\infty(\mathcal{A}(\rho)\|\mathcal{A}(\sigma))\leq \epsilon$ if $\frac{1}{2}\|\rho-\sigma\|_1\leq \tau$.
Assume $\tau = \Theta(1)$. Then, for any input state $\rho$, the output $\mathcal{A}(\rho)$ satisfies the following concentration inequality:
\[
\Pr_{\mathcal{A}(\rho)} (|O|\geq a n) \leq K' e^{-a^2n},
\]
where $K'=e^{O(\epsilon)}$.
\end{theorem}
\begin{proof}
For two arbitrary quantum states, we have

\begin{equation}
\label{eq:trace}
    \forall \rho,\sigma : D_\infty(\mathcal{A}(\rho)\|\mathcal{A}(\sigma))\leq \epsilon/\tau.
\end{equation}
This can be seen by building the following chain :
\[
\rho_i = \rho \max(0,1-i\tau) + \sigma \min(1,i\tau)
\]
We note that $\frac{1}{2}\|\rho_i -\rho_{i+1}\|_1\leq \tau$ which implies $D_\infty(\mathcal{A}(\rho_i)\|\mathcal{A}(\rho_{i+1}))\leq \epsilon$. Then \eqref{eq:trace} can be deduced by iterating the triangle inequality. Combining it with \eqref{eq:transferred}, we obtain
\[
\forall \rho :\Pr_{\mathcal{A}(\rho)} (|O|\geq a n) \leq K e^{-a^2n},
\]
where $K= e^{\epsilon/\tau}=e^{O(\epsilon)}$.
\end{proof}
To showcase the implications of the \thmref{thm:trace}, we set $\tau = 0.1$ and we consider $\rho:= \ket{1^n}\bra{1^n}$. We remark that $\rho$ is an eigenvector of $O$, with eigenvalue $n$. However, instead of measuring $O$ directly, we can post-process $\rho$ with a $\epsilon$-DP channel $\mathcal{A}$ as defined in the statement of the theorem. Set $\epsilon=1$. In order to achieve an error smaller than, say, $0.5 n$, we need to ensure that the outcome is larger than $0.9n$. Then \thmref{thm:trace} implies that the error is larger than $0.5 n$ with high probability:
\[
\Pr_{\mathcal{A}(\rho)}(|n-O|\leq 0.5 n) =\Pr_{\mathcal{A}(\rho)}(O\geq 0.5 n) \leq \Pr_{\mathcal{A}(\rho)}(|O|\geq 0.5 n)\leq e^{10 - 0.25n}
\]
and hence setting $n=100$ we obtain
\[
\Pr_{\mathcal{A}(\rho)}(|n-O|\leq 0.5 n) \leq 3 \times 10^{-7}.
\]

Now, we provide a similar result for another neighbouring definition.
In \cite{wasserstein2021}, the authors extend the Wasserstein distance of order 1 (or  $W_1$ distance) to quantum states and suggest quantum differential privacy as a potential application of their work. 
Recall that the $W_1$ distance between the quantum states $\rho$ and $\sigma$ of $\mathcal{H}_n$ is defined as
\begin{align*}
    W_1(\rho,\sigma) = \min \bigg{(} \sum_{i=1}^n c_i : c_i \geq 0, \rho -\sigma = \sum_{i=1}^n c_i \left(\rho^{(i)}-\sigma^{(i)}\right), \\ \rho^{(i)}, \sigma^{(i)}\in \mathcal{S}_n, \Tr_i \rho^{(i)} = \Tr_i\sigma^{(i)}\bigg{)}. 
\end{align*}
The following theorem shows that the $W_1$ distance leads to the following undesired concentration inequality.

\begin{theorem}[Concentration inequality for bounded $W_1$ distance]
\label{thm:conc_w1}
Consider the observable $O =\sum_{i=1}^n Z_i$ and let $\mathcal{A}$ be a unital quantum channel satisfying $\epsilon$-DP with respect stated with $W_1$ distance bounded by 1, i.e. $ D_\infty(\mathcal{A}(\rho_1)\|\mathcal{A}(\rho_2))\leq \epsilon$ if $W_1(\rho_1,\rho_2)\leq 1$.    
Then, for any input state $\rho$, the output $\mathcal{A}(\rho)$ satisfies the following concentration inequality:
\[
\Pr_{\mathcal{A}(\rho)} (|O|\geq a n) \leq K' e^{-a^2n},
\]
where $K'=e^\epsilon(n -e^{-\epsilon}(n-1)) $.
\end{theorem}
\begin{proof}
Quantum differential privacy with respect to bounded Wasserstein distance of order 1 can be expressed as:
\[
W_1(\rho_1,\rho_2)\leq 1 \implies D_\infty(\mathcal{A}(\rho_1)\|\mathcal{A}(\rho_2))\leq \epsilon.
\]
We show that even this definition causes the output state to be highly concentrated around zero, independent of the input state.
In particular, we show that for two arbitrary quantum states $\rho$ and $\sigma$, we have
\begin{equation}
\label{eq:w1_conc}
   \forall \rho,\sigma : D_\infty(\mathcal{A}(\rho)\|\mathcal{A}(\sigma))\leq \epsilon', 
\end{equation}
where $\epsilon' = \epsilon + \log(n-ne^{-\epsilon} + e^{-\epsilon})$.
This can be seen considering the mixture $\rho':=\left(1-\frac{1}{n}\right)\rho + \frac{\sigma}{n}$ and noting that $W_1(\rho,\rho')\leq 1$.
Then, by the definition of $\epsilon$-differential privacy,
\begin{align*}
    \left(1-\frac{1}{n}\right)\Tr[M_m \mathcal{A}(\rho)]+ \frac{1}{n} \Tr[M_m \mathcal{A}(\sigma)]  \\=\Tr[M_m \mathcal{A}(\rho')] \leq e^\epsilon \Tr[M_m \mathcal{A}(\rho)]
\end{align*}
And thus
\begin{align*}
    \Tr[M_m \mathcal{A}(\sigma)] \leq e^\epsilon(n -e^{-\epsilon}(n-1)) \Tr[M_m \mathcal{A}(\rho)] \\= e^{\epsilon'} \Tr[M_m \mathcal{A}(\rho)],
\end{align*}
which implies \eqref{eq:w1_conc}.
Then, for any input $\rho$, $\mathcal{A}(\rho)$ is $\epsilon$-close to the maximally mixed state in quantum max-relative entropy, up to additive logarithmic factors. Applying \eqref{eq:transferred} yields
\[
\Pr_{\mathcal{A}(\rho)} (|O|\geq a n) \leq  K' e^{-a^2n},
\]
where where $K= e^{\epsilon'} = e^\epsilon(n -e^{-\epsilon}(n-1))$.
\end{proof}
Proceeding similarly as for the trace distance, set $\rho:=\ket{1^n}\bra{1^n}$ and $\epsilon =1$. \thmref{thm:conc_w1} implies that
\[
\Pr_{\mathcal{A}(\rho)}(|n-O|\leq 0.5 n) =\Pr_{\mathcal{A}(\rho)}(O\geq 0.5n) \leq \Pr_{\mathcal{A}(\rho)}(|O|\geq 0.5n)\leq (en -(n-1)) e^{-0.25 n}
\]
and hence setting $n=100$ we obtain
\[
\Pr_{\mathcal{A}(\rho)}(|n-O|\leq 0.5 n) \leq 2.4 \times 10^{-9}.
\]
Then the above example can be considered as a no-go result concerning $(\epsilon, 0)$-DP under Wasserstein distance of order $1$. 
We emphasise that the main argument of \thmref{thm:conc_w1} is based on the construction of a classical mixed state, and then it holds both for the classical and the quantum $W_1$ distance. On the other hand, one could define the neighbouring relationship solely on pure states and hence overcome our no-go result. However, it is not obvious whether this definition can lead to a good privacy-utility tradeoff. We leave this possibility as an open problem for future explorations.

We also remark that  $(0,\delta)$-DP under the $W_1$ distance is equivalent to  $(0,\delta)$-DP with respect to $(1,1)$-neighbouring quantum states. Assume that a channel $\mathcal{A}$ is $(0,\delta)$-DP with respect to $(1,1)$-neighbouring quantum states and let $M=(M_1,\dots,M_k)$ be a POVM measurement  
\[
\forall \rho_1 \overset{(1,1)}{\sim} \rho_2 \forall S \subseteq [k] \sum_{j\in S }\Tr\left[ M_j (\mathcal{A}(\rho_1) - \mathcal{A}(\rho_2))\right] \leq \delta.
\]
Then, 
\begin{align*}
    \sum_{j\in S }\Tr\left[ M_j (\mathcal{A}(\rho) - \mathcal{A}(\sigma))\right] \leq \sum_{j\in S }\sum_{i=1}^n c_i \Tr\left[ M_j (\mathcal{A}(\rho^{(i)}) - \mathcal{A}(\sigma^{(i)}))\right]\\
    = \sum_{i=1}^n c_i \sum_{j\in S }\Tr\left[ M_j (\mathcal{A}(\rho^{(i))} - \mathcal{A}(\sigma^{(i)}))\right] \leq \sum_{i=1}^n c_i \delta = W_1(\rho,\sigma) \delta,
\end{align*}

where the last inequality follows from $\rho^{(i)}\overset{(1,1)}{\sim}\sigma^{(i)}$. Since $(1,1)$-neighbouring states satisfies $W_1(\rho,\sigma)\leq 1$, the equivalence follows.

\subsection[A positive result for the generalised neighbouring relationship]{A positive result for $(\ell,\tau)$-neighbouring states}
We conclude this section with a positive result: adopting the definition introduced in \secref{sec:local}, we can privately sample from an observable that approximates $O = \sum_{i=1}^n Z_i$, with a small loss in accuracy.
We remark that the special case $\ell=\tau=1$ has already been studied in \cite{aaronson2019gentle}. 

\begin{theorem}[Efficient private measurement for $(\ell,\tau)$-neighbouring states]
\label{thm:efficient_local}
Let $\mathcal{O}$ be the quantum to classical channel implementing a measurement of the observable $O =\sum_{i=1}^n Z_i$. 
Assume that a state $\rho$ satisfies
\[
\Pr_\rho[ |O - \langle O \rangle_\rho | > a]\leq b.
\]
and let $\alpha:=\frac{2\ell}{\log((e^\epsilon-1)\tau^{-1}+1)}\approx {2\ell \tau}{\epsilon^{-1}}.$
Then there exits a a quantum-to-classical channel $\mathcal{O}_\epsilon$ such that:
\begin{enumerate}
    \item $\mathcal{O}_\epsilon$ is $\epsilon$-DP with respect to $(\ell,\tau)$-neighbouring states.
    \item The following concentration inequality holds:
    \[
        \Pr[ |\mathcal{O}_\epsilon(\rho) - \langle O \rangle_\rho | > a+t\alpha]\leq b + e^{-t}.
    \]
\end{enumerate}
\end{theorem}
\begin{proof}
Let $\Lambda_\mathcal{L}$ be the Laplace noise of magnitude $\alpha$.
The first part of the theorem follows directly from \thmref{thm:cl_noise}, by choosing $\mathcal{O}_\epsilon = \Lambda_\mathcal{L} \circ \mathcal{O}$.
Moreover, if $Y\sim \text{Lap}(\alpha)$, then \[
\Pr[|Y|> t\cdot \alpha] = e^{-t}.
\]
Define the event $E = \{|Y|\leq t\alpha\}$. Then we have
\begin{align*}
   \Pr[ |\mathcal{O}_\epsilon(\rho) - \langle O \rangle_\rho | > a+t\alpha] 
   \\ \leq \Pr[ |\mathcal{O}_\epsilon(\rho) - \langle O \rangle_\rho | > a+t\alpha |E] \Pr[E] + \Pr[ |\mathcal{O}_\epsilon(\rho) - \langle O \rangle_\rho | > a+t\alpha |\overline{E}] \Pr[\overline{E}] 
   \\ \leq \Pr_\rho[ |O - \langle O \rangle_\rho | > a] + \Pr[|Y|> t\cdot \alpha] \leq b + e^{-t}.
\end{align*}

\end{proof}

So, in particular,  $\rho=\ket{1^n}\bra{1^n}$, we have that $\Pr_\rho[ O = \langle O \rangle_\rho  ]=1$ since $\rho$ is an eigenvector of $O$.
Then \thmref{thm:efficient_local} yields
\[
\Pr[ |\mathcal{O}_\mathcal{L}(\rho) - n | < t\cdot \alpha]\geq 1- e^{-t}.
\]
Finally, we plot the upper bounds derived in this section in \figref{fig:max-div} and \figref{fig:error-privacy}.

\begin{figure}
    \centering
    \includegraphics[width=10cm]{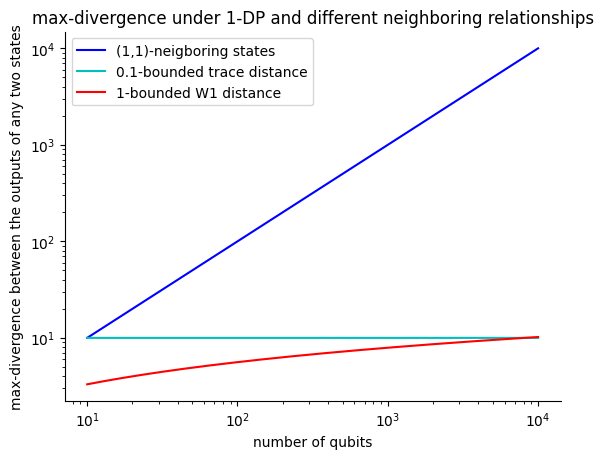}
    \caption{Upper bounds on the quantum max-relative entropy between any two states under 1-DP for several neighbouring relationships and various values of $n$.  }
    \label{fig:max-div}
\end{figure}
\begin{figure}
    \centering
    \includegraphics[width=10cm]{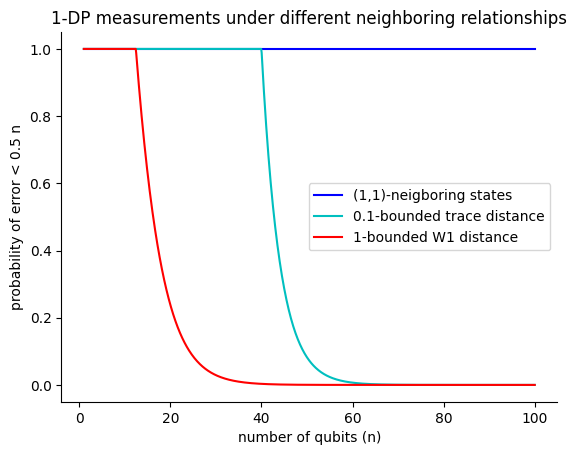}
    \caption{Upper bounds on the probability of achieving error lower than $0.5 n$ for a measurement of $\frac{1}{n}\sum_{i=1}^nZ_i$ on the state $\ket{1^n}$, for several neighbouring relationships and various values of $n$. We assumed the input state undergoes a 1-DP channel.}
    \label{fig:error-privacy}
\end{figure}
\section{Privacy-preserving estimation of expected values}
\label{sec:multiple}
In this section, we provide differentially private mechanisms for estimating the expected values of observables given $m$ copies of a quantum state. Despite their similarities, performing private measurements on a single state and privately estimating the expected value of these measurements given many copies are inherently different tasks.
In principle, we could perform an $\epsilon$-DP measurement on each copy and then average the results. Then the overall algorithm satisfies $(\epsilon',\delta')$-DP with $\epsilon'\approx \epsilon\sqrt{m \log(1/\delta')}$ by advanced composition (Theorem 6 in \cite{quantumDP}). 

However, this approach is highly suboptimal as the privacy loss (i.e. the parameter $\epsilon$) grows as $\sqrt{m}$. We present here a simpler and more efficient approach based on the concentration of measure, whose privacy loss decreases as $m$ increases.
Given an observable $O$ and set of quantum states equipped with a relationship denoted as $\overset{Q}{\sim}$, we'll define the \emph{average quantum sensitivity} of $O$ as follows:

\[
\overline{\Delta}(O) = \max_{\rho \overset{Q}{\sim} \sigma} \Tr\{O(\rho-\sigma)\}.
\]
Notably, we will present a simple technique whose privacy loss is proportional to $\overline{\Delta}(O) + \sqrt{1/m}$.
This newly defined quantity is closely related to other notions introduced in prior work.
Remark that the Lipschitz constant \cite{wasserstein2021} can be recovered as a special case by considering as $Q$-neighbouring the states with $W_1$ distance at most one, i.e. $\rho \overset{Q}{\sim}\sigma \iff W_1(\rho,\sigma)=1$.
Moreover, if a quantum encoding $\rho(\cdot)$ is $Q$-neighbouring-preserving, then the above can be related to the classical definition of sensitivity introduced in \eqref{eq:classical-sens}. Consider the function $f(x) = \Tr\{O\rho(x)\}$, then
\[
\Delta_f = \max_{x\sim x'} |f(x)-f(x')|\leq \max_{\rho(x)\overset{Q}{\sim}\rho(x')} |\Tr(O(\rho(x)-\rho(x'))|\leq\max_{\rho\overset{Q}{\sim}\sigma} |\Tr(O(\rho-\sigma))|:= \overline{\Delta}(O)
\]

We now prove that there exists a simple differentially private algorithm consisting of measurements and classical post-processing that gives a suitable tradeoff between sensitivity and privacy.
We first consider a general post-processing channel and then we provide more concrete bounds for the Laplace and Gaussian noises.
\begin{theorem}
\label{thm:avg}
Consider a neighbouring relationship $\overset{Q}{\sim}$ over the set of quantum states $\mathcal{S}_n$. Let $\rho^{\otimes m}$ be a collection of $m$ copies of a quantum state $\rho\in S_n$ and $O$ an observable. Let $\Lambda(\cdot)$ be a classical channel with the following property. For $\delta'\in(0,1]$ and $x,x'\in \mathbb{R}$,
\[
|x-x'|\leq \overline{\Delta}(O) + \sqrt{m^{-1}\log(4/\delta')} \implies E_{e^\epsilon}(\Lambda(x)\|\Lambda(x'))\leq \delta.
\]
Consider the following algorithm $\mathcal{A}$:
\begin{enumerate}
    \item Measure $O$ on each copy of $\rho$ and collect the outcomes $y_1,\dots,y_m$.
    \item Compute the average $\hat{\mu}=\frac{1}{m}\sum_{i=1}^m y_i$ and output $\Lambda(\hat{\mu})$.
\end{enumerate}
Then the algorithm $\mathcal{A}$ is $(\epsilon,\delta+\delta')$-DP.
\end{theorem}
\begin{proof}
Consider two neighbouring quantum states $\rho\overset{Q}{\sim}\sigma$. For $X\in\{\rho,\sigma\}$, let $\hat{\mu}_X$ the average obtained on input $X^{\otimes m}$.
By Chernoff-Hoeffding's bound,
\[
\Pr\left[\left|\hat{\mu}_X- \Tr[OX]\right| \geq \frac{t}{2}\right]\leq 2e^{-mt^2}.
\]
Hence, by union bound,
\[
\Pr[E]\leq\delta':= 4e^{-mt^2},
\]
where $E$ is the following event:
\[
E:=\left\{\left(\left|\hat{\mu}_\rho- \Tr[O\rho]\right| \geq \frac{t}{2}\right)\vee  \left(|\hat{\mu}_\sigma- \Tr[O\sigma]| \geq \frac{t}{2}\right)\right\}.
\]
Conditioning on the complementary event $\overline{E}$ and observing that $t=\sqrt{m^{-1}\log(4/\delta')}$, we have,
\begin{align*}
    |\hat{\mu}_\rho -\hat{\mu}_\sigma| \leq |\hat{\mu}_\rho- \Tr[O\rho]| + |\Tr[O\rho]-\Tr[O\sigma]|+ |\Tr[O\sigma]-\hat{\mu}_\sigma| 
    \\\leq \Delta + t = \Delta + \sqrt{m^{-1}\log(4/\delta')} .
\end{align*}
This implies that, conditioning on $\overline{E}$,
\[
E_{e^\epsilon}(\Lambda(\hat{\mu}_\rho)\|\Lambda(\hat{\mu}_\sigma))\leq \delta,
\]
equivalently, we have
\[
\forall S:\Pr[\Lambda(\hat{\mu}_\rho) \in S|\overline{E}]\leq e^\epsilon\Pr[\Lambda(\hat{\mu}_\sigma) \in S|\overline{E}] + \delta.
\]
Then we also have that, for all $S$
\begin{align*}
    \Pr[\Lambda(\hat{\mu}_\rho) \in S] = \Pr[\Lambda(\hat{\mu}_\rho) \in S|{E}] \Pr[E]+ \Pr[\Lambda(\hat{\mu}_\rho) \in S|\overline{E}]\Pr[\overline{E}] 
\\ \leq \Pr[\Lambda(\hat{\mu}_\rho) \in S|\overline{E}] + \delta'  \leq e^\epsilon\Pr[\Lambda(\hat{\mu}_\sigma) \in S|\overline{E}] + \delta + \delta'
\\ \leq  e^\epsilon\Pr[\Lambda(\hat{\mu}_\sigma) \in S] + \delta + \delta'.
\end{align*}
\end{proof}
Finally, plugging the Laplace and the Gaussian channels in \thmref{thm:avg}, we obtain the following corollary.
\begin{cor}
\label{cor:avg}
Let $\mathcal{A}, \rho^{\otimes m}$ and $O$ as in \thmref{thm:avg} and let $\Delta:=\overline{\Delta}(O)$. The following privacy guarantees hold.
\begin{itemize}
    \item (Laplace noise) Let $\Lambda_{\mathcal{L},b}$ the Laplace channel of scale $b:=(\Delta+ \sqrt{m^{-1}\log(4/\delta')})/\epsilon$. Then the algorithm $\mathcal{A}$ is $(\epsilon,\delta')$-DP.
    \item (Gaussian noise) Let $\Lambda_{\mathcal{G}, \sigma}$ the Gaussian channel of variance 
    \[\sigma^2 \geq 2 \log(1.25/\delta)(\Delta +\sqrt{m^{-1}\log(4/\delta'})^2/\epsilon^2.\] 
    Then the algorithm $\mathcal{A}$ is 
    $(\epsilon,\delta + \delta')$-DP.
\end{itemize}
\end{cor}

\subsection{Bounding the average quantum sensitivity}
\label{sec:sensitivity}
Here we provide several bounds for the quantum sensitivity based on different neighbouring relationships.
The first bound is based on H\"older's inequality, i.e. $|\Tr(LR)|\leq \|L\|_p \|R\|_q$ for $p^{-1} + q^{-1} = 1$, where $\|\cdot\|_p$ is the Schatten $p$-norm.
Say that $\rho\overset{Q}{\sim}\sigma$ if $\|\rho-\sigma\|_p \leq \tau$. Then applying H\"older's inequality yields
\[
\Delta(O) \leq \|O\|_q \tau.
\]
For the special case of $p=1$  (which corresponds to the trace distance) a stronger bound holds:
\[
\Delta(O) = \max_{\rho,\sigma : \|\rho-\sigma\|_1\leq\tau} \Tr[O(\rho-\sigma)] \leq  \frac{1}{2}\|O\|_\infty \|\rho-\sigma\|_1 \leq \frac{\tau}{2} \|O\|_\infty.
\]

We can also consider a neighbouring relationship based on the Wasserstein distance of order 1, i.e. $\rho \overset{Q}{\sim} \sigma$ if $W_1(\rho,\sigma)\leq \tau$. Then the quantum sensitivity is proportional to the Lipschitz constant.
\[
\Delta(O) = \max_{\rho,\sigma : W_1(\rho,\sigma)\leq \tau} \Tr\{O(\rho-\sigma)\} \leq \|O\|_{Lip} \tau.
\]
By \lemref{lem:w1}, we also have that if $\rho\overset{(\Xi,\tau)}{\sim}\sigma$, then $W_1(\rho,\sigma) \leq \frac{3}{2}\max_{\mathcal{I}\in\Xi}|\mathcal{I}|\tau$. This implies
\[
\Delta(O) = \max_{\rho,\sigma : \rho\overset{(\Xi,\tau)}{\sim}\sigma} \Tr\{O(\rho-\sigma)\} \leq \frac{3}{2}\|O\|_{Lip}  \max_{\mathcal{I}\in\Xi}|\mathcal{I}|\tau.
\]
The above bounds for $\Delta(O)$ are listed concisely in \tableref{table:sensitivity}.

\begin{table}[t]
\caption{Here we summarize the results of \secref{sec:sensitivity}. For each neighbouring relationship over quantum states, we list the corresponding average quantum sensitivity $\Delta(O)$ of an observable $O$. }
\label{table:sensitivity}
\vskip 0.15in
\begin{center}
\begin{small}
\begin{sc}
\begin{tabular}{lcccr}
\toprule
$\rho\overset{Q}{\sim}\sigma$ & $\Delta(O)$ \\
\midrule

$\|\rho -\sigma\|_p \leq\tau$ & $\tau \|O\|_q$ \\
$\frac{1}{2}\|\rho-\sigma\|_1\leq\tau$ & $\tau \|O\|_1$ \\
$W_1(\rho,\sigma)\leq \tau$    & $\|O\|_{Lip}\tau $ \\
$\rho \overset{(\Xi,\tau)}{\sim}\sigma$ & $\min\{\frac{3}{2}\|O\|_{Lip}\max_{\mathcal{I}\in\Xi}|\mathcal{I}|\tau, \|O\|_{Lip} n\tau\} $ \\

\bottomrule
\end{tabular}
\end{sc}
\end{small}
\end{center}
\vskip -0.1in
\end{table}
\section{Private quantum machine learning}\label{sec:applications}
In this section, we demonstrate the applications of the results and tools we derived so far to variational quantum algorithms for machine learning.
Let $\rho(\boldsymbol{\theta}; \boldsymbol{x})$ be the output of a variational quantum circuit. We will assume that the parameters $\boldsymbol{\theta}$ are trained using a suitable (classical) dataset $S = (\boldsymbol{s^{(1)}},\dots,\boldsymbol{s^{(m)}})$. Given a test set $\mathcal{X}$, we're asked to approximate a function $f: \mathbb{R}^d \rightarrow \mathbb{R}$.
Thus, we can use variational quantum algorithms to find a set of parameters $\boldsymbol{\theta}$ that satisfy
\[
\forall \boldsymbol{x}\in\mathcal{X}:f(\boldsymbol{x}) \simeq \Tr (O \rho(\boldsymbol{\theta}; \boldsymbol{x})),
\] 
where $O$ is a suitable observable.
Given this simple scenario, differential privacy can come in different flavours.
\begin{itemize}
    \item Let $\boldsymbol{x}=({x}_1,\dots,{x}_d)\in\mathcal{X}$ be the input vector. Given a neighbouring relationship $\boldsymbol{x}\sim \boldsymbol{x'}$, we can ensure differential privacy with respect to the input $\boldsymbol{x}$. This is particularly useful when $\boldsymbol{x}$ contains the sensitive information of multiple individuals or when $\boldsymbol{x}$ might be corrupted by an \emph{adversarial attack}.
    \item In the alternative, we can require differential privacy with respect to the training set $S=(\boldsymbol{s}^{(1)},\dots,\boldsymbol{s}^{(m)})$, where $S\sim S'$ if they differ only in a single entry $\boldsymbol{s}^{(j)}$. This notion of privacy is meant to protect the sensitive information of the individuals who compose the training set. Furthermore, it also enhances \emph{generalisation}, i.e. it allows to upper bound of the discrepancy between the error on the training set and the generalisation error.
\end{itemize}

\subsection{Private evaluation with respect to the input \textit{x}}
Given a suitable notion of neighbouring inputs $\boldsymbol{x} \sim \boldsymbol{x'}$, we want to find a neighbouring relationship over quantum states $\overset{Q}{\sim}$ such that $\rho(\cdot,\boldsymbol{\theta})$ is $Q$-neighbouring-preserving. In other terms, we need to ensure that
\[
\boldsymbol{x} \sim \boldsymbol{x'} \implies \rho(\boldsymbol{x}, \boldsymbol{\theta}) \overset{Q}{\sim} \rho(\boldsymbol{x'}, \boldsymbol{\theta}).
\]
First, we select the relationship $Q$ according to \tableref{tab:enc}.
If a single copy of $\rho(\boldsymbol{x}, \boldsymbol{\theta}) $ is available, we can make the measurement differentially private either by adding a final quantum noisy channel (\thmref{thm:local_noisy}) or by classical post-processing (\thmref{thm:cl_noise_local}). 
If, instead, we're able to prepare multiple copies of $\rho(\boldsymbol{x}, \boldsymbol{\theta})$, it's convenient to post-process the average outcome with classical noise. Then differential privacy is guaranteed by \corref{cor:avg}.

\subsubsection*{Certified adversarial robustness}  
\label{sec:rob}
Now, we outline the connection between differential privacy and adversarial robustness, which has been previously established in  \cite{lecuyer2019certified} and extended to the quantum setting in \cite{liu2021, hirche2023, huang2023certified}. 
We consider a slightly different setting, known as $k$-\emph{class classification}, where a classification algorithm $\mathcal{A}$ outputs a label $y \in [k]$ on input $\boldsymbol{x}$. For instance, for $k=2$, we can consider an algorithm that outputs label $1$ if $\boldsymbol{x}$ represents a dog and $2$ if $\boldsymbol{x}$ represents a cat.
Consider $k$ observables $O_1,\dots, O_{k}$, and assume, for simplicity, that their spectrum lies in $[0,1]$. The algorithm $\mathcal{A}$ works as follows.
\begin{enumerate}
    \item On input $\boldsymbol{x}$, for each $i \in [k]$, the algorithm measures the observable $O_i$ on the state $\rho(\boldsymbol{x},\boldsymbol{\theta})$ $m$ times and stores the outcomes in $y^{(i)}_1,\dots,y^{(i)}_m$.
    \item For each $i \in[k]$, let $y^{(i)} = \sum_{j=1}^m y_j^{(i)}$.
    \item $\mathcal{A}$ returns the index $i^* \in [k]$ such that $i^* = \arg\max y^{(i)} $.
\end{enumerate}

We adopt Proposition 1 in \cite{lecuyer2019certified} to the quantum setting.
\begin{prop}[Robustness condition]
\label{prop:rob-condition}
Let $\beta \in (0,1]$. Let $\rho(\cdot,\boldsymbol{\theta})$ be $Q$-neighbouring-preserving  and assume that each of the $m$ measurements in step (1) satisfies $(\epsilon,\delta)$-DP with respect to $Q$-neighbouring quantum states. 
For any input $\boldsymbol{x}$, if for some $i \in [k]$,
\begin{equation}
    \label{eq:rob-condition}
    y^{(i)} > e^{2\epsilon} \max_{j\neq i } y^{(j)} + (1+e^\epsilon)\delta + \sqrt{\frac{2}{m}\log\left(\frac{4k}{\beta}\right)},
\end{equation}
then the algorithm $\mathcal{A}$ satisfies, for all $\boldsymbol{x}\sim \boldsymbol{x'}$
\[
 \Pr [\mathcal{A}(\boldsymbol{x}) = \mathcal{A}(\boldsymbol{x'})] \geq 1- \beta.
\]
In this case, we say that the classifier $\mathcal{A}$ is $\beta$-robust to adversarial attacks.
\end{prop}

\begin{proof}
Let $\boldsymbol{x}\sim\boldsymbol{x'}$. Since $\rho(\cdot,\boldsymbol{\theta})$ is $Q$-neighbouring-preserving, $\rho(\boldsymbol{x},\boldsymbol{\theta})\overset{Q}{\sim}\rho(\boldsymbol{x'},\boldsymbol{\theta})$.
The assumption that each measurement satisfies $(\epsilon,\delta)$-DP implies
\[
    \forall i \in [k], \forall F \subseteq \mathsf{range}(O_i) : \Pr_{\rho(\boldsymbol{x},\boldsymbol{\theta})}[O_i \in F] \leq e^\epsilon \Pr_{\rho(\boldsymbol{x'},\boldsymbol{\theta})}[O_i \in F] +\delta.
\]
We first need to prove the following inequality.  For all $i$,
\begin{equation}
    \label{eq:conc-rob}
    \Tr[O_i \rho(\boldsymbol{x},\boldsymbol{\theta})] \leq e^\epsilon \Tr[O_i \rho(\boldsymbol{x'},\boldsymbol{\theta})] + \delta.
\end{equation}
Recall that the expectation of a non-negative random variable $X$ can be expressed as 
\[\mathbb{E}(X) = \int_{t\geq0} \Pr[X > t]dt.\]
Combining this with differential privacy, we obtain
\begin{align*}
   \Tr[O_i \rho(\boldsymbol{x},\boldsymbol{\theta})] = \int_{t\geq0} \Pr_{\rho(\boldsymbol{x},\boldsymbol{\theta})}[O_i > t]dt 
   \\\leq e^\epsilon\int_{t\geq0} \Pr_{\rho(\boldsymbol{x'},\boldsymbol{\theta})}[O_i > t]dt +\delta = e^\epsilon \Tr[O_i \rho(\boldsymbol{x},\boldsymbol{\theta})] + \delta,
\end{align*}
which proves \eqref{eq:conc-rob}.
It remains to show that the discrepancy between $y^{(i)}=\frac{1}{m}\sum_{j=1}^m y_j^{(i)}$ and of $\Tr[O_i \rho(\boldsymbol{x},\boldsymbol{\theta})]$ is small enough with high probability. To this end, we can use concentration of measure. 
By Chernoff-Hoeffding's bound,
\[
\Pr\left[\left|\frac{1}{m}\sum_{j=1}^m y_j^{(i)}- \Tr[O_i \rho(\boldsymbol{x},\boldsymbol{\theta})] \right| \geq {t}\right]\leq 2e^{-2mt^2}.
\]
and thus $y^{(i)} = \Tr[O_i \rho(\boldsymbol{x},\boldsymbol{\theta})] \pm {t}$ with probability at least $1-2e^{-2mt^2}$. Denote by $\widetilde{y}^{(1)},\dots, \widetilde{y}^{(k)}$ the average of the measurements on the state 
$\rho(\boldsymbol{x'},\boldsymbol{\theta})$.
By union bound, with probability at least $1-4ke^{-2mt^2} = 1- \beta$ we have that
\begin{equation}
\label{eq:good-event}
    \forall i \in[k]: \bigg(y^{(i)} = \Tr[O_i \rho(\boldsymbol{x},\boldsymbol{\theta})] \pm {t} \bigg) \wedge \bigg(\widetilde{y}^{(i)} = \Tr[O_i \rho(\boldsymbol{x'},\boldsymbol{\theta})] \pm {t}\bigg).
\end{equation}
Assume, by contradiction, that $\mathcal{A}(x) \neq \mathcal{A}(x')$ and \eqref{eq:good-event} hold simultaneously
Since $\mathcal{A}(x) \neq \mathcal{A}(x')$, there exists $i\neq i'$ such that
\[
y^{(i)} > \max_{j\neq i} y^{(j)} \;\text{ and } \; \widetilde{y}^{(i')} > \max_{j\neq i} \widetilde{y}^{(j)}.
\]
Putting them all together, we have
\begin{align*}
    \widetilde{y}^{(i)} \geq \Tr[O_i \rho(\boldsymbol{x'},\boldsymbol{\theta})] - t \geq  e^{-\epsilon} (\Tr[O_i \rho(\boldsymbol{x},\boldsymbol{\theta})] - t) - e^{-\epsilon}\delta
    \\ \geq  e^{-\epsilon} ({y}^{(i)} - 2t) - e^{-\epsilon}\delta > \max_{j \neq i} e^{\epsilon }y^{(j)} + \delta 
    \\ \geq \max_{j \neq i}\widetilde{y}^{(j)}\geq \widetilde{y}^{(i')}
\end{align*}
Thus we obtained $\widetilde{y}^{(i)} > \widetilde{y}^{(i')}$ contradicting the assumptions $\mathcal{A}(x) \neq \mathcal{A}(x')$. 
This proves that $\mathcal{A}(x) = \mathcal{A}(x')$ with probability at least $1-\beta$.
\end{proof}

It's easy to see how the above proposition is related to adversarial attacks. Assume that an adversary has the capabilities of tampering with the input $\boldsymbol{x}$ by replacing it with $\boldsymbol{x'}$ such that $\boldsymbol{x}\sim \boldsymbol{x'}$. 
We remark that there's no unique way of choosing the neighbouring relationship in this context, as it is closely related to the capabilities of the adversary.  
Under the same assumptions of \propref{prop:rob-condition}, the adversarial attack doesn't alter the output with high probability. 
The condition expressed in \eqref{eq:rob-condition} can be interpreted as the classifier being ``fairly confident'' about its prediction.
We also remark that \propref{prop:rob-condition} can be applied to virtually any algorithm $\mathcal{A}$, even in the absence of an explicit private mechanism, since all algorithms are by default $(0,\tau)$-DP with respect to $\tau$-neighbouring states. This can be easily checked from the properties of the trace distance.

Following \cite{lecuyer2019certified}, given a distribution $\mathcal{D}$ over labeled inputs of the form $(\boldsymbol{x},f(\boldsymbol{x}))$, we can define the \emph{certified accuracy} $\mathcal{R}(\mathcal{A})$ of an $(\epsilon,\delta)$-DP algorithm $\mathcal{A}$ as follows
\[
\mathcal{R}(\mathcal{A}):= \Pr_{(\boldsymbol{x},f(\boldsymbol{x}))\sim\mathcal{D}}\left[ \left(i^* = f(\boldsymbol{x})\right) \wedge \left(\delta < \frac{y^{(i^{*})} - e^{2\epsilon}\max_{j\neq i^*}y^{(j)} - g(k,\beta,m)}{1+ e^\epsilon}\right)\right],
\]
where $g(k,\beta,m) := \sqrt{{2}{m^{-1}}\log\left({4k}/{\beta}\right)}$ and $i^* = \arg\max y^{(i)}$.
In other terms, $\mathcal{R}$ is a lower bound on the probability that an instance is classified correctly and the classification is $\beta$-robust to adversarial attacks.
We remark that $\mathcal{R}$ can be easily estimated by computing the fraction of the test set that is classified correctly and, simultaneously, satisfies \eqref{eq:rob-condition}.

\paragraph{Numerical results.}
Finally, we complement our theoretical analysis with a numerical simulation implemented in PennyLane. 
We consider a classification task based on the first two classes of the famous IRIS dataset and each input $\boldsymbol{x} = (x_1,x_2,x_3,x_4)$ is susceptible to be perturbed by an adversarial attack. We assume that the adversary can select a single entry $x_i$ and map it to $x_i'$ with $|x_i - x_i'| \leq \tau$, for some threshold $0\leq \tau \leq 1$.
We trained a simple $4$-qubit binary classifier, based on the variational circuit depicted in \figref{fig:circuit}, whose gates are parametrised by a trainable vector $\boldsymbol{\theta}$ and the input vector $\boldsymbol{x}$. Hence, the output is measured according to  $O= \frac{1}{8}\sum_{i=1}^4 (Z_i +1)$ and the classifier outputs $0$ if the outcome is larger than $0.5$ and $1$ otherwise.  
It's easy to see that this encoding is $(1,\tau)$-privacy-preserving with respect to the neighbouring definition induced by the adversarial attack. The circuit is ended by a final layer of local depolarising noise $\mathcal{N}_p^{\otimes n}$, which ensures $(\epsilon,\delta_1)$-differential privacy with respect to $(1,\tau)$-neighbouring states, with $\delta_1$ defined as in \thmref{thm:local_noisy}.  We trained the model with the \texttt{Adam} optimiser \cite{kingma2014adam} with several noise levels $p$  and then we used the test set to estimate the certified accuracy for each $p$, and we plotted it against the threshold $\tau$ in \figref{fig:rob_plot}.
The results show that the noise level should be set according to attack threshold $\tau$, as for $\tau\leq 0.2$ the circuit with $p=0.1$ outperforms the others, while for $\tau \geq 0.2$ the circuit with $p=0.3$ achieves the best certified accuracy. 

Our simulation differs from previous experiments in multiple ways. First, we remark that our simulation combines local noisy channels with the novel neighbouring relationship we introduced in the present paper. In contrast to this, the simulation in \cite{liu2021} is based on $\tau$-neighbouring states and ensures privacy via multiple layers of \emph{global} depolarizing noise. 
On the other hand, \cite{huang2023certified} combines local noisy channels with $\tau$-neighbouring states, resulting in privacy guarantees that degrade exponentially fast as the number of qubits increases. This stems from the fact that in Lemma 3 in \cite{huang2023certified}, the authors show quantum differential privacy with $\epsilon = \log (1 + \tau/p^n) \simeq \tau/ p^n$.
In addition, both \cite{liu2021} and \cite{huang2023certified} are based on $\epsilon$-differential privacy while \propref{prop:rob-condition} is stated in terms of $(\epsilon,\delta)$-differential privacy.
This is particularly useful to assess the certified accuracy of various noise regimes, including the case with no noise at all ($p=0$).

\begin{figure}
    \centering
    \begin{quantikz}[row sep={12mm,between origins}]
  \lstick{\ket{0}} &\gate{R\left(\boldsymbol{\theta_1^{(1)}}\right)}  &\ctrl{1} &  \qw & \targ{} & \gate{R\left(\boldsymbol{\theta_1^{(2)}}\right)}  &\ctrl{2} &  \qw & \targ{} & \qw & \gate{R_x\left({x_1}\right)} & \gate{\mathcal{N}_p} &\meter{Z}\\
  \lstick{\ket{0}} &\gate{R\left(\boldsymbol{\theta_2^{(1)}}\right)}  &\targ{}  & \ctrl{1} & \qw & 
  \gate{R\left(\boldsymbol{\theta_2^{(2)}}\right)} & \qw  &\ctrl{2}  & \qw & \targ{} &\gate{R_x\left({x_2}\right)} &\gate{\mathcal{N}_p} &\meter{Z}\\
  \lstick{\ket{0}} &\gate{R\left(\boldsymbol{\theta_3^{(1)}}\right)} & \ctrl{1}  & \targ{} & \qw &
  \gate{R\left(\boldsymbol{\theta_3^{(2)}}\right)}  & \targ{} & \qw & \ctrl{-2} & \qw{}&\gate{R_x\left({x_3}\right)} &\gate{\mathcal{N}_p} &\meter{Z}\\
  \lstick{\ket{0}} &\gate{R\left(\boldsymbol{\theta_4^{(1)}}\right)} &\targ{}  & \qw & \ctrl{-3} &
  \gate{R\left(\boldsymbol{\theta_4^{(2)}}\right)} &\qw  & \targ{} & \qw & \ctrl{-2} &\gate{R_x\left({x_4}\right)}&\gate{\mathcal{N}_p} &\meter{Z}\\ 
\end{quantikz}
    \caption{The parametric quantum circuit used in the simulation. We placed the encoding gates after the trainable gates in order to produce a $(1,\tau)$-neighbouring-preserving encoding. The output state is measured according to the observable $O= \frac{1}{8}\sum_{i=1}^4 (Z_i +1)$.}
    \label{fig:circuit}
\end{figure}
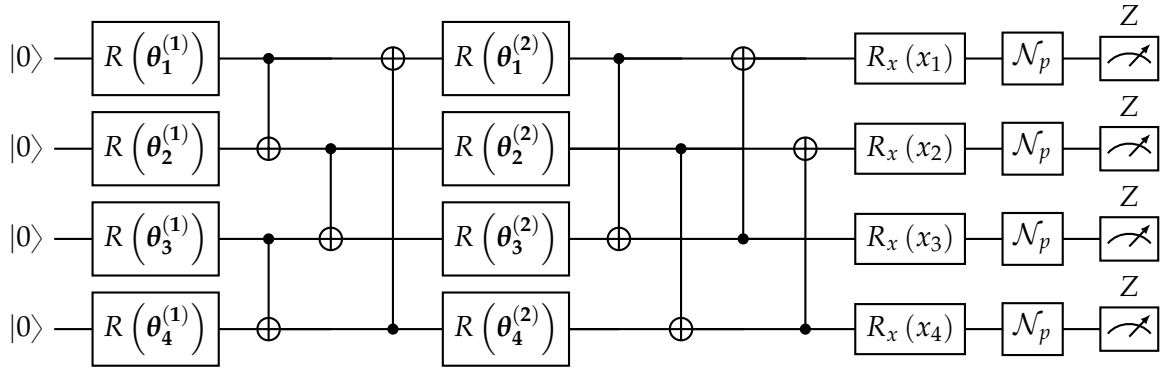

\begin{figure}
    \centering
    \includegraphics[width = 10 cm]{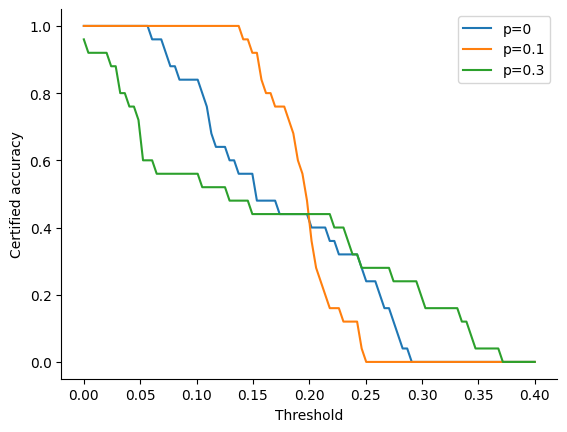}
    \caption{This plot contains the values of the certified accuracy estimated for various noise levels $p$ and various attack thresholds $\tau$. }
    \label{fig:rob_plot}
\end{figure}

\subsection{Private prediction with respect to the training set \em{S}}
Training a variational quantum algorithm involves finding a set of parameters $\boldsymbol{\theta^*}$ that minimizes a loss function $\mathcal{L}(\boldsymbol{\theta},S) = \frac{1}{m}\sum_{i=1}^m \Tr\{O(y_i) \rho(\boldsymbol{\theta}; \boldsymbol{x_i})\}= \frac{1}{m}\sum_{i=1}^m \ell(\boldsymbol{\theta},\boldsymbol{s_i})$ with respect to a given training set $S = (\boldsymbol{s_1},\dots,\boldsymbol{s_m})$ where $\boldsymbol{s_i} = (\boldsymbol{x_i},y_i)$.
In this setting, we let $S$ and $S'$ be neighbouring if $\exists i \in [m], \forall j\neq i:\boldsymbol{s_j} = \boldsymbol{s_j}$, i.e. if they differ in at most one element.
Despite the existence of quantum algorithms for optimising a loss function, they're often not suitable for near-term devices. In most near-term applications, a variational quantum circuit is paired with a classical optimiser. Thus, standard techniques for differentially private (classical) optimisation can be adapted \cite{bassily2014private,abadi2016}. For instance, \citet{watkins2023quantum} implements the algorithm for private stochastic gradient descent (SGD) provided in \cite{abadi2016} to optimize the parameters of a variational quantum circuit, achieving good empirical performance.
The technique provided in \cite{abadi2016} involves a procedure known as \emph{gradient clipping}, which consists in rescaling the gradient $\nabla_{\boldsymbol{\theta}}\ell(\boldsymbol{\theta},\boldsymbol{s}_i)$ to ensure that its $\ell_2$ norm is bounded by a suitable constant $C$, i.e. $\|\nabla_{\boldsymbol{\theta}}\ell(\boldsymbol{\theta},\boldsymbol{s}_i)\|_2 \leq C$. Then, privacy is ensured by the addition of Gaussian noise with variance proportional to $C^2$ on each estimate of the gradient.
Instead of clipping the gradient, alternative techniques such as \cite{bassily2014private}, estimates an upper bounds $UB$, where
\[
\forall\boldsymbol{\theta}: \|\nabla_{\boldsymbol{\theta}} \ell(\boldsymbol{\theta},\boldsymbol{s_i})\|_{2} \leq UB.
\]
and add Gaussian noise proportional to $UB^2$ on each estimate of the gradient.

Here we show that $UB$ can be easily estimated for some classes of variational quantum circuits. Assuming $\ell$ is differentiable with respect to $\boldsymbol{\theta}$ we have
\[|\ell(\boldsymbol{\theta},\boldsymbol{s_i}) - \ell(\boldsymbol{\theta'},\boldsymbol{s_i})| \leq UB\|\boldsymbol{\theta}-\boldsymbol{\theta'}\|_{\ell_2}
\implies \|\nabla_{\boldsymbol{\theta}} \ell(\boldsymbol{\theta},\boldsymbol{s_i})\|_{2} \leq UB.
\]

For $\boldsymbol{\theta} = (\theta_1,\dots,\theta_d)$, assume that each coordinate $\theta_j$ is encoded via a single gate Hamiltonian encoding, i.e. $e^{-i \theta_jH_j}$ with $\|H_i\|_2\leq 1$ . Moreover, assume that the output state is produced by a $1D$ circuit with bounded depth $L$ (and thus the light-cone of each single qubit gate is upper bounded by $2L$).
As shown in \appref{app:encodings}, the Hamiltonian encoding $\rho(\cdot,\boldsymbol{s}_i)$ is $(\Xi,\tau)$-neighbouring-preserving, where
\[
  \tau \leq \sqrt{\frac{d}{2}}\|\boldsymbol{\theta} - \boldsymbol{\theta'}\|_2 \;\; \text{ and } \;\; \max_{\mathcal{I}\in \Xi} |\mathcal{I}|\leq 2L.
\]
Hence, we have
\begin{align*}
    |\ell(\boldsymbol{\theta},\boldsymbol{s_i}) - \ell(\boldsymbol{\theta'},\boldsymbol{s_i})|  \leq |\Tr\{O(y_i) \rho(\boldsymbol{\theta}; \boldsymbol{x_i})-\Tr\{O(y_i) \rho(\boldsymbol{\theta'}; \boldsymbol{x_i})\}| 
    \\ \leq \|O(y_i)\|_{Lip} W_1(\rho(\boldsymbol{\theta}; \boldsymbol{x_i}),\rho(\boldsymbol{\theta'}; \boldsymbol{x_i})) \leq  3L\sqrt{\frac{d}{2}} \|O(y_i)\|_{Lip}\|\boldsymbol{\theta} - \boldsymbol{\theta'}\|_2.
\end{align*}
And then
\[
\forall\boldsymbol{\theta}: \|\nabla_{\boldsymbol{\theta}} \ell(\boldsymbol{\theta},\boldsymbol{s_i})\|_{2} \leq 3L \sqrt{\frac{d}{2}} \|O(y_i)\|_{Lip}.
\]

\subsubsection*{Generalisation}
We conclude by recalling the connection between differential privacy and generalisation.
Given a randomised algorithm $M: \mathcal{X}^{m}\times \mathcal{X}\rightarrow [0,B]$ and two datasets $S,S'\in \mathcal{X}^m$ we define the following quantity:
\[
\mathcal{E}_{S}[M(S)]:= \frac{1}{m}\sum_{z \in S} \mathbb{E}_M [M(S,z)],\;\;\;\mathcal{E}_{S'}[M(S)]:= \frac{1}{m}\sum_{z' \in S'} \mathbb{E}_M [M(S,z')].
\]

\begin{lemma}[Lemma 6.4, \cite{feldman2017generalization}]
Let $S \in \mathcal{X}^m$ and $\boldsymbol{x}\in \mathcal{X}$. Let $M$ be an algorithm that on input $(S,\boldsymbol{x})$ outputs a value $y \in [0,B]$. Assume that $M$ is $(\epsilon, \delta)$-differentially private with respect to $S$, where $S\sim S'$ if they differ in at most one entry.
Let $\mathcal{P}$ be an arbitrary distribution over $\mathcal{X}$. Then:
\[
\mathbb{E}_{S,S'\sim \mathcal{P}^m}[(\mathcal{E}_{S'}[M(S)])^k]\leq e^{k^2\epsilon}\mathbb{E}_{S\sim \mathcal{P}^m}[(\mathcal{E}_{S'}[M(S)]+k\delta B)^k].
\]
\end{lemma}

We also define $\mathcal{E}_{\mathcal{P}}[M(S)] := \mathbb{E}_{z\sim \mathcal{P},\ M}[M(S,z)] $. Clearly,
$$\mathbb{E}_{S' \sim \mathcal{P}^n}[ \mathcal{E}_{S'}[M(S)]] = \mathcal{E}_{\mathcal{P}}[M(S)].  $$ 

Moreover, as noted in \cite{feldman2017generalization}, standard concentration inequalities implies that $\mathcal{E}_{S'}[M(S)]$ is strongly concentrated around $\mathcal{E}_\mathcal{P}[M(S)]$.
Note that for $M(S,(x,y)) = \ell(M'(S,x),y)$, $\mathcal{E}_S[M(S)] = \mathcal{E}_S[\ell(M'(S))]$ and $\mathcal{E}_\mathcal{P}[M(S)] = \mathcal{E}_\mathcal{P}[\ell(M'(S))]$, in other words these are exactly the empirical and the expected loss of the predictor given by $M'$.

\section*{Acknowledgments}
The authors thank Daniel Stilck-França, Christoph Hirche, Yihui Quek and Chirag Wadhwa for helpful discussions at different phases of this project. 
AA acknowledges financial support from the QICS (Quantum Information Center Sorbonne) and the H2020-FETOPEN Grant PHOQUSING (GA no.: 899544).


\bibliography{qdp}

\begin{thebibliography}{66}
\providecommand{\natexlab}[1]{#1}
\providecommand{\url}[1]{\texttt{#1}}
\expandafter\ifx\csname urlstyle\endcsname\relax
  \providecommand{\doi}[1]{doi: #1}\else
  \providecommand{\doi}{doi: \begingroup \urlstyle{rm}\Url}\fi

\bibitem[Narayanan and Shmatikov(2007)]{narayanan2007break}
Arvind Narayanan and Vitaly Shmatikov.
\newblock How to break anonymity of the netflix prize dataset, 2007.

\bibitem[Dwork et~al.(2006)Dwork, McSherry, Nissim, and Smith]{dwork1}
Cynthia Dwork, Frank McSherry, Kobbi Nissim, and Adam Smith.
\newblock Calibrating noise to sensitivity in private data analysis.
\newblock In \emph{Proceedings of the Third Conference on Theory of
  Cryptography}, TCC'06, page 265–284, Berlin, Heidelberg, 2006.
  Springer-Verlag.
\newblock ISBN 3540327312.
\newblock \doi{10.1007/11681878_14}.
\newblock URL \url{https://doi.org/10.1007/11681878_14}.

\bibitem[Dwork and Roth(2014)]{dwork2}
Cynthia Dwork and Aaron Roth.
\newblock The algorithmic foundations of differential privacy.
\newblock 9\penalty0 (3–4):\penalty0 211–407, August 2014.
\newblock ISSN 1551-305X.
\newblock \doi{10.1561/0400000042}.
\newblock URL \url{https://doi.org/10.1561/0400000042}.

\bibitem[Cummings et~al.(2023)Cummings, Desfontaines, Evans, Geambasu,
  Jagielski, Huang, Kairouz, Kamath, Oh, Ohrimenko,
  et~al.]{cummings2023challenges}
Rachel Cummings, Damien Desfontaines, David Evans, Roxana Geambasu, Matthew
  Jagielski, Yangsibo Huang, Peter Kairouz, Gautam Kamath, Sewoong Oh, Olga
  Ohrimenko, et~al.
\newblock Challenges towards the next frontier in privacy.
\newblock \emph{arXiv preprint arXiv:2304.06929}, 2023.

\bibitem[Chaudhuri et~al.(2011)Chaudhuri, Monteleoni, and
  Sarwate]{JMLR:v12:chaudhuri11a}
Kamalika Chaudhuri, Claire Monteleoni, and Anand~D. Sarwate.
\newblock Differentially private empirical risk minimization.
\newblock \emph{Journal of Machine Learning Research}, 12\penalty0
  (29):\penalty0 1069--1109, 2011.
\newblock URL \url{http://jmlr.org/papers/v12/chaudhuri11a.html}.

\bibitem[Abadi et~al.(2016)Abadi, Chu, Goodfellow, McMahan, Mironov, Talwar,
  and Zhang]{abadi2016}
Martin Abadi, Andy Chu, Ian Goodfellow, H.~Brendan McMahan, Ilya Mironov, Kunal
  Talwar, and Li~Zhang.
\newblock Deep learning with differential privacy.
\newblock \emph{Proceedings of the 2016 ACM SIGSAC Conference on Computer and
  Communications Security}, Oct 2016.
\newblock \doi{10.1145/2976749.2978318}.
\newblock URL \url{http://dx.doi.org/10.1145/2976749.2978318}.

\bibitem[Papernot et~al.(2017)Papernot, Abadi, Úlfar Erlingsson, Goodfellow,
  and Talwar]{papernot2017semisupervised}
Nicolas Papernot, Martín Abadi, Úlfar Erlingsson, Ian Goodfellow, and Kunal
  Talwar.
\newblock Semi-supervised knowledge transfer for deep learning from private
  training data, 2017.

\bibitem[Bassily et~al.(2018)Bassily, Thakkar, and Thakurta]{bassily2018}
Raef Bassily, Om~Thakkar, and Abhradeep Thakurta.
\newblock Model-agnostic private learning.
\newblock In \emph{Proceedings of the 32nd International Conference on Neural
  Information Processing Systems}, NIPS'18, page 7102–7112, Red Hook, NY,
  USA, 2018. Curran Associates Inc.

\bibitem[Kasiviswanathan et~al.(2011)Kasiviswanathan, Lee, Nissim,
  Raskhodnikova, and Smith]{equiv}
Shiva~Prasad Kasiviswanathan, Homin~K. Lee, Kobbi Nissim, Sofya Raskhodnikova,
  and Adam Smith.
\newblock What can we learn privately?
\newblock \emph{SIAM J. Comput.}, 40\penalty0 (3):\penalty0 793–826, June
  2011.
\newblock ISSN 0097-5397.
\newblock \doi{10.1137/090756090}.
\newblock URL \url{https://doi.org/10.1137/090756090}.

\bibitem[Wang et~al.(2016)Wang, Lei, and Fienberg]{JMLR:v17:15-313}
Yu-Xiang Wang, Jing Lei, and Stephen~E. Fienberg.
\newblock Learning with differential privacy: Stability, learnability and the
  sufficiency and necessity of erm principle.
\newblock \emph{Journal of Machine Learning Research}, 17\penalty0
  (183):\penalty0 1--40, 2016.
\newblock URL \url{http://jmlr.org/papers/v17/15-313.html}.

\bibitem[Bun et~al.(2020)Bun, Livni, and Moran]{online}
M.~Bun, R.~Livni, and S.~Moran.
\newblock An equivalence between private classification and online prediction.
\newblock In \emph{2020 IEEE 61st Annual Symposium on Foundations of Computer
  Science (FOCS)}, pages 389--402, Los Alamitos, CA, USA, nov 2020. IEEE
  Computer Society.
\newblock \doi{10.1109/FOCS46700.2020.00044}.
\newblock URL
  \url{https://doi.ieeecomputersociety.org/10.1109/FOCS46700.2020.00044}.

\bibitem[Arunachalam et~al.(2021)Arunachalam, Quek, and
  Smolin]{arunachalam2021private}
Srinivasan Arunachalam, Yihui Quek, and John Smolin.
\newblock Private learning implies quantum stability.
\newblock In \emph{Advances in Neural Information Processing Systems 34
  pre-proceedings (NeurIPS 2021)}, NIPS'21, 2021.

\bibitem[Dwork et~al.(2015)Dwork, Feldman, Hardt, Pitassi, Reingold, and
  Roth]{dwork2015preserving}
Cynthia Dwork, Vitaly Feldman, Moritz Hardt, Toniann Pitassi, Omer Reingold,
  and Aaron~Leon Roth.
\newblock Preserving statistical validity in adaptive data analysis.
\newblock In \emph{Proceedings of the forty-seventh annual ACM symposium on
  Theory of computing}, pages 117--126, 2015.

\bibitem[Bassily et~al.(2021)Bassily, Nissim, Smith, Steinke, Stemmer, and
  Ullman]{bassily2021algorithmic}
Raef Bassily, Kobbi Nissim, Adam Smith, Thomas Steinke, Uri Stemmer, and
  Jonathan Ullman.
\newblock Algorithmic stability for adaptive data analysis.
\newblock \emph{SIAM Journal on Computing}, 50\penalty0 (3):\penalty0
  STOC16--377, 2021.

\bibitem[Feldman and Steinke(2017)]{feldman2017generalization}
Vitaly Feldman and Thomas Steinke.
\newblock Generalization for adaptively-chosen estimators via stable median.
\newblock In \emph{Conference on Learning Theory}, pages 728--757. PMLR, 2017.

\bibitem[McSherry and Talwar(2007)]{design}
Frank McSherry and Kunal Talwar.
\newblock Mechanism design via differential privacy.
\newblock In \emph{48th Annual IEEE Symposium on Foundations of Computer
  Science (FOCS'07)}, pages 94--103, 2007.
\newblock \doi{10.1109/FOCS.2007.66}.

\bibitem[Senekane et~al.(2017)Senekane, Mafu, and Taele]{senekane2017privacy}
Makhamisa Senekane, Mhlambululi Mafu, and Benedict~Molibeli Taele.
\newblock Privacy-preserving quantum machine learning using differential
  privacy.
\newblock In \emph{2017 IEEE AFRICON}, pages 1432--1435. IEEE, 2017.

\bibitem[Li et~al.(2021)Li, Lu, and Deng]{Li_2021}
Weikang Li, Sirui Lu, and Dong-Ling Deng.
\newblock Quantum federated learning through blind quantum computing.
\newblock \emph{Science China Physics, Mechanics {\&} Astronomy}, 64\penalty0
  (10), sep 2021.
\newblock \doi{10.1007/s11433-021-1753-3}.
\newblock URL \url{https://doi.org/10.1007%2Fs11433-021-1753-3}.

\bibitem[Du et~al.(2022)Du, Hsieh, Liu, You, and Tao]{Du_2022}
Yuxuan Du, Min-Hsiu Hsieh, Tongliang Liu, Shan You, and Dacheng Tao.
\newblock Quantum differentially private sparse regression learning.
\newblock \emph{{IEEE} Transactions on Information Theory}, 68\penalty0
  (8):\penalty0 5217--5233, aug 2022.
\newblock \doi{10.1109/tit.2022.3164726}.
\newblock URL \url{https://doi.org/10.1109%2Ftit.2022.3164726}.

\bibitem[Watkins et~al.(2023)Watkins, Chen, and Yoo]{watkins2023quantum}
William~M Watkins, Samuel Yen-Chi Chen, and Shinjae Yoo.
\newblock Quantum machine learning with differential privacy.
\newblock \emph{Scientific Reports}, 13\penalty0 (1):\penalty0 2453, 2023.

\bibitem[Preskill(2018)]{preskill_quantum_2018}
John Preskill.
\newblock Quantum {Computing} in the {NISQ} era and beyond.
\newblock \emph{Quantum}, 2:\penalty0 79, August 2018.
\newblock \doi{10.22331/q-2018-08-06-79}.
\newblock URL \url{https://quantum-journal.org/papers/q-2018-08-06-79/}.
\newblock Publisher: Verein zur F{\"o}rderung des Open Access Publizierens in
  den Quantenwissenschaften.

\bibitem[Zhou and Ying(2017)]{quantumDP}
Li~Zhou and Mingsheng Ying.
\newblock Differential privacy in quantum computation.
\newblock In \emph{2017 IEEE 30th Computer Security Foundations Symposium
  (CSF)}, pages 249--262, 2017.
\newblock \doi{10.1109/CSF.2017.23}.

\bibitem[Aaronson and Rothblum(2019)]{aaronson2019gentle}
Scott Aaronson and Guy~N. Rothblum.
\newblock Gentle measurement of quantum states and differential privacy.
\newblock In \emph{Proceedings of the 51st Annual ACM SIGACT Symposium on
  Theory of Computing}, STOC 2019, page 322–333, New York, NY, USA, 2019.
  Association for Computing Machinery.
\newblock ISBN 9781450367059.
\newblock \doi{10.1145/3313276.3316378}.
\newblock URL \url{https://doi.org/10.1145/3313276.3316378}.

\bibitem[Hirche et~al.(2022{\natexlab{a}})Hirche, Rouzé, and França]{franca}
Christoph Hirche, Cambyse Rouzé, and Daniel~Stilck França.
\newblock Quantum differential privacy: An information theory perspective,
  2022{\natexlab{a}}.
\newblock URL \url{https://arxiv.org/abs/2202.10717}.

\bibitem[Farokhi(2023)]{farokhi2023privacy}
Farhad Farokhi.
\newblock Privacy against hypothesis-testing adversaries for quantum computing,
  2023.

\bibitem[Nuradha et~al.(2023)Nuradha, Goldfeld, and Wilde]{nuradha2023quantum}
Theshani Nuradha, Ziv Goldfeld, and Mark~M. Wilde.
\newblock Quantum pufferfish privacy: A flexible privacy framework for quantum
  systems, 2023.

\bibitem[Lecuyer et~al.(2019)Lecuyer, Atlidakis, Geambasu, Hsu, and
  Jana]{lecuyer2019certified}
Mathias Lecuyer, Vaggelis Atlidakis, Roxana Geambasu, Daniel Hsu, and Suman
  Jana.
\newblock Certified robustness to adversarial examples with differential
  privacy, 2019.

\bibitem[De~Palma et~al.(2021)De~Palma, Marvian, Trevisan, and
  Lloyd]{wasserstein2021}
Giacomo De~Palma, Milad Marvian, Dario Trevisan, and Seth Lloyd.
\newblock The quantum wasserstein distance of order 1.
\newblock \emph{IEEE Transactions on Information Theory}, 67\penalty0
  (10):\penalty0 6627–6643, Oct 2021.
\newblock ISSN 1557-9654.
\newblock \doi{10.1109/tit.2021.3076442}.
\newblock URL \url{http://dx.doi.org/10.1109/TIT.2021.3076442}.

\bibitem[Arunachalam et~al.(2020)Arunachalam, Grilo, and
  Yuen]{arunachalam2020quantum}
Srinivasan Arunachalam, Alex~B. Grilo, and Henry Yuen.
\newblock Quantum statistical query learning, 2020.
\newblock URL \url{https://arxiv.org/abs/2002.08240}.

\bibitem[Angrisani and Kashefi(2022)]{Angrisani2022QuantumLD}
Armando Angrisani and Elham Kashefi.
\newblock Quantum local differential privacy and quantum statistical query
  model.
\newblock \emph{ArXiv}, abs/2203.03591, 2022.

\bibitem[Yoshida and Hayashi(2020)]{yoshida2020classical}
Yuuya Yoshida and Masahito Hayashi.
\newblock Classical mechanism is optimal in classical-quantum differentially
  private mechanisms.
\newblock In \emph{2020 IEEE International Symposium on Information Theory
  (ISIT)}, pages 1973--1977. IEEE, 2020.

\bibitem[Yoshida(2021)]{yoshida2021mathematical}
Yuuya Yoshida.
\newblock Mathematical comparison of classical and quantum mechanisms in
  optimization under local differential privacy, 2021.

\bibitem[Vadhan(2017)]{complexity}
Salil Vadhan.
\newblock \emph{The Complexity of Differential Privacy}, pages 347--450.
\newblock Springer, Yehuda Lindell, ed., 2017.
\newblock URL
  \url{https://link.springer.com/chapter/10.1007/978-3-319-57048-8_7}.

\bibitem[Polyanskiy et~al.(2010)Polyanskiy, Poor, and Verdu]{hs-div}
Yury Polyanskiy, H.~Vincent Poor, and Sergio Verdu.
\newblock Channel coding rate in the finite blocklength regime.
\newblock \emph{IEEE Transactions on Information Theory}, 56\penalty0
  (5):\penalty0 2307--2359, 2010.
\newblock \doi{10.1109/TIT.2010.2043769}.

\bibitem[Bun and Steinke(2016)]{bun2016concentrated}
Mark Bun and Thomas Steinke.
\newblock Concentrated differential privacy: Simplifications, extensions, and
  lower bounds.
\newblock In \emph{Theory of Cryptography: 14th International Conference, TCC
  2016-B, Beijing, China, October 31-November 3, 2016, Proceedings, Part I},
  pages 635--658. Springer, 2016.

\bibitem[Meiser(2018)]{meiser2018approximate}
Sebastian Meiser.
\newblock Approximate and probabilistic differential privacy definitions.
\newblock \emph{Cryptology ePrint Archive}, 2018.

\bibitem[Mironov(2017)]{Mironov_2017}
Ilya Mironov.
\newblock R{\'{e}}nyi differential privacy.
\newblock In \emph{2017 {IEEE} 30th Computer Security Foundations Symposium
  ({CSF})}. {IEEE}, aug 2017.
\newblock \doi{10.1109/csf.2017.11}.
\newblock URL \url{https://doi.org/10.1109%2Fcsf.2017.11}.

\bibitem[Tomamichel(2015)]{tomamichel2015quantum}
Marco Tomamichel.
\newblock \emph{Quantum information processing with finite resources:
  mathematical foundations}, volume~5.
\newblock Springer, 2015.

\bibitem[Du et~al.(2021)Du, Hsieh, Liu, Tao, and Liu]{liu2021}
Yuxuan Du, Min-Hsiu Hsieh, Tongliang Liu, Dacheng Tao, and Nana Liu.
\newblock Quantum noise protects quantum classifiers against adversaries.
\newblock \emph{Physical Review Research}, 3\penalty0 (2), May 2021.
\newblock ISSN 2643-1564.
\newblock \doi{10.1103/physrevresearch.3.023153}.
\newblock URL \url{http://dx.doi.org/10.1103/PhysRevResearch.3.023153}.

\bibitem[Cerezo et~al.(2021)Cerezo, Sone, Volkoff, Cincio, and
  Coles]{cerezo2021cost}
Marco Cerezo, Akira Sone, Tyler Volkoff, Lukasz Cincio, and Patrick~J Coles.
\newblock Cost function dependent barren plateaus in shallow parametrized
  quantum circuits.
\newblock \emph{Nature communications}, 12\penalty0 (1):\penalty0 1791, 2021.

\bibitem[Wang et~al.(2021)Wang, Fontana, Cerezo, Sharma, Sone, Cincio, and
  Coles]{nibp}
Samson Wang, Enrico Fontana, M.~Cerezo, Kunal Sharma, Akira Sone, Lukasz
  Cincio, and Patrick~J. Coles.
\newblock Noise-induced barren plateaus in variational quantum algorithms.
\newblock \emph{Nature Communications}, 12\penalty0 (1), nov 2021.
\newblock \doi{10.1038/s41467-021-27045-6}.
\newblock URL \url{https://doi.org/10.1038%2Fs41467-021-27045-6}.

\bibitem[Balle et~al.(2018)Balle, Barthe, and Gaboardi]{subsampling}
Borja Balle, Gilles Barthe, and Marco Gaboardi.
\newblock Privacy amplification by subsampling: Tight analyses via couplings
  and divergences.
\newblock In \emph{Proceedings of the 32nd International Conference on Neural
  Information Processing Systems}, NIPS'18, page 6280–6290, Red Hook, NY,
  USA, 2018. Curran Associates Inc.

\bibitem[Tang(2019)]{tang1}
Ewin Tang.
\newblock A quantum-inspired classical algorithm for recommendation systems.
\newblock In \emph{Proceedings of the 51st Annual ACM SIGACT Symposium on
  Theory of Computing}, STOC 2019, page 217–228, New York, NY, USA, 2019.
  Association for Computing Machinery.
\newblock ISBN 9781450367059.
\newblock \doi{10.1145/3313276.3316310}.
\newblock URL \url{https://doi.org/10.1145/3313276.3316310}.

\bibitem[Tang(2021)]{tang2}
Ewin Tang.
\newblock Quantum principal component analysis only achieves an exponential
  speedup because of its state preparation assumptions.
\newblock \emph{Physical Review Letters}, 127\penalty0 (6), Aug 2021.
\newblock ISSN 1079-7114.
\newblock \doi{10.1103/physrevlett.127.060503}.
\newblock URL \url{http://dx.doi.org/10.1103/PhysRevLett.127.060503}.

\bibitem[Gilyén et~al.(2018)Gilyén, Lloyd, and Tang]{tang3}
András Gilyén, Seth Lloyd, and Ewin Tang.
\newblock Quantum-inspired low-rank stochastic regression with logarithmic
  dependence on the dimension, 2018.

\bibitem[Chia et~al.(2018)Chia, Lin, and Wang]{tang4}
Nai-Hui Chia, Han-Hsuan Lin, and Chunhao Wang.
\newblock Quantum-inspired sublinear classical algorithms for solving low-rank
  linear systems, 2018.
\newblock URL \url{https://arxiv.org/abs/1811.04852}.

\bibitem[Chia et~al.(2020)Chia, Gily\'{e}n, Li, Lin, Tang, and
  Wang]{dequantizing}
Nai-Hui Chia, Andr\'{a}s Gily\'{e}n, Tongyang Li, Han-Hsuan Lin, Ewin Tang, and
  Chunhao Wang.
\newblock \emph{Sampling-Based Sublinear Low-Rank Matrix Arithmetic Framework
  for Dequantizing Quantum Machine Learning}, page 387–400.
\newblock Association for Computing Machinery, New York, NY, USA, 2020.
\newblock ISBN 9781450369794.
\newblock URL \url{https://doi.org/10.1145/3357713.3384314}.

\bibitem[Ullman(2017)]{ullman}
Jonathan Ullman.
\newblock Cs7880: Rigorous approaches to data privacy, 2017.
\newblock URL \url{https://www.ccs.neu.edu/home/jullman/cs7880s17/HW1sol.pdf}.

\bibitem[Fran{\c{c}}a and Garc{\'{\i}}a-Patr{\'{o}}n(2021)]{Stilck_Fran_a_2021}
Daniel~Stilck Fran{\c{c}}a and Raul Garc{\'{\i}}a-Patr{\'{o}}n.
\newblock Limitations of optimization algorithms on noisy quantum devices.
\newblock \emph{Nature Physics}, 17\penalty0 (11):\penalty0 1221--1227, oct
  2021.
\newblock \doi{10.1038/s41567-021-01356-3}.
\newblock URL \url{https://doi.org/10.1038%2Fs41567-021-01356-3}.

\bibitem[De~Palma et~al.(2023)De~Palma, Marvian, Rouz{\'e}, and
  Fran{\c{c}}a]{de2023limitations}
Giacomo De~Palma, Milad Marvian, Cambyse Rouz{\'e}, and Daniel~Stilck
  Fran{\c{c}}a.
\newblock Limitations of variational quantum algorithms: a quantum optimal
  transport approach.
\newblock \emph{PRX Quantum}, 4\penalty0 (1):\penalty0 010309, 2023.

\bibitem[Hirche(2023)]{hirche2023}
Christoph Hirche.
\newblock Benefits and detriments of noise in quantum classification.
\newblock 2023.

\bibitem[Huang et~al.(2023)Huang, Tsai, Yang, Su, Yu, Chen, and
  Kuo]{huang2023certified}
Jhih-Cing Huang, Yu-Lin Tsai, Chao-Han~Huck Yang, Cheng-Fang Su, Chia-Mu Yu,
  Pin-Yu Chen, and Sy-Yen Kuo.
\newblock Certified robustness of quantum classifiers against adversarial
  examples through quantum noise.
\newblock In \emph{ICASSP 2023-2023 IEEE International Conference on Acoustics,
  Speech and Signal Processing (ICASSP)}, pages 1--5. IEEE, 2023.

\bibitem[Kingma and Ba(2014)]{kingma2014adam}
Diederik~P Kingma and Jimmy Ba.
\newblock Adam: A method for stochastic optimization.
\newblock \emph{arXiv preprint arXiv:1412.6980}, 2014.

\bibitem[Bassily et~al.(2014)Bassily, Smith, and Thakurta]{bassily2014private}
Raef Bassily, Adam Smith, and Abhradeep Thakurta.
\newblock Private empirical risk minimization: Efficient algorithms and tight
  error bounds.
\newblock In \emph{2014 IEEE 55th annual symposium on foundations of computer
  science}, pages 464--473. IEEE, 2014.

\bibitem[Mosonyi and Hiai(2011)]{mosonyi2011quantum}
Mil{\'a}n Mosonyi and Fumio Hiai.
\newblock On the quantum r{\'e}nyi relative entropies and related capacity
  formulas.
\newblock \emph{IEEE Transactions on Information Theory}, 57\penalty0
  (4):\penalty0 2474--2487, 2011.

\bibitem[Müller-Lennert et~al.(2013)Müller-Lennert, Dupuis, Szehr, Fehr, and
  Tomamichel]{Muller_Lennert_2013}
Martin Müller-Lennert, Frédéric Dupuis, Oleg Szehr, Serge Fehr, and Marco
  Tomamichel.
\newblock On quantum r{\'{e}}nyi entropies: A new generalization and some
  properties.
\newblock \emph{Journal of Mathematical Physics}, 54\penalty0 (12):\penalty0
  122203, dec 2013.
\newblock \doi{10.1063/1.4838856}.
\newblock URL \url{https://doi.org/10.1063%2F1.4838856}.

\bibitem[Hirche et~al.(2022{\natexlab{b}})Hirche, Rouz{\'e}, and
  Fran{\c{c}}a]{hirche2022contraction}
Christoph Hirche, Cambyse Rouz{\'e}, and Daniel~Stilck Fran{\c{c}}a.
\newblock On contraction coefficients, partial orders and approximation of
  capacities for quantum channels.
\newblock \emph{Quantum}, 6:\penalty0 862, 2022{\natexlab{b}}.

\bibitem[van Erven and Harremoes(2014)]{van_Erven_2014}
Tim van Erven and Peter Harremoes.
\newblock R{\'{e} }nyi divergence and kullback-leibler divergence.
\newblock \emph{{IEEE} Transactions on Information Theory}, 60\penalty0
  (7):\penalty0 3797--3820, jul 2014.
\newblock \doi{10.1109/tit.2014.2320500}.
\newblock URL \url{https://doi.org/10.1109%2Ftit.2014.2320500}.

\bibitem[Sharma and Warsi(2012)]{sharma2012strong}
Naresh Sharma and Naqueeb~Ahmad Warsi.
\newblock On the strong converses for the quantum channel capacity theorems.
\newblock \emph{arXiv preprint arXiv:1205.1712}, 2012.

\bibitem[Hirche and Tomamichel(2023)]{hirche2023quantum}
Christoph Hirche and Marco Tomamichel.
\newblock Quantum r$\backslash$'enyi and $ f $-divergences from integral
  representations.
\newblock \emph{arXiv preprint arXiv:2306.12343}, 2023.

\bibitem[Bretagnolle and Huber(1978)]{bretagnolle1978estimation}
Jean Bretagnolle and Catherine Huber.
\newblock Estimation des densit{\'e}s: risque minimax.
\newblock \emph{S{\'e}minaire de probabilit{\'e}s de Strasbourg}, 12:\penalty0
  342--363, 1978.

\bibitem[Canonne(2022)]{canonne2022short}
Clément~L. Canonne.
\newblock A short note on an inequality between kl and tv, 2022.

\bibitem[Park and Cho(2018)]{Park_2018}
Chae-Yeun Park and Jaeyoon Cho.
\newblock Correlations in local measurements and entanglement in many-body
  systems.
\newblock \emph{Physical Review A}, 98\penalty0 (1), jul 2018.
\newblock \doi{10.1103/physreva.98.012107}.
\newblock URL \url{https://doi.org/10.1103%2Fphysreva.98.012107}.

\bibitem[Schuld(2021)]{schuld2021supervised}
Maria Schuld.
\newblock Supervised quantum machine learning models are kernel methods, 2021.
\newblock URL \url{https://arxiv.org/abs/2101.11020}.

\bibitem[Chatterjee and Yu(2016)]{chatterjee2016generalized}
Rupak Chatterjee and Ting Yu.
\newblock Generalized coherent states, reproducing kernels, and quantum support
  vector machines.
\newblock \emph{arXiv preprint arXiv:1612.03713}, 2016.

\bibitem[Berberich et~al.(2023)Berberich, Fink, and
  Holm]{berberich2023robustness}
J.~Berberich, D.~Fink, and C.~Holm.
\newblock Robustness of quantum algorithms against coherent control errors,
  2023.

\end{thebibliography}
\bibliographystyle{unsrtnat}
\newpage
\appendix

\section{Preliminaries}
We present here several technical tools that are used throughout the paper.
\label{app:pre}
\subsection{Schatten \emph{p}-norms}
Schatten $p$-norm can be used to define distances between linear operators. The Schatten $p$-norm of an operator $A\in\mathcal{L}(\mathcal{H}_n)$ is given by
\[
\|A\|_p:= [\Tr\{|A|^{p}\}]^{1/p},
\]
where $|A| := \sqrt{A^\dag A}$ and $p\geq 1$.
For each $p \in [1,\infty]$, we consider the dual index $q$ such that $\frac{1}{p}+\frac{1}{q}=1$. The H\"older inequality gives:
\begin{equation}
     \Tr\{A^\dag B\} \leq \|A\|_p\|B\|_q.
\end{equation}

\subsection{Rényi divergences}
Differential privacy, both in the classical and the quantum settings, can be expressed in terms of information-theoretic divergences. 
For two probability measures $P,Q$ the Rényi divergences of order $\alpha \in  (1, \infty)$ are defined as
\[
D_\alpha(P\|Q) = \frac{1}{\alpha - 1} \log \mathbb{E}_{x\sim Q}\left(\frac{P(x)}{Q(x)}\right)^\alpha,
\]
where we adopt the conventions that $0/0 = 0$ and $z/0 = \infty$ for $z > 0$. In the limit $\alpha \rightarrow 1$, the Rényi divergence reduces to the relative entropy, also known as the Kullback-Leibler divergence, i.e. $\lim_{\alpha \rightarrow 1} D_\alpha(P\|Q) = D(P\|Q) = \mathbb{E}_{x\sim P} \log \frac{P(x)}{Q(x)}$.
Moreover, by taking the limit $\alpha \rightarrow \infty$, we obtain the max-divergence
\[
D_\infty(P\|Q) = \sup_{S \subseteq \mathsf{supp }(Q)} \log\frac{P(S)}{Q(S)}.
\]
We will also need the related smooth max-divergence,
\[
D_\infty^\delta(P\|Q) = \sup_{S \subseteq \mathsf{supp }(Q): P(S)\geq \delta} \log\frac{P(S) - \delta}{Q(S)}.
\]
We emphasise that $D_\infty^\delta(P\|Q) \leq \epsilon$ if and only if for every subset $S$,
\[
P(S) \leq e^\epsilon Q(S) +\delta.
\]
Notably, the (smooth) max-divergence occurs in the definition of differential privacy. 
\par
Now we introduce divergences for quantum states. We make use of the quantum Petz-Rényi divergences \cite{mosonyi2011quantum, Muller_Lennert_2013} of order $\alpha \in  (1, \infty)$. For two states $\rho, \sigma$ such that the support of $\rho$ is included in the support of $\sigma$, they are defined as
\[
D_\alpha(\rho\|\sigma) = \frac{1}{\alpha -1} \log \Tr [\rho^\alpha \sigma^{1-\alpha}].
\]
In case the support of $\rho$ is not contained in that of $\sigma$, all the divergences above are defined to be $+ \infty$.
In the limit $\alpha \rightarrow 1$, the quantum Petz-Rényi divergence reduces to the quantum relative entropy, i.e., $\lim_{\alpha \rightarrow 1} D_\alpha(\rho\|\sigma) = D(\rho\|\sigma) = \Tr[\rho (\log \rho -\log \sigma)]$.
We also consider the divergence obtained by taking the limit $\alpha \rightarrow \infty$, known as quantum max-divergence,
\[
D_\infty(\rho\|\sigma) = \inf\{\lambda : \rho \leq e^\lambda \sigma\},
\]
and the related quantum smooth max-divergence \cite{franca},
\[
D_\infty^\delta (\rho\|\sigma)= \inf_{\overline{\rho}\in B_\delta(\rho)} D_\infty(\overline{\rho}\|\sigma),
\]
where $B^\delta (\rho) = \{\overline{\rho} : \overline{\rho}^\dag = \overline{\rho}\geq 0 \wedge \|\rho-\overline{\rho}\|_1<2\delta\}$.
Similarly to its classical counterpart, the quantum (smooth) max-divergence is at the heart of our work as it occurs in the definition of differentially private quantum channels.

Rényi divergences also play a key role in the analysis of quantum algorithms on noisy devices, as shown by the following result, which follows from Corollary 5.6 in \cite{hirche2022contraction}.

\begin{lemma}[Supplementary Lemma 6, \cite{nibp}]
\label{lem:nibp}
Consider a single instance of the noise channel $\mathcal{N} = \mathcal{N}_1\otimes \dots \otimes \mathcal{N}_n$ where each local
noise channel $\mathcal{N}_j$ is a Pauli noise channel that satisfies
\[
\mathcal{N}_j(\sigma) = q_{\sigma} \sigma
\]
for $\sigma \in \{X,Y,Z\}$ and $|q_{\sigma}| < 1$. Let $q=\sqrt{\max \{|q_X|, |q_Y|, |q_Z|\}}$.  Then, we have
\begin{equation}
    D_2(\mathcal{N}(\rho)\|\mathbb{1}/2^n) \leq q^2 D_2(\rho\|\mathbb{1}/2^n).
\end{equation}
\end{lemma}

The (standard) joint convexity of the Rényi divergence for $\alpha \in [0,\infty]$ is proven in \cite{van_Erven_2014} (Theorem 13). For the max divergence have
\[
D_\infty(\sum_i \lambda_i P_i \| \sum_i \lambda_i Q_i) \leq \max_i D_\infty(P_i\|Q_i).
\]
For the smooth max divergence, we can easily prove the statement from scratch. Assume $P_i(x)\leq e^\epsilon Q_i(x) +\delta$:
\[
\sum_i \lambda_i P_i(x) \leq \sum_i \lambda_i (e^\epsilon Q_i(x) +\delta) = e^\epsilon \left(\sum_i \lambda_i Q_i(x) \right) + \delta.
\]

\subsection{The quantum hockey-stick divergence}
The quantum hockey-stick divergence was first introduced in~\cite{sharma2012strong}, in the context of exploring strong converse bounds for the quantum capacity, and further investigate in \cite{franca,hirche2023quantum} in the context of quantum differential privacy. It is defined as
\begin{align}
E_\gamma(\rho\|\sigma) := \Tr(\rho-\gamma\sigma)^+, 
\end{align}
for $\gamma\geq 1$. Here $X^+$ denotes the positive part of the eigendecomposition of a Hermitian matrix $X = X^+ - X^-$. In~\cite{sharma2012strong} it was noted that this quantity is closely related to the trace norm via
\begin{align}
E_\gamma(\rho\|\sigma) = \frac12 \| \rho -\gamma\sigma\|_1 + \frac12(\Tr(\rho)-\gamma\Tr(\sigma)), \label{Eq:Hs-trace}
\end{align}
so for $\rho,\sigma$ quantum states, $E_1(\rho\|\sigma) = \frac12 \| \rho -\sigma\|_1$ equals the trace distance.
We also state some useful properties of the hockey-stick divergence proven in (\cite{franca}, Proposition II.5).
\begin{itemize}
 \item{(Triangle inequality)}   For $\gamma_1, \gamma_2 \geq 1$ and $\rho,\sigma\in\mathcal{S}_n$, we have
\begin{align}
E_{\gamma_1\gamma_2}(\rho\|\sigma) \leq E_{\gamma_1}(\rho\|\tau) + \gamma_1 E_{\gamma_2}(\tau\|\sigma). \label{Eq:Triangle}
\end{align}
\item{(Convexity)} Let $\gamma_1,\gamma_2\geq 1$, $\rho=\sum_x p(x)\rho_x$ and  $\sigma=\sum_x q(x)\sigma_x$ with $\rho_x,\sigma_x\in\mathcal{S}_n$, we have
\begin{align}
    E_{\gamma_1\gamma_2}(\rho\|\sigma) \leq \sum_x p(x) E_{\gamma_1}(\rho_x\|\sigma_x) + \gamma_1  E_{\gamma_2}(\tilde p\| \tilde q), \label{eq:conv-hs}
\end{align}
where $\tilde p$ and $\tilde q$ are non-normalised  distributions $\tilde p(x) = p(x)\tr\sigma_x$ and $\tilde q(x) = q(x)\tr\sigma_x$, respectively. This also implies convexity and joint convexity. 
\item{(Stability)} For $\gamma\geq 1$ and $\rho,\sigma,\tau\in\mathcal{S}_n$ , we have
\begin{align}
    E_\gamma(\rho\otimes\tau\|\sigma\otimes\tau) = \tr\left[\tau\right]E_\gamma(\rho\|\sigma). \label{eq:stab-hs}
\end{align}
\end{itemize}
\subsection{The Wasserstein distance of order 1}
We adopt the definition of quantum Wasserstein distance of order 1 proposed in \cite{wasserstein2021}. This is based on the following notion of neighbouring quantum states, which also arises in the context of differentially private measurements \cite{aaronson2019gentle}.
We say that $\rho$ and $\sigma \in \mathcal{S}_n$ are neighbouring if they coincide after discarding
one qudit, i.e., if $\Tr_i \rho = \Tr_i \sigma$ for some $i \in [n]$.
The quantum $W_1$ distance between the quantum states $\rho$ and $\sigma$ of $\mathcal{H}_n$ is defined as
\begin{align*}
    W_1(\rho,\sigma) = \min \bigg{(} \sum_{i=1}^n c_i : c_i \geq 0, \rho -\sigma = \sum_{i=1}^n c_i \left(\rho^{(i)}-\sigma^{(i)}\right), \\ \rho^{(i)}, \sigma^{(i)}\in \mathcal{S}_n, \Tr_i \rho^{(i)} = \Tr_i\sigma^{(i)}\bigg{)}. 
\end{align*}
The Wasserstein distance of order 1 and the trace distance are within a multiplicative factor $n$,
\begin{equation}
\label{eq:w1-tr}
   \frac{1}{2}\|\rho-\sigma\|_1\leq W_1(\rho,\sigma) \leq \frac{n}{2}\|\rho-\sigma\|_1.
\end{equation}

We also define the quantum Lipschitz constant of a self-adjoint linear operator $H\in \mathcal{O}_n$:
\[
\|H\|_L = \max_{i \in [n]}(\max(\Tr[H(\rho-\sigma)] : \rho,\sigma \in \mathcal{S}_n, \Tr_i \rho = \Tr_i \sigma)).
\]
From the definition of Wasserstein distance, we can readily derive that
\[
\Tr[H(\rho-\sigma)]\leq \|H\|_L W_1(\rho,\sigma).
\]
We also need the following technical lemma that can be used to upper bound the quantum $W_1$ distance under the action of a local evolution.
\begin{lemma}[Proposition 5, \cite{wasserstein2021}]
\label{lem:w1}
Let $\mathcal{I} \subseteq [n]$, and let $\rho,\sigma \in \mathcal{S}_n$ such that $\Tr_{\mathcal{I}}\rho = \Tr_{\mathcal{I}}\sigma$,
\[
W_1(\rho,\sigma) \leq |\mathcal{I}| \frac{d^2 -1}{d^2}\|\rho-\sigma\|_1.
\]
\end{lemma}

\begin{lemma}[Proposition IV.8, \cite{franca}]
\label{lem:high-noise}
Given a noisy circuit $\mathcal{A}$ over $n$ qubits, consisting in $L$ layers interspersed by local depolarising noise of parameter $0 \leq p \leq 1$, we assume that each
layer of the circuit is a quantum channel of light-cone $\mathcal{I}$. Then, we have that for any two input states $\rho, \sigma$ we have
\[
W_1(\mathcal{A}(\rho),\mathcal{A}(\sigma))\leq (2|\mathcal{I}|(1-p))^L W_1(\rho,\sigma),
\]
and hence,
\[
\frac{1}{2}\|\mathcal{A}(\rho)-\mathcal{A}(\sigma)\|_1\leq \frac{n}{2}(2|\mathcal{I}|(1-p))^L \|\rho-\sigma\|_1,
\]

In other words, the trace distance between any two output
states vanishes in logarithmic depth as soon as $p$ satisfies $2|I|(1 - p) < 1$.
\end{lemma}

\section{Improved bounds for quantum divergences}
\label{sec:div}
We present two technical contributions that establish tighter bounds for quantum divergences. 
First, we prove here a quantum version of the Bretagnolle-Huber (BH) inequality \cite{bretagnolle1978estimation,canonne2022short}.
The proof closely follows the one of the classical BH inequality, and for this reason the quantum BH can be regarded as a folklore result. However, we include here the complete proof since, to the best of our knowledge, it doesn't appear in any previous reference. 
We remark that a different quantum generalisation of the BH inequality result was provided in \cite{Park_2018} in the context of local measurements.

\begin{lemma}[Quantum Bretagnolle-Huber inequality]
\label{lem:qbh}
For every $\rho,\sigma$ we have
\[
\frac{1}{2}\|\rho-\sigma\|_1 \leq \sqrt{1-e^{-D(\rho\|\sigma)}}
\]
\end{lemma}

\begin{proof}
We define the following quantity
\begin{align*}
    U:=\rho^{-1}\sigma,\\
    V:=(U-\mathbb{1})^+,\\
    W:=\mathbb{1} + V - U = (U-\mathbb{1})^-.
\end{align*}
It's well known that 
\begin{align*}
    \Tr(\rho V) = \Tr(\sigma-\rho)^{+}  = \frac{1}{2}\|\rho-\sigma\|_1,\\
    \Tr(\rho W) = \Tr(\sigma-\rho)^{-}  = \frac{1}{2}\|\rho-\sigma\|_1.
\end{align*}
Moreover, remark that $(1+V)(1-W) = U$ and hence $\log U = \log(\mathbb{1}+V)+\log(\mathbb{1}-W)$. Applying the Jensen's inequality, we obtain 
\begin{align*}
    -D(\rho\|\sigma) \leq \Tr[\rho \log (\rho^{-1}\sigma)] = \Tr[\rho \log U] \\
    = \Tr[\rho\log(\mathbb{1}+V)] + \Tr[\rho\log(\mathbb{1}-W)]\leq \log \Tr[\rho(\mathbb{1}+V)] + \log \Tr[\rho(\mathbb{1}-W)]\\
    = \log(1-\Tr[\rho V]) + \log(1-\Tr[\rho W]) = \log \left(1-\frac{1}{2}\|\rho-\sigma\|_1^2\right).
\end{align*}
Exponentiating both sides, rearranging and taking the square root, proves the lemma.
\end{proof}

Following \cite{subsampling}, we prove a quantum version of the \emph{advanced joint convexity} of the hockey-stick divergence.
\begin{lemma}[Advanced joint convexity of the quantum hockey-stick divergence]
\label{lem:ajc}
For all states $\rho_0,\rho_1,\rho_2$  and $\gamma' = 1+ (1-p)(\gamma - 1)$, we have
\[
 E_{\gamma'}(p\rho_0 + (1-p)\rho_1\|p\rho_0 + (1-p)\rho_2)
\leq (1-p)(1-\beta) E_\gamma (\rho_1\|\rho_0) + (1-p)\beta E_\gamma(\rho_1\|\rho_2),
\]
where $\beta = \gamma'/\gamma$.
\end{lemma}

\begin{proof}
Recall that
\[E_{\gamma}(\rho\|\sigma):=\Tr(\rho-\gamma\sigma)^+ = \frac{1}{2}\|\rho-\gamma\sigma\|_1 + \frac{1}{2}(1-\gamma).\]
We have

\begin{align*}
 E_{\gamma'}(p\rho_0 + (1-p)\rho_1\|p\rho_0 + (1-p)\rho_2) = \Tr[p\rho_0 + (1-p)\rho_1 -\gamma'(p\rho_0 + (1-p)\rho_2)]^+\\
 =\Tr[p\rho_0 + (1-p)\rho_1 -(1+ (1-p)(\gamma - 1))(p\rho_0 + (1-p)\rho_2)]^+
 \\=(1-p) \Tr[\rho_1 - \gamma(\rho_0(1-\beta)+ \beta\rho_2)]^+ = (1-p)E_\gamma(\rho_1\|\rho_0(1-\beta)+\beta\rho_2)\\
\leq (1-p)(1-\beta) E_\gamma (\rho_1\|\rho_0) + (1-p)\beta E_\gamma(\rho_1\|\rho_2),
\end{align*}
where the inequality follows from the (standard) joint-convexity of the quantum hockey-stick divergence.
\end{proof}

\section{Quantum encodings}
\label{app:encodings}
Quantum encodings, also known as quantum feature maps or quantum embedding, are classical-to-quantum functions mapping vectors to quantum states. In this section, we review some popular encodings and highlight their connection with various quantum distances and neighbouring relationships. We refer to \cite{schuld2021supervised} for more details about the encodings and their corresponding kernel (i.e. the value of $|\langle \psi_{\boldsymbol{x}}| \psi_{\boldsymbol{x'}}\rangle|^2$ for two vectors $\boldsymbol{x}, \boldsymbol{x'}$).
Throughout this section, we will show that encoding vectors close in various $p$-distance leads to states that are either close in trace distance or that can be mapped one into the other by a local operation.
The results of this section are summarised in \tableref{tab:enc}.

\paragraph{Amplitude encoding.}
A normalised vector $\boldsymbol{x} = (x_1,\dots,x_{2^n}) \in \mathbb{C}^{2^n}$, $\|\boldsymbol{x}\|_2=1$ can
be represented by the amplitudes of a quantum state $\ket{\psi_{\boldsymbol{x} }}$ via
\[
\boldsymbol{x} \mapsto \ket{\psi_{\boldsymbol{x} }}= \sum_{j=1}^{2^n} x_j\ket{j}.
\]
For two normalised vectors $\boldsymbol{x},\boldsymbol{x'}$ we have
\[
|\langle{\psi_{\boldsymbol{x} }}|{\psi_{\boldsymbol{x'}}}\rangle| = |\boldsymbol{x}^\dag\boldsymbol{x'}| = \left|1 - \frac{1}{2}\|\boldsymbol{x}  - \boldsymbol{x'}\|_{2}^2\right|,
\]
where the second identity holds for any pair of normalised vectors. Hence,
\begin{align*}
    \frac{1}{2}\|\ket{\psi_{\boldsymbol{x} }}\bra{\psi_{\boldsymbol{x} }} - \ket{\psi_{\boldsymbol{x'} }}\bra{\psi_{\boldsymbol{x'} }}\|_1 = \sqrt{1 - |\langle{\psi_{\boldsymbol{x} }}|{\psi_{\boldsymbol{x'} }}\rangle|^2} \\= \sqrt{1 - |\boldsymbol{x}^\dag \boldsymbol{x'}|^2} 
    = \sqrt{1 - \left(1 - \frac{1}{2}\|\boldsymbol{x}  - \boldsymbol{x'}\|_{2}^2 \right)^2} 
    \\ \leq \|\boldsymbol{x}- \boldsymbol{x'}\|_2.
\end{align*}

\paragraph{Rotation encoding.}
Rotation encoding is a qubit-based embedding without any normalisation condition. Given a vector $\boldsymbol{x}$ in the hypercube  $[0, 2\pi]^{\otimes n}$, the $i^{th}$ feature $x_i$ is encoded into the $i^{th}$ qubit via a Pauli rotation. For example, a Pauli-Y rotation puts the qubit into state $\ket{q_i(x_i)} = \cos(x_i) \ket{0} + \sin(x_i) \ket{1}$. The data-encoding feature map is therefore given by

\[\phi: \boldsymbol{x} \rightarrow \rho(\boldsymbol{x}):=\ket{\phi(\boldsymbol{x})}\bra{\phi(\boldsymbol{x})} \text{ with } \ket{\phi( \boldsymbol{x})} = \sum_{q_1,\dots, q_n = 0}^{1} \prod_{k=1}^n \cos (x_k)^{q_k}  \sin (x_k)^{1-q_k} \ket{q_1, \dots, q_n}.
\]
Let $\mathcal{I} = \{ i : x_i \neq x_i'\}$. We have that $|\mathcal{I}| = \|\boldsymbol{x}-\boldsymbol{x'}\|_0$.
We immediately see that
\[
\Tr_{\mathcal{I}} \rho(\boldsymbol{x}) = \Tr_{\mathcal{I}} \rho(\boldsymbol{x'}).
\]

\paragraph{Coherent-state encoding.}
Coherent states are known in the field of quantum optics as a description of light modes. Formally, they are superpositions of so-called \textit{Fock states}, which are basis states from an infinite-dimensional discrete basis $\{\ket{0}, \ket{1}, \ket{2},...\}$, instead of the binary basis of qubits. 
A coherent state has the form
\[
\ket{\alpha} = e^{- \frac{|\alpha|^2}{2}} \sum\limits_{k=0}^{\infty} \frac{\alpha^k}{\sqrt{k!}} \ket{k},
\]
for $\alpha \in \mathbb{C}$.
Encoding a real scalar input $x_i \in \mathbb{R}$ into a coherent state $\ket{\alpha_{x_i}}$ corresponds to a data-encoding feature map with an infinite-dimensional feature space, 
\[
\phi: x_i   \rightarrow \ket{\alpha_{x_i}}\bra{\alpha_{x_i}}, \text{ with } \ket{\alpha_{x_i}} = e^{- \frac{|x_i|^2}{2}} \sum\limits_{k=0}^{\infty} \frac{x_i^k}{\sqrt{k!}} \ket{k} . 
\]
We can encode a real vector $\boldsymbol{x} =(x_1,...,x_n)$ into $n$ joint coherent states,
\[
\ket{\alpha_{\boldsymbol{x}}}= \ket{\alpha_{x_1}}\otimes  \dots \otimes  \ket{\alpha_{x_n}}. 
\]
Following \cite{schuld2021supervised, chatterjee2016generalized}, we have:
\[
|\langle \alpha_{\boldsymbol{x}}|\alpha_{\boldsymbol{x}}\rangle |^2 =  \left|e^{- \left( \frac{ \|\boldsymbol{x}\|_2^2}{2} +\frac{\|\boldsymbol{x}'\|_2^2}{2} - \boldsymbol{x}^{\dag}\boldsymbol{x}'   \right)}\right|^2 = e^{- \|\boldsymbol{x} - \boldsymbol{x}'\|_2^2 },
\]
and hence,
\begin{align*}
    \frac{1}{2}\|\ket{\phi_{\boldsymbol{x} }}\bra{\phi_{\boldsymbol{x} }} - \ket{\phi_{\boldsymbol{x'} }}\bra{\phi_{\boldsymbol{x'} }}\|_1 = \sqrt{1 - |\langle \alpha_{\boldsymbol{x}}|\alpha_{\boldsymbol{x}}\rangle |^2} 
    = \sqrt{1-e^{- \|\boldsymbol{x} - \boldsymbol{x}'\|_2^2}}.
\end{align*}
Moreover, let $\mathcal{I} = \{ i : x_i \neq x_i'\}$, where $|\mathcal{I}| = \|\boldsymbol{x}-\boldsymbol{x'}\|_0$.
Hence,
\[
\Tr_{\mathcal{I}} \rho(\boldsymbol{x}) = \Tr_{\mathcal{I}} \rho(\boldsymbol{x'}).
\]

\paragraph{Hamiltonian encoding.}
Let $\boldsymbol{x}=(x_1,\dots,x_N)\in\mathbb{R}^{N}$ be a vector. Following \citet{berberich2023robustness}, consider the following parameterised  quantum circuit
\begin{align}\label{eq:qc}
    \ket{{\psi}(\boldsymbol{x})}={U}_1(x_1)\cdots {U}_N(x_N)\ket{\psi_0},
\end{align}
consisting of $N$ parametric unitary operators ${U}_i(x_i)\in\mathcal{U}_{n}$ acting on the initial state $\ket{\psi_0}$. Let $\rho(\boldsymbol{x}):=\ket{\psi(\boldsymbol{x})}\bra{\psi(\boldsymbol{x})}$.
These unitaries can also be written as ${U}_j(x_j)=e^{-i x_jH_j}$, where the Hamiltonian $H_i=H_i^\dagger$ generates the gate ${U}_i$.
The following result shows that quantum circuits are robust to slight perturbation of the classical parameters.
\begin{lemma}[Adapted from Theorem 2.2, \cite{berberich2023robustness}]
\label{lem:control}
Let $\boldsymbol{x},\boldsymbol{x'} \in \mathbb{R}^N$. $U(\theta) = e^{- i \theta H }$. For any initial state $\ket{\psi_0}$ we have
\[
\|\ket{\psi(\boldsymbol{x})}\bra{\psi(\boldsymbol{x})}  - \ket{\psi(\boldsymbol{x'})}\bra{\psi(\boldsymbol{x'})}  \|_{2} \leq \sum_{i=1}^N \|H_i\|_2 |x_i - x_i'| \leq \|\boldsymbol{x} - \boldsymbol{x'}\|_1 \max_i\|H_i\|_2.
\]
\end{lemma}
Remark also that for $\rho,\sigma$ pure states we have $\|\rho -\sigma\|_1 = \sqrt{2} \|\rho-\sigma\|_2$ and for any vectors $x,x'\in \mathbb{R}^N$ we have $\|\boldsymbol{x}-\boldsymbol{x'}\|_1 \leq \sqrt{N}\|\boldsymbol{x}-\boldsymbol{x'}\|_2$. Then we have:
\begin{align}
  \frac{1}{2} \|\rho(\boldsymbol{x}) -\rho(\boldsymbol{x'}) \|_{1} \leq \sqrt{\frac{1}{2}}\|\boldsymbol{x} - \boldsymbol{x'}\|_1\max_i\|H_i\|_2 
   \\ \label{eq:l2}\leq \sqrt{\frac{N}{2}}\|\boldsymbol{x} - \boldsymbol{x'}\|_2\max_i\|H_i\|_2 .
\end{align}
It's easy to see that the circuits $U(\boldsymbol{x})$ and $U(\boldsymbol{x'})$ coincides excepts for  $\|\boldsymbol{x}-\boldsymbol{x'}\|_0$ gates. In order to investigate the local structure of the output, we need to introduce some assumptions on the circuit architecture. For instance, assuming that the circuit has 1-dimensional connectivity and depth $L$, there exists $\mathcal{I}\subseteq [n]$, $|\mathcal{I}|\leq 2 L \|\boldsymbol{x}-\boldsymbol{x'}\|_0$, such that
\[
\Tr_{\mathcal{I}} \rho(\boldsymbol{x}) = \Tr_{\mathcal{I}} \rho(\boldsymbol{x'}).
\]

\subsection{Noisy encodings}

A case of interest is when the circuit $U(\boldsymbol{x})$ is interspersed of $L$ layers of local Pauli noise $\mathcal{P}_q$. Let $\mathcal{C}_{\boldsymbol{x}}$ be the channel describing the composition of unitaries and noise:
\[
\mathcal{C}_{\boldsymbol{x}}(\rho_0) = \mathcal{P}^{\otimes n}_q \circ U_N(x_N)(\cdot)U_N(x_N)^\dag \circ \mathcal{P}^{\otimes n}_q  \circ\dots \circ \mathcal{P}^{\otimes n}_q \circ U_1(x_1)(\rho_0)U_1^\dag(x_1)
\]
Then by \lemref{lem:nibp}, we get:
\[
D_2(\mathcal{C}_{\boldsymbol{x}}(\rho_0)\|\mathbb{1}/2^n) \leq q^{2L}n. 
\]
and by Pinsker's inequality,
\begin{equation}
\label{eq:pinsker}
    \frac{1}{2}\bigg\|\mathcal{C}_{\boldsymbol{x}}(\rho_0) - \frac{\mathbb{1}}{2^n}\bigg\|_{1} \leq\sqrt{\frac{q^{2L}n}{2}}.
\end{equation}
Alternatively, by the quantum Bretagnolle-Huber inequality (\lemref{lem:qbh}),
\[
\frac{1}{2}\bigg\|\mathcal{C}_{\boldsymbol{x}}(\rho_0) - \frac{\mathbb{1}}{2^n}\bigg\|_{1}\leq \sqrt{1-\exp{(-q^{2L}n)}}.
\]
And by the triangle inequality
\[
\frac{1}{2}\|\mathcal{C}_{\boldsymbol{x}}(\rho_0) - \mathcal{C}_{{\boldsymbol{x'}}}(\rho_0)\|_1 \leq 2 \min\left\{\sqrt{\frac{q^{2L}n}{2}},\sqrt{1-\exp{(-q^{2L}n)}}\right\}.
\]

\paragraph{High noise regime}
Now, assume that $\rho(\cdot)$ is an encoding post-processed by a channel $\mathcal{A}$, consisting in $L$ layers such that each of them has light-cone $\mathcal{I}$ and its followed by local depolarising noise with noise parameter $p$.
If $p$ satisfies $2|\mathcal{I}|(1 - p) < 1$, we have from \lemref{lem:high-noise},

\begin{align*}
   \frac{1}{2}\|\mathcal{A}(\rho(\boldsymbol{x}))-\mathcal{A}(\rho(\boldsymbol{x'}))\|_1 
   \\\leq (2|\mathcal{I}|(1-p))^L W_1(\rho(\boldsymbol{x}),\rho(\boldsymbol{x'})) 
\end{align*}
For $\rho(\boldsymbol{x})\overset{(\Xi,\tau)}{\sim}\rho(\boldsymbol{x'})$ we have
\[
\frac{1}{2}\|\rho(\boldsymbol{x})-\rho(\boldsymbol{x'})\|_1\leq W_1(\rho(\boldsymbol{x}),\rho(\boldsymbol{x'})) \leq \min \left\{\max_{\mathcal{I}\in\Xi}|\mathcal{I}| \frac{3}{2}\tau, n\tau\right\}.
\]

\section{Private quantum-inspired sampling}
\label{app:inspired}
Our argument is similar to the one of (Problem 1.b, \cite{ullman}) for uniform subsampling, but we include the complete proof here for clarity.
Given a normalised vector $x=(x_1,\dots,x_n)\in \mathbb{C}^n$, let $\ket{x}:=\sum_{i=1}^n x_i\ket{i}$ be the amplitude encoding defined in the previous section.
\begin{theorem}[DP amplification by quantum-inspired sampling]
\label{thm:quantum-inspired}
For any $x\in \mathbb{C}^n$, let $s=(s_1,\dots, s_m)$ be the measurement outcomes in the computational basis of $\ket{x}^{\otimes m}$. Denote $\mathcal{S}$ as the sampling mechanism that maps $x$ into $s$.
Let $\mathcal{A}$ be a $(\epsilon,\delta)$-DP algorithm that takes only $s$ as input.  Then $\mathcal{A'}=\mathcal{A}\circ \mathcal{S}$ is $(\epsilon',\delta')$-DP, with $\epsilon' = \log(1+ (e^\epsilon-1)m (\alpha+\beta))$ and $\delta'= \delta m (\alpha+\beta)$.
\end{theorem}
\begin{proof}
We will use $T \subseteq \{1,\dots,n\}$ to denote the identities of the $m$-subsampled elements $s_1,\dots,s_m$ (i.e. their index, not their actual value). Note that $T$ is a random variable and that the randomness of $\mathcal{A'}:=\mathcal{A}\circ \mathcal{S}$ includes both the randomness of the sample $T$ and the random coins of $\mathcal{A}$. Let $x\sim x'$ be adjacent datasets and assume that $x$ and $x'$ differ only on some row $t$. Let $s$ (or $s'$) be a subsample from $x$ (or $x'$) containing the rows in $T$ . Let $F$ be an arbitrary subset of the range of $\mathcal{A}$). For convenience, define $p= (\alpha+\beta) m$.
Note that, by definition of quantum amplitude encoding and by union bound,
\begin{align*}
  \Pr[i \in T] \leq m \times \Pr[\ket{x}\text{ collapses to state }\ket{i}] \leq m (\alpha+\beta):=p
\end{align*}
To show $(\log(1+ p(e^\epsilon-1)), p\delta)$-DP, we have to bound the ratio
\begin{align*}
    \frac{\Pr[\mathcal{A}'(x)\in F]-p\delta}{\Pr[\mathcal{A}'(x')\in F]}\leq
     \frac{p\Pr[\mathcal{A}(s)\in F|i \in T]+ (1-p)\Pr[\mathcal{A}(s)\in F|i \not\in T]-p\delta}{p\Pr[\mathcal{A}(s')\in F|i \in T]+ (1-p)\Pr[\mathcal{A}(s')\in F|i \not\in T]}
\end{align*}
by $p(1+ (e^\epsilon-1))$. For simplicity, define the quantities
\begin{align*}
    C = \Pr[\mathcal{A}(s)\in F|i \in T]
    \\ C' = \Pr[\mathcal{A}(s')\in F|i \in T]
    \\ D = \Pr[\mathcal{A}(s)\in F|i \not\in T] = \Pr[\mathcal{A}(s')\in F|i \not\in T].
\end{align*}
We can rewrite the ratio as 
\begin{align*}
    \frac{\Pr[\mathcal{A}'(x)\in F]-p\delta}{\Pr[\mathcal{A}'(x')\in F]} = \frac{pC + (1-p)D - p\delta}{p C' + (1-p)D}.
\end{align*}
Now we use the fact that, by $(\epsilon, \delta)$-DP, $C \leq \min\{C',D\} + \delta$. Plugging all together, we get 
\begin{align*}
    pC + (1 - p)D - p\delta 
    \leq p(e^\epsilon \min\{C',D\}) + (1 - p)D\\
    \leq p(\min\{C',D\} + (e^\epsilon-1) \min\{C',D\}) 
    + (1 - p)D \\
    \leq p(C' + (e^\epsilon-1) (pC' + (1 -p)D)) 
    + (1 - p)D \\
    \leq (pC' + (1 - p)D) + p(e^\epsilon-1))(pC' + (1 - p)D)
    \leq (1+ p(e^\epsilon-1))(pC' + (1 - p)D),
\end{align*}
where the third-to-last line follow from $\min\{x,y\} \leq \alpha x + (1 - \alpha)y$ for every $0 \leq \alpha \leq 1$. 
To conclude the proof, we rewrite the ratio and get the desired bound.
\[
\frac{\Pr[\mathcal{A}'(x)\in F]-p\delta}{\Pr[\mathcal{A}'(x')\in F]}\leq 1 + p(e^\epsilon -1).
\]
\end{proof}

\end{document}